\documentclass[11pt]{article}
\usepackage{amsfonts,amsthm,amssymb}
\usepackage[pdftex]{graphicx}
\usepackage{color}
\usepackage[fleqn]{amsmath}
\usepackage{enumitem}
\usepackage[leftcaption]{sidecap}
\newcommand{\be}{\begin{eqnarray}}
\newcommand{\ee}{\end{eqnarray}}
\newcommand{\bez}{\begin{eqnarray*}}
\newcommand{\eez}{\end{eqnarray*}}
\renewcommand{\d}{\mathrm{d}}
\newcommand{\bd}{\bar{\mathrm{d}}}
\newcommand{\cA}{\mathcal{A}}
\newcommand{\bbN}{\mathbb{N}}
\newcommand{\bsy}{\boldsymbol}

\theoremstyle{plain}
\newtheorem{theorem}{Theorem}[section]
\newtheorem{lemma}[theorem]{Lemma}
\newtheorem{proposition}[theorem]{Proposition}

\theoremstyle{definition}

\newtheorem{remark}[theorem]{Remark}
\newtheorem{example}[theorem]{Example}

\numberwithin{equation}{section}
\numberwithin{theorem}{section}

\allowdisplaybreaks

\renewcommand{\theequation} {\arabic{section}.\arabic{equation}}

\begin{document}

\title{\bf Tropical limit of matrix solitons \\ and entwining Yang-Baxter maps}

\author{ \sc
 Aristophanes Dimakis$^1$ and Folkert M\"uller-Hoissen$^{2,3}$\footnote{Corresponding author} \\ \small
 $^1$ Dept. of Financial and Management Engineering, University of the Aegean, Chios, Greece \\  
 \small e-mail: dimakis@aegean.gr \\
 \small
 $^2$ Max Planck Institute for Dynamics and Self-Organization, G\"ottingen, Germany \\ 
      \small e-mail: folkert.mueller-hoissen@ds.mpg.de \\
      \small
 $^3$ Institute for Theoretical Physics, Georg August University,
      37077 G\"ottingen, Germany \\
      \small e-mail: folkert.mueller-hoissen@theorie.physik.uni-goettingen.de
}

\date{ }

\maketitle

\begin{abstract}
We consider a matrix refactorization problem, i.e., a ``Lax representation'', for the Yang-Baxter map that originated 
as the map of polarizations from the ``pure'' 2-soliton solution of a matrix KP equation. Using the Lax matrix 
and its inverse, a related refactorization problem determines another map, 
which is not a solution of the Yang-Baxter equation, but satisfies a mixed version of the Yang-Baxter equation 
together with the Yang-Baxter map. Such maps have been called ``entwining Yang-Baxter maps'' in recent work.
In fact, the map of polarizations obtained from a pure 2-soliton solution of a matrix KP equation, and already 
for the matrix KdV reduction, is \emph{not} in general a Yang-Baxter map, but it is described by one of the 
two maps or their inverses. We clarify why the weaker version of the Yang-Baxter equation holds, by exploring 
the pure 3-soliton solution in the ``tropical limit'', where the 3-soliton interaction decomposes into 2-soliton 
interactions. Here this is elaborated for pure soliton solutions, generated via a binary Darboux transformation, 
of matrix generalizations of the two-dimensional Toda lattice equation, where we meet the same entwining Yang-Baxter 
maps as in the KP case, indicating a kind of universality. 
\end{abstract}

\section{Introduction}
\label{sec:intro}
The quantum Yang-Baxter equation is known to be a crucial structure underlying two-dimensional integrable 
QFT models. A typical feature of the latter is the factorization of the scattering matrix into 
contributions from 2-particle interactions (see \cite{Bomb16} and references therein).  
This factorization is also typical for the scattering of solitons of classical nonlinear integrable field equations. 
Indeed, we tend to think of a multi-soliton solution of some vector or matrix version of an integrable nonlinear 
partial differential of difference equation as being composed of 2-soliton interactions. Matrix solitons carry 
``internal degrees of freedom'', called ``polarization''. In some cases, like matrix KdV \cite{Gonc+Vese04} 
or vector NLS \cite{APT04IP,Tsuc04,Tsuc15}, the map from incoming to outgoing matrix data, i.e., polarizations,
 has been found to satisfy the Yang-Baxter equation. But why should we expect the Yang-Baxter property?
The latter is a statement about three particles, here solitons, and may be thought of as expressing independence 
of the different ways in which a three-particle interaction can be decomposed into 2-particle interactions. 
First of all, how to decompose a 3-soliton solution into 2-soliton interactions? Because of the wave nature 
of solitons, there are no definite events at which the interaction takes place, and (with the exception of 
asymptotically incoming and outgoing solitons) 
there are no definite values of the dependent variable which could define a corresponding map.
However, there is a certain limit, called ``tropical limit'', that takes soliton waves to ``point particles'' and then 
indeed determines events at which an interaction occurs. This has been a decisive tool in our previous work
\cite{DMH11KPT,DMH12KPBT,DMH14KdV,DMH19LMP,DMH18p,DMHC19} and this will be so also in this work, 
which continues an exploration started in \cite{DMH19LMP}, also see \cite{DMH18p,DMHC19}. It will indeed 
lead us to a deeper understanding of the question concerning the Yang-Baxter property raised above and to 
a revision of the previous picture.
\vspace{.2cm}

Let $K$ be a constant $n \times m$ matrix of maximal rank.\footnote{In this work, we only consider matrices over the real 
or complex numbers.} 
In \cite{DMH19LMP}, we explored a matrix version of the potential 
KP equation, the \emph{pKP$_K$ equation}
\bez
 \left( 4 \phi_{t} - \phi_{xxx}  - 6 (\phi_x K \phi_x) \right)_x - 3 \phi_{yy} + 6 (\phi_x K \phi_y - \phi_y K \phi_x) = 0 \, , 
\eez
from which the KP$_K$ equation is obtained via $u = 2 \, \phi_x$. With the restriction to a subclass of solutions, 
which we called ``pure solitons'' in \cite{DMH19LMP}, the 2-soliton solution, generated by a binary Darboux transformation 
with trivial seed solution, determines a realization of the following Yang-Baxter map.
\vspace{.2cm}

Let $\boldsymbol{S}$ be the set of rank one $m \times n$ $K$-projection matrices ($X K X = X$), and
\bez
 \mathcal{R}(1,2) := \mathcal{R}(p_1,q_1;p_2,q_2) \, : \quad 
     && \hspace{.35cm} \boldsymbol{S} \times \boldsymbol{S} \rightarrow \boldsymbol{S} \times \boldsymbol{S} \\
     && (X_1, X_2) \mapsto (X_1' , X_2') 
\eez
be given by
\be
  && X_1' = \alpha_{12} \, \Big( 1_m - \frac{p_2-q_2}{p_2-p_1} X_2 K \Big) \, X_1 \, 
                    \Big( 1_n - \frac{p_2-q_2}{q_1-q_2} K X_2 \Big) \, ,  \nonumber \\
  && X_2' = \alpha_{12} \, \Big( 1_m - \frac{p_1-q_1}{q_2-q_1} X_1 K \Big) \, X_2 \, 
                    \Big( 1_n - \frac{p_1-q_1}{p_1-p_2} K X_1 \Big) \, ,    \label{KP_YB_map}               
\ee
where $1_m$ denotes the $m \times m$ identity matrix and
\be
     \alpha_{12} := \alpha(p_1,q_1,X_1;p_2,q_2,X_2)
   := \Big( 1 - \frac{(p_1-q_1)(p_2-q_2)}{(p_2-p_1)(q_2-q_1)} \mathrm{tr}(K X_1 K X_2) \Big)^{-1} = \alpha_{21} \, .
       \label{alpha_12}
\ee
This is a parameter-dependent Yang-Baxter map, which means that it satisfies the Yang-Baxter equation 
\be
    \mathcal{R}_{\bsy{12}}(1,2) \circ \mathcal{R}_{\bsy{13}}(1,3) \circ \mathcal{R}_{\bsy{23}}(2,3) 
    = \mathcal{R}_{\bsy{23}}(2,3) \circ \mathcal{R}_{\bsy{13}}(1,3) \circ \mathcal{R}_{\bsy{12}}(1,2)   \label{YB_eq}
\ee
on $\boldsymbol{S} \times \boldsymbol{S} \times \boldsymbol{S}$. The indices of $\mathcal{R}_{\bsy{ij}}$ 
specify on which two of the three factors the map $\mathcal{R}$ acts. 
Here we have to assume that the constants $p_i$, $i=1,2,3$, and also $q_i$, $i=1,2,3$, are pairwise distinct, 
and that the expressions for $\alpha_{ij}$ make sense. 
\vspace{.2cm}

If $q_i =-p_i$, this is the Yang-Baxter map obtained from the 2-soliton solution of the KdV$_K$ equation
\be
    4 u_t - u_{xxx} - 3 (u K u)_x = 0 \, .   \label{KdV_K}
\ee
 For the matrix KdV equation (where $m=n$ and $K=1_n$), the Yang-Baxter map has first been derived in 
 \cite{Gonc+Vese04}.\footnote{The factor $\alpha_{12}$ is missing in the latter work, but it is necessary for 
 the Yang-Baxter property.}  
\vspace{.2cm}

In the particular case where $n=1$ (and correspondingly for $m=1$), the above (generically highly nonlinear) 
Yang-Baxter map becomes \emph{linear}:
\be
      (X'_1, X'_2) = (X_1 , X_2) \, R(i,j) \, , \quad
      R(i,j) := \left(\begin{array}{cc} \frac{p_i-p_j}{p_i-q_j} & \frac{p_i-q_i}{p_i-q_j} \\
                                       \frac{p_j-q_j}{p_i-q_j} & \frac{q_i-q_j}{p_i-q_j} \end{array}\right) \, .
                                         \label{R-matrix}
\ee
Here we used the fact that, for $X \in \boldsymbol{S}$, $K X$ is now a scalar, so that the $K$-projection 
property requires $K X=1$. The R-matrix, which ermerges here, solves the Yang-Baxter equation  
on a threefold direct sum of an $m$-dimensional vector space, which extends the set $\boldsymbol{S}$.
\vspace{.2cm}

In the tropical limit of a pure $N$-soliton solution of the KP$_K$ equation, the dependent variable $u$ 
has support on a piecewise linear structure in $\mathbb{R}^3$ (with coordinates $x,y,t$), a configuration of pieces of planes, 
and the dependent variable takes a constant value on each plane segment. This piecewise linear structure is obtained as 
the boundary of ``dominating phase regions''. In the KdV reduction, the support of the dependent variable in the 
tropical limit is a piecewise linear graph in 2-dimensional space-time. For the KdV 2-soliton solution, we have 
four dominating phase regions, numbered by $11$, $12$, $21$, respectively 
$22$.\footnote{The parameters $p_k,q_k$ belong to the $k$-th soliton. The first digit of the phase number $ab$ refers to 
soliton $1$, the second to soliton $2$. Since we write $p_k=:p_{k,1}$ and $q_k =: p_{k,2}$, we have $a,b \in \{1,2\}$. }
Fig.~\ref{fig:KdV_2s_21} shows an example.
\begin{figure}[h] 
\begin{center}
\includegraphics[scale=.2]{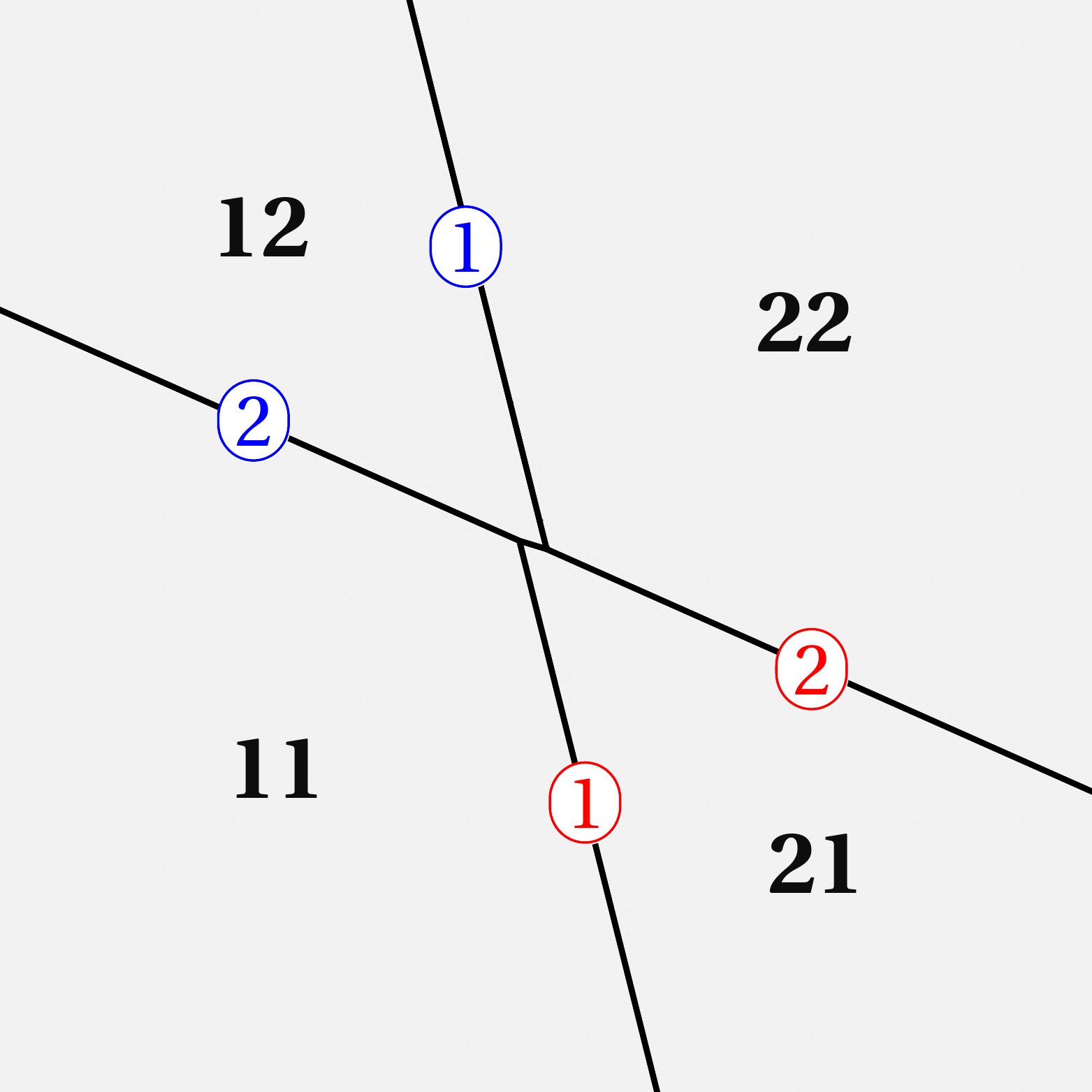} 
\parbox{15cm}{
\caption{Dominating phase regions and tropical limit graph in 2-dimensional space-time, for a 2-soliton solution 
of the KdV$_K$ equation. Here time $t$ is the vertical coordinate.
\label{fig:KdV_2s_21} }
}
\end{center}
\end{figure}
Using the general 2-soliton solution to compute the values of the dependent variable $u$ along the boundary line segments 
of the tropical limit graph, after normalization to $\hat{u}$ such that $\mathrm{tr}(K \hat{u})=1$, the map 
$(\hat{u}_{11,21}, \hat{u}_{21,22}) \mapsto (\hat{u}_{12,22}, \hat{u}_{11,12})$ yields the above Yang-Baxter map (with 
the KdV reduction $q_i =-p_i$). Here, e.g., $\hat{u}_{11,21}$ is the polarization along the boundary line between 
the dominating phase regions numbered by $11$ and $21$.  
For a rank one matrix, the above normalization condition is equivalent to the $K$-projection property.
\vspace{.2cm}

But what about a phase constellation different from the one shown in Fig.~\ref{fig:KdV_2s_21} ? Indeed, 
Fig.~\ref{fig:KdV_2s_other} displays alternatives.
\begin{figure}[h] 
\begin{center}
\includegraphics[scale=.2]{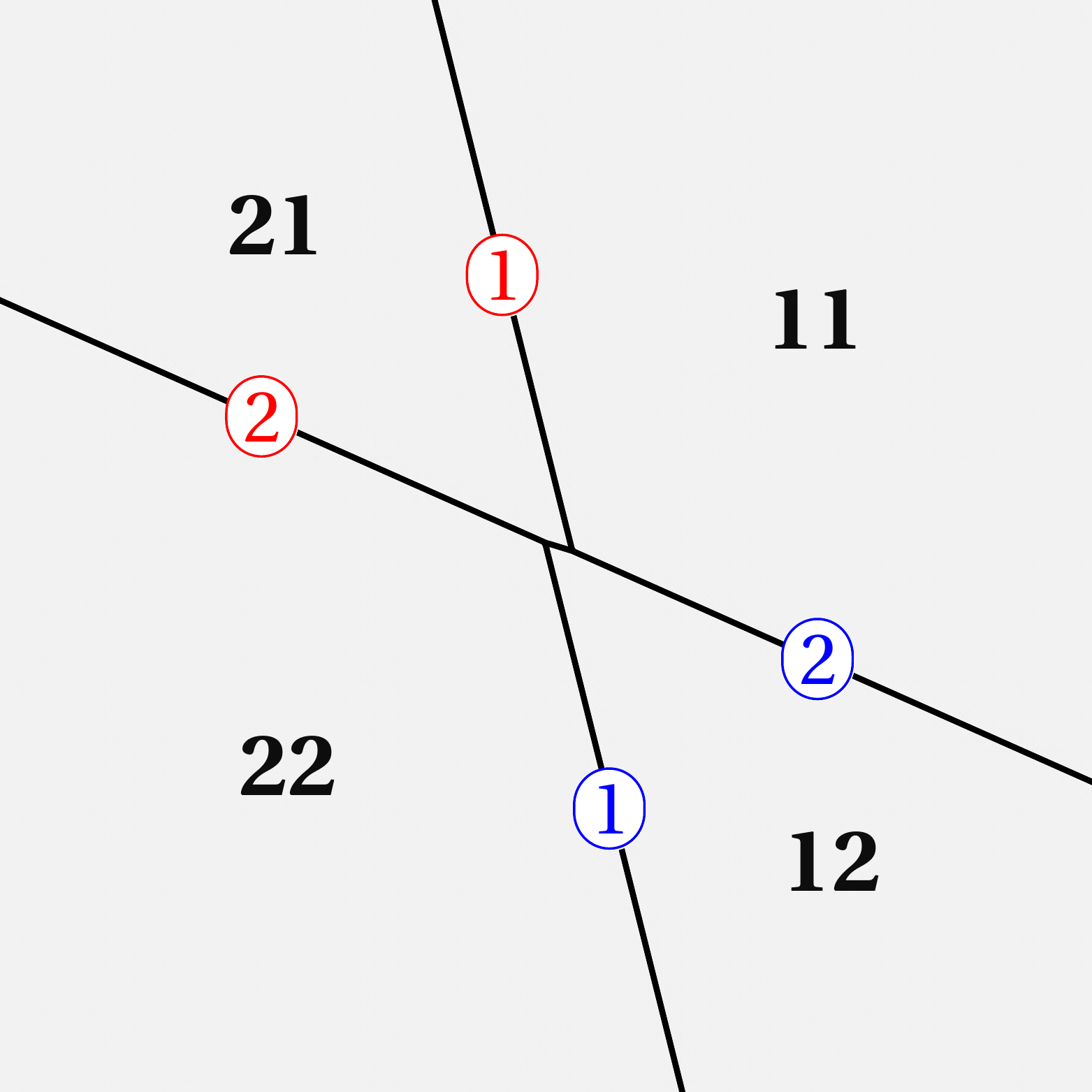}
\hspace{1cm}
\includegraphics[scale=.2]{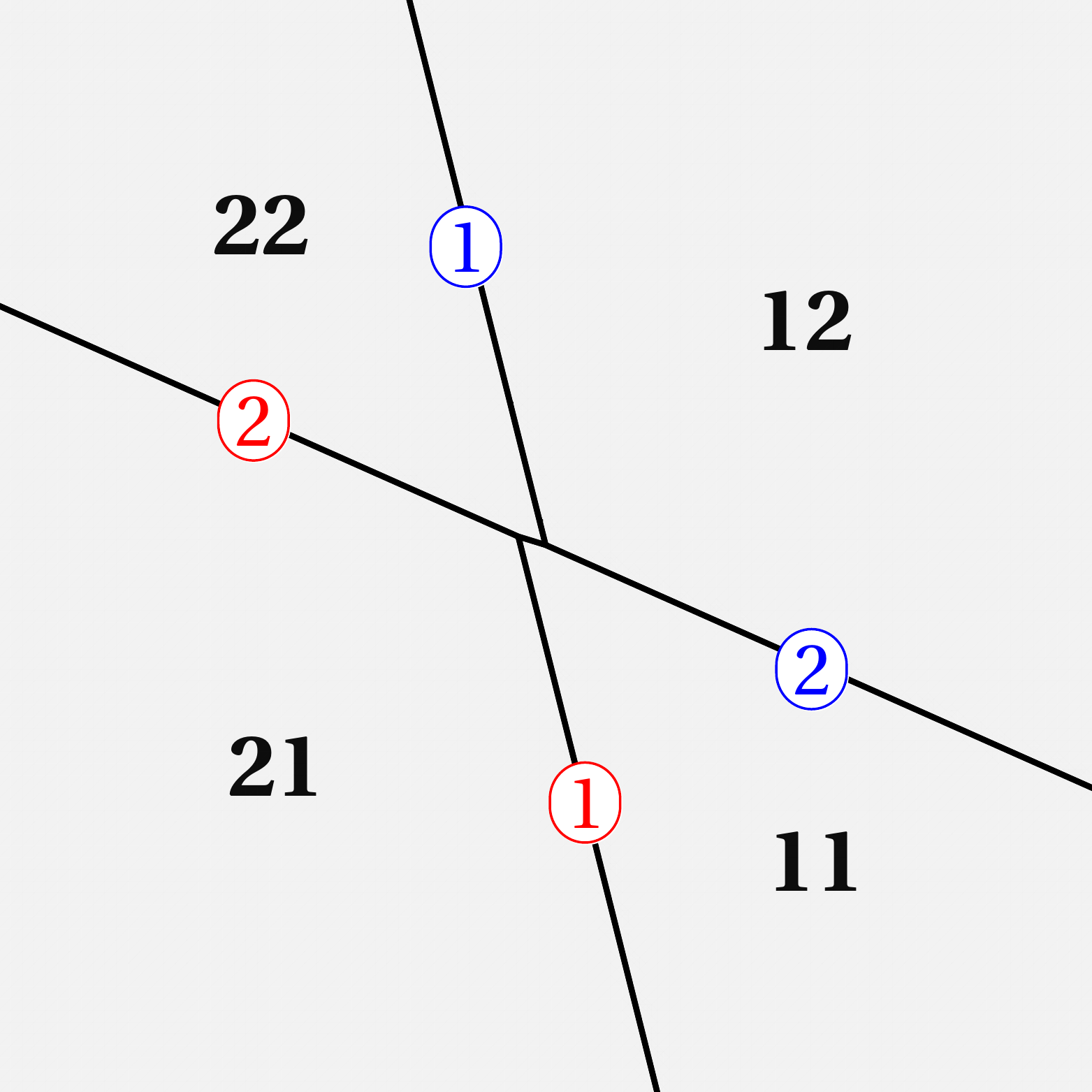}
\hspace{1cm}
\includegraphics[scale=.2]{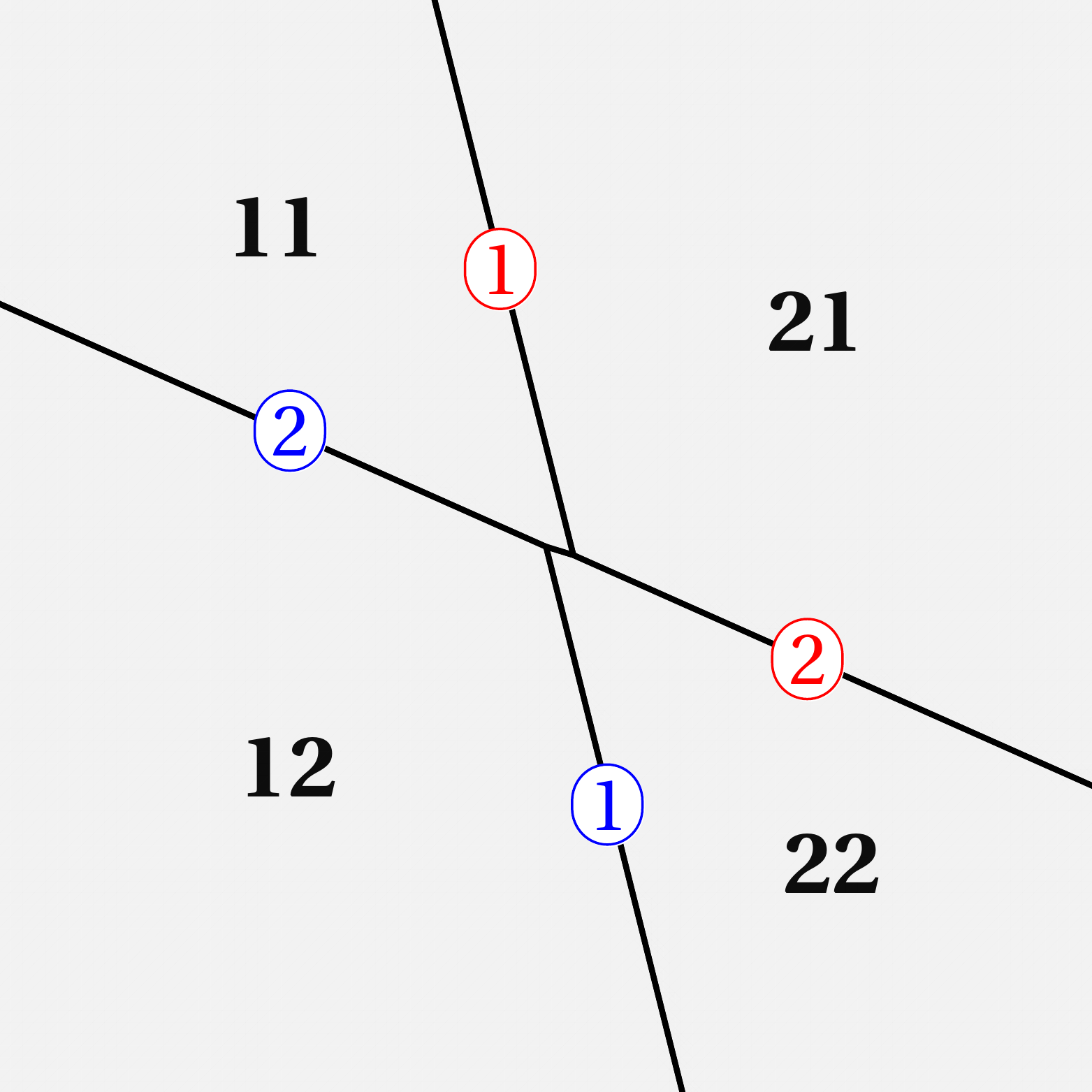}
\parbox{15cm}{
\caption{Other dominating phase region constellations and tropical limit graphs in 2-dimensional space-time, for  
2-soliton solutions of the KdV$_K$ equation. In each case, the Yang-Baxter map $\mathcal{R}$ is recovered if we 
consider the graph as defining a map of polarizations not from bottom to top, but in a different space-time direction
(here from the soliton lines marked by red numbers to those marked by blue numbers).
\label{fig:KdV_2s_other} }
}
\end{center}
\end{figure}
In the same way as for the phase constellation in Fig.~\ref{fig:KdV_2s_21}, one finds that the first alternative in 
Fig.~\ref{fig:KdV_2s_other} leads to the inverse of the above KdV Yang-Baxter map. The remaining possibilities, however,
determine maps that are \emph{not} Yang-Baxter. Nevertheless, they are realized in matrix KdV 2-soliton interactions (see 
Appendix~\ref{app:KdV_2s}), regarding them as a process evolving in time $t$ and  
by choosing the parameters appropriately. We learn that there are matrix KdV 2-soliton solutions for which the map 
of polarizations \emph{in $t$-direction} is \emph{not} Yang-Baxter!\footnote{\cite{Gonc+Vese04} uses results of \cite{Gonc01}, 
where a restriction has been imposed on the parameters of the matrix KdV 2-soliton solution. As a consequence of this, 
the non-Yang-Baxter cases are excluded in \cite{Gonc+Vese04}.} 
How to understand this, in view of our different expectation? 
\vspace{.2cm}

In all cases of phase constellations, shown in 
Fig.~\ref{fig:KdV_2s_other}, the Yang-Baxter map $\mathcal{R}$ is present, however.
It is recovered by regarding the plot not as a process in $t$-direction, but in a different direction 
in space-time.\footnote{As a process in $t$-direction, the first plot in Fig.~\ref{fig:KdV_2s_other} corresponds 
to an application of the inverse of the Yang-Baxter map $\mathcal{R}$.}
\vspace{.2cm}
 
First of all, this means that we overlooked something in our analysis of the KP$_K$ case in \cite{DMH19LMP}. 
As in the KdV reduction, of course also the general pure 2-soliton solution of KP$_K$ contains constellations, 
for certain parameter values, where incoming and outgoing polarizations (with respect to a chosen direction) are related 
by a map that is \emph{not} a Yang-Baxter map. Besides the above Yang-Baxter map, this map and the inverses of 
both maps are needed to describe the propagation of polarizations along the support of pure 
multi-soliton solutions in the tropical limit. 
\vspace{.2cm}

We were actually led to the new insights by exploring a matrix version of the two-dimensional Toda 
lattice equation (see, e.g., \cite{Mikh79,Taka18} for the scalar equation). 
This is the subject of Section~\ref{sec:Toda}. In particular, with the restriction to ``pure solitons'', 
it turns out that the same Yang-Baxter map is here at work as in the KP$_K$ case, indicating a 
kind of universality. This may not come as a surprise, however, since both equations are known to be related 
(also see Remark~\ref{rem:scalar2DTL} below). 
\vspace{.2cm}

The tropical limit associates with a soliton solution a configuration of plane segments, together with values of the 
dependent variable on the segments.\footnote{Here we think of the discrete independent variable $k$ in the 
Toda lattice equation as being continuously extended (also see \cite{Bion+Wang10}). Such a smoothing of the 
discrete variable is actually done in the plots presented in Section~\ref{sec:Toda} of this work. 
But, of course, it is not assumed in any of our computations.}
It is found that, at intersections, these polarizations are related by one of \emph{two} maps (and their inverses), 
of which only one is a Yang-Baxter map, but the two maps satisfy a mixed version of the Yang-Baxter equation 
(see (\ref{YB_mixed}) below). 
They are ``entwining Yang-Baxter maps'' in the sense of \cite{Koul+Papa11BCP}. 
\vspace{.2cm}

Section~\ref{sec:Lax} presents a ``Lax representation'' for the above map $\mathcal{R}$. This is a matrix refactorization 
problem. The basic argument\footnote{It is actually more generally based on the relation between 
neighboring simplex equations, see \cite{DMH15} and references cited there.} 
is the same as in \cite{Gonc+Vese04} for the matrix KdV case (also see \cite{Suri+Vese03,Resh+Vese05}), 
but we prove more directly, as compared with \cite{Gonc+Vese04}, that the refactorization problem 
determines the map $\mathcal{R}$.  

In Section~\ref{sec:Lax_system}, we show that this refactorization problem, written in a different way,
also determines the abovementioned mixed version of the Yang-Baxter equation.  
It implies further relations which in particular lead to solutions of the ``WXZ system'' 
in \cite{Hlav+Snob99},  called ``Yang-Baxter system'' in \cite{Brze+Nich05}. To our knowledge, such a system 
first appeared in \cite{Vlad93}.

In Section~\ref{sec:Toda} we explore soliton solutions of the abovementioned matrix 2-dimensional 
Toda lattice equation. Section~\ref{subsec:Toda_bDT} presents a binary Darboux transformation for the matrix 
potential two-dimensional Toda lattice equation. Its origin from a general result in bidifferential calculus is 
explained in Appendix~\ref{app:BDT}. 
We then concentrate on the case of vanishing seed solution. Section~\ref{subsec:pure} further restricts to 
a subclass of soliton solutions, which we call ``pure'', and we define the tropical limit of such solitons.  
In Section~\ref{subsec:2sol} we derive the Yang-Baxter map $\mathcal{R}$ from the pure 2-soliton solution.
The relevance of the aforementioned additional non-Yang-Baxter map is explained in Section~\ref{subsec:other_maps}, and 
Section~\ref{subsec:3sol}, which treats the case of three pure solitons, shows explicitly how the Yang-Baxter map 
and the non--Yang-Baxter map (and their inverses) are at work, and why they have to be ``entwining''. 

Finally, Section~\ref{sec:concl} contains some concluding remarks.

\section{A Lax representation for the Yang-Baxter map}
\label{sec:Lax}
Let $K$ be an $n \times m$ matrix with maximal rank, and 
\be
  &&  A_i(\lambda,X) := A(p_i,q_i,\lambda,X) := 1_m - \frac{p_i-q_i}{\lambda -q_i} \, X \, K \, , \quad
    \tilde{A}_i(\lambda,X) := 1_n - \frac{p_i-q_i}{\lambda -q_i} \, K \, X \, , \nonumber \\
  &&  B_i(\lambda,X) := B(p_i,q_i,\lambda,X) := 1_m + \frac{p_i-q_i}{\lambda -p_i} \, X \, K \, , \quad
    \tilde{B}_i(\lambda,X) := 1_n + \frac{p_i-q_i}{\lambda -p_i} \, K \, X \, ,   \label{A,tA,B,tB}
\ee
where $X$ is an $m \times n$ matrix and $\lambda$ a parameter. Then we have
\bez
    K A_i = \tilde{A}_i K \, , \qquad  K B_i = \tilde{B}_i K \, .
\eez     
If $X$ is a $K$-projection matrix, which means $XKX=X$, then 
\bez
    B_i = A_i^{-1} \, , \qquad \tilde{B}_i = \tilde{A}_i^{-1} \, ,  
\eez
if $\lambda \notin \{q_i,p_i\}$.

\begin{theorem}
\label{thm:Lax}
Let $p_1,p_2,q_1,q_2$ be pairwise distinct and $X_i$, $i=1,2$, rank one $K$-projections, 
hence $X_i \in \boldsymbol{S}$. Then the 
refactorization equations\footnote{These are local 1-simplex equations, see \cite{DMH15}, for example.}
\be
   && A_1(\lambda,X_1) \, A_2(\lambda,X_2) = A_2(\lambda,X_2') \, A_1(\lambda,X_1') \, , 
      \nonumber \\
   && \tilde{A}_1(\lambda,X_1) \, \tilde{A}_2(\lambda,X_2) 
       = \tilde{A}_2(\lambda,X_2') \, \tilde{A}_1(\lambda,X_1')   \label{Lax_eqs}
\ee
imply the map $\mathcal{R}(1,2)$, defined in the introduction (see (\ref{KP_YB_map})).
\end{theorem}

A proof is given in Appendix~\ref{app:unique_map}. 
Recalling a well-known argument (see \cite{DMH15} and references cited there), exploiting associativity in 
different ways, we obtain
\bez
    A_1(\lambda,X_1) \, A_2(\lambda,X_2) \, A_3(\lambda,X_3)
    &\stackrel{\mathcal{R}_{\bsy{12}}}{=}& A_2(\lambda,Y_2) \, A_1(\lambda,Y_1)  \, A_3(\lambda,X_3) \\
    &\stackrel{\mathcal{R}_{\bsy{13}}}{=}& A_2(\lambda,Y_2)  \, A_3(\lambda,Y_3) \, A_1(\lambda,Z_1)  \\
    &\stackrel{\mathcal{R}_{\bsy{23}}}{=}& A_3(\lambda,Z_3) \, A_2(\lambda,Z_2)  \, A_1(\lambda,Z_1)  \, ,
\eez
where we abbreviated $\mathcal{R}_{\bsy{ij}}(i,j)$ to $\mathcal{R}_{\bsy{ij}}$ and set, for example, 
$\mathcal{R}(1,3)(Y_1,X_3) =: (Z_1,Y_3)$, and also 
\bez
    A_1(\lambda,X_1) \, A_2(\lambda,X_2) \, A_3(\lambda,X_3)
    &\stackrel{\mathcal{R}_{\bsy{23}}}{=}& A_1(\lambda,X_1) \, A_3(\lambda,Y_3') \, A_2(\lambda,Y_2')   \\
    &\stackrel{\mathcal{R}_{\bsy{13}}}{=}& A_3(\lambda,Z_3') \, A_1(\lambda,Y_1') \, A_2(\lambda,Y_2')   \\
    &\stackrel{\mathcal{R}_{\bsy{12}}}{=}& A_3(\lambda,Z_3') \, A_2(\lambda,Z_2')  \,  A_1(\lambda,Z_1')   \, . 
\eez
There are corresponding chains with $A_i$ replaced by $\tilde{A}_i$.
If
\bez
  &&  A_1(\lambda,X_1) \, A_2(\lambda,X_2) \, A_3(\lambda,X_3) = A_3(\lambda,Z_3) \, A_2(\lambda,Z_2)  \, A_1(\lambda,Z_1) \, , \\
  && \tilde{A}_1(\lambda,X_1) \, \tilde{A}_2(\lambda,X_2) \, \tilde{A}_3(\lambda,X_3) = \tilde{A}_3(\lambda,Z_3) 
     \, \tilde{A}_2(\lambda,Z_2)  \, \tilde{A}_1(\lambda,Z_1) 
\eez
determines a unique map $(X_1,X_2,X_3) \mapsto (Z_1, Z_2, Z_3)$, which means that\footnote{Also see 
Proposition~3.1 in \cite{Koul+Papa11BCP}.} 
\bez
  \left. \begin{array}{l} A_3(\lambda,Z_3) \, A_2(\lambda,Z_2) \, A_1(\lambda,Z_1) 
             = A_3(\lambda,Z'_3) \, A_2(\lambda,Z'_2) \, A_1(\lambda,Z'_1)  \\
         \tilde{A}_3(\lambda,Z_3) \, \tilde{A}_2(\lambda,Z_2) \, \tilde{A}_1(\lambda,Z_1) 
             = \tilde{A}_3(\lambda,Z'_3) \, \tilde{A}_2(\lambda,Z'_2) \, \tilde{A}_1(\lambda,Z'_1)   
             \end{array} \right\} 
    \quad \Rightarrow \quad Z_i'=Z_i \, , \, i=1,2,3 \, , \quad  
\eez
we can conclude the statement of the following theorem. But it can also be verified directly, 
using computer algebra.

\begin{theorem}
Let $X_i \in \boldsymbol{S}$. Then $\mathcal{R}$, given by (\ref{KP_YB_map}), is a Yang-Baxter map.
\end{theorem}

(\ref{Lax_eqs}) is called a ``Lax representation'' for the map $\mathcal{R}$.
\vspace{.2cm}

Using (\ref{A,tA,B,tB}), (\ref{KP_YB_map}) can be expressed as
\be
   X_1' = \frac{B_2(p_1,X_2) \, X_1 \, \tilde{A}_2(q_1,X_2)}
                      {\mathrm{tr}[B_2(p_1,X_2) \, X_1 \, \tilde{A}_2(q_1,X_2) \, K]} \, , \quad
   X_2' = \frac{A_1(q_2,X_1) \, X_2 \, \tilde{B}_1(p_2,X_1)}
                      {\mathrm{tr}[A_1(q_2,X_1) \, X_2 \, \tilde{B}_1(p_2,X_1) \, K]} \, .
                      \label{X_1,2,out}
\ee
In particular, we have
\bez
      \alpha_{12}^{-1} = \mathrm{tr}[B_2(p_1,X_2) \, X_1 \, \tilde{A}_2(q_1,X_2) \, K]
                  = \mathrm{tr}[A_1(q_2,X_1) \, X_2 \, \tilde{B}_1(p_2,X_1) \, K] \, . 
\eez

\begin{remark}
We also have
\bez
      \alpha_{12}^{-1} &=& 1 - \frac{(p_1-q_1)(p_2-q_2)}{(p_1-p_2)(q_1-q_2)} \mathrm{tr}(X_1' K X_2' K) \\
                  &=& \mathrm{tr}[B_2(p_1,X_2') \, X_1' \, \tilde{A}_2(q_1,X_2') \, K]  
                  = \mathrm{tr}[A_1(q_2,X_1') \, X_2' \, \tilde{B}_1(p_2,X_1') \, K] \, , 
\eez
which in particular means that $\alpha_{12}$ is an invariant of the map $\mathcal{R}$.
\end{remark}

\section{Further aspects of the Lax representation}
\label{sec:Lax_system}
According to Section~\ref{sec:Lax}, 
\be
       A_1(\lambda,X_1) \, A_2(\lambda,X_2) = A_2(\lambda,X_2') \, A_1(\lambda,X_1')    \label{Lax_R}
\ee
is a Lax representation for the Yang-Baxter map $\mathcal{R}$. More precisely, 
we have to supplement this equation by $\tilde{A}_1(\lambda,X_1) \, \tilde{A}_2(\lambda,X_2) 
= \tilde{A}_2(\lambda,X_2') \, \tilde{A}_1(\lambda,X_1')$, if $K$ is not an invertible square matrix.
A corresponding extension is also necessary for the other versions of (\ref{Lax_R}) considered below, but 
for simplicity we will suppress it. 

\paragraph{1.}
A Lax representation for the inverse of $\mathcal{R}$ is given by
\be
   B_1(\lambda,X_1) \, B_2(\lambda,X_2) = B_2(\lambda,X_2') \, B_1(\lambda,X_1') \, .  
    \label{Lax_R^-1}
\ee
As a consequence, we have
\bez
    \mathcal{R}(1,2)^{-1} : \; (X_1,X_2) \mapsto (X_1',X_2') \, , 
\eez    
where
\be
   X_1' = \frac{A_2(q_1,X_2) \, X_1 \, \tilde{B}_2(p_1,X_2)}
                      {\mathrm{tr}[A_2(q_1,X_2) \, X_1 \, \tilde{A}_2(p_1,X_2) \, K]} \, , \quad
   X_2' = \frac{B_1(p_2,X_1) \, X_2 \, \tilde{A}_1(q_2,X_1)}
                      {\mathrm{tr}[B_1(p_2,X_1) \, X_2 \, \tilde{A}_1(q_2,X_1) \, K]} \, .   
                      \label{X_in->X_out_inverse}
\ee
Comparison with (\ref{X_1,2,out}) shows that this is obtained from the latter by 
exchanging the two indices 1 and 2. 
This means that $\mathcal{R}$ is a \emph{reversible} Yang-Baxter map,
\bez
      \mathcal{R}_{\bsy{21}}(2,1) \circ \mathcal{R}_{\bsy{12}}(1,2) = \mathrm{id} \, .     
\eez

\paragraph{2.} Let us write 
\be
    A_1(\lambda,X_1) \, B_2(\lambda,X_2) = B_2(\lambda,X_2') \, A_1(\lambda,X_1') \, .   \label{Lax_T}
\ee
instead of (\ref{Lax_R}). As in Section~\ref{sec:Lax} and Appendix~\ref{app:unique_map}, it can be shown 
that this equation uniquely determines the map
\bez
    \mathcal{T}(1,2) := \mathcal{T}(p_1,q_1;p_2,q_2): 
      && \hspace{.35cm} \boldsymbol{S} \times \boldsymbol{S} \rightarrow \boldsymbol{S} \times \boldsymbol{S} \\
      && (X_1, X_2) \mapsto (X_1' , X_2')  \, ,
\eez 
where
\be
  X_1' &=& \alpha_{12}^{-1} \, \Big( 1_m - \frac{p_2-q_2}{p_1-q_2} X_2 K \Big) \, X_1 \, 
                    \Big( 1_n - \frac{p_2-q_2}{p_2-q_1} K X_2 \Big)  \nonumber \\ 
       &=& \frac{A_2(p_1,X_2) \, X_1 \, \tilde{B}_2(q_1,X_2)} 
                      {\mathrm{tr}[A_2(p_1,X_2) \, X_1 \, \tilde{B}_2(q_1,X_2) \, K]} \, ,  \nonumber \\                      
  X_2' &=& \alpha_{12}^{-1} \, \Big( 1_m - \frac{p_1-q_1}{p_2-q_1} X_1 K \Big) \, X_2 \, 
                    \Big( 1_n - \frac{p_1-q_1}{p_1-q_2} K X_1 \Big) \nonumber \\
       &=& \frac{A_1(p_2,X_1) \, X_2 \, \tilde{B}_1(q_2,X_1)}
                      {\mathrm{tr}[A_1(p_2,X_1) \, X_2 \, \tilde{B}_1(q_2,X_1) \, K]}  \, .     
                    \label{X_in->X_out_T}  
\ee
The denominators of the final expressions are both equal to $\alpha_{12}$. 
This map is invariant under exchange of the two indices 1 and 2, hence
\bez
      \mathcal{T}_{\bsy{21}}(2,1) =  \mathcal{T}_{\bsy{12}}(1,2) \, .
\eez
Although (\ref{X_in->X_out_T}) resembles (\ref{KP_YB_map}), in contrast to the latter 
it does \emph{not} yield a Yang-Baxter map. This can be checked using computer algebra. 
As a consequence of associativity, we have
\bez
    A_1(\lambda,X_1) \, B_2(\lambda,X_2) \, A_3(\lambda,X_3)
    &\stackrel{\mathcal{T}_{\bsy{12}}}{=}& B_2(\lambda,Y_2) \, A_1(\lambda,Y_1)  \, A_3(\lambda,X_3) \\
    &\stackrel{\mathcal{R}_{\bsy{13}}}{=}& B_2(\lambda,Y_2)  \, A_3(\lambda,Y_3) \, A_1(\lambda,Z_1)  \\
    &\stackrel{\mathcal{T}^{-1}_{\bsy{23}}}{=}& A_3(\lambda,Z_3) \, B_2(\lambda,Z_2)  \, A_1(\lambda,Z_1)  \, ,
\eez
and also
\bez
    A_1(\lambda,X_1) \, B_2(\lambda,X_2) \, A_3(\lambda,X_3)
    &\stackrel{\mathcal{T}^{-1}_{\bsy{23}}}{=}&  A_1(\lambda,X_1) \, A_3(\lambda,Y'_3) \, B_2(\lambda,Y'_2)   \\
    &\stackrel{\mathcal{R}_{\bsy{13}}}{=}& A_3(\lambda,Z_3') \,  A_1(\lambda,Y'_1) \, B_2(\lambda,Y'_2)   \\
    &\stackrel{\mathcal{T}_{\bsy{12}}}{=}& A_3(\lambda,Z_3') \, B_2(\lambda,Z_2')  \,  A_1(\lambda,Z_1')   \, ,
\eez
where we used (\ref{Lax_R}) and set, for example, $\mathcal{T}^{-1}(X_2,X_3) =: (Y'_2,Y'_3)$. 
One can argue that this implies $Z_i'=Z_i$, $i=1,2,3$, in which case we can conclude that
\be
      \mathcal{T}^{-1}_{\bsy{23}}(2,3) \circ \mathcal{R}_{\bsy{13}}(1,3) \circ \mathcal{T}_{\bsy{12}}(1,2)
    =  \mathcal{T}_{\bsy{12}}(1,2) \circ \mathcal{R}_{\bsy{13}}(1,3) \circ \mathcal{T}^{-1}_{\bsy{23}}(2,3)  
       \, .    \label{YB_mixed}
\ee
This can also be verified using computer algebra. 
Hence, writing the Lax representation (\ref{Lax_R}) in the form (\ref{Lax_T}), we are led to a ``mixed Yang-Baxter equation''
for the two maps $\mathcal{R}$ and $\mathcal{T}$.  

\paragraph{3.} 
The inverse $\mathcal{T}^{-1}$ of $\mathcal{T}$ is given by $(X_1,X_2) \mapsto (X_1',X_2')$, where
\be
  X_1' = \frac{B_2(q_1,X_2) \, X_1 \, \tilde{A}_2(p_1,X_2)}
                      {\mathrm{tr}[B_2(q_1,X_2) \, X_1 \, \tilde{A}_2(p_1,X_2) \, K]} 
                            \, , \quad   
  X_2' = \frac{B_1(q_2,X_1) \, X_2 \, \tilde{A}_1(p_2,X_1)}
                      {\mathrm{tr}[B_1(q_2,X_1) \, X_2 \, \tilde{A}_1(p_2,X_1) \, K]} \, .
                      \label{X_in->X_out_alt_inverse}
\ee 
A corresponding Lax representation is the version 
\be
    B_1(\lambda,X_1) \, A_2(\lambda,X_2) = A_2(\lambda,X_2') \, B_1(\lambda,X_1')  \label{Lax_T^-1}
\ee
of (\ref{Lax_R}).

\begin{remark}
In the special case where $n=1$, besides $\mathcal{R}$ also $\mathcal{T}$ becomes \emph{linear}:
\bez
     (X'_1,X'_2) = (X_1,X_2) \,  T(1,2) \, , \qquad
     T(i,j) := \left(\begin{array}{cc} \frac{q_i-p_j}{q_i-q_j} & \frac{q_j-p_i}{q_i-q_j} \\
                                       \frac{p_j-q_j}{q_i-q_j} & \frac{p_i-q_j}{q_i-q_j} \end{array}\right) \, .
\eez
It is easily verified that the matrix $T$ does \emph{not} satisfy the Yang-Baxter equation.
\end{remark}

\subsection{Further consequences of the Lax representation}
There are actually further consequences of the fact that (\ref{Lax_R}), (\ref{Lax_R^-1}), (\ref{Lax_T}) and (\ref{Lax_T^-1})
uniquely determine maps. 
\begin{enumerate}
\item The two ways to rewrite $A_1(\lambda,X_1) \, A_2(\lambda,X_2) \, B_3(\lambda,X_3)$ in the form 
$B_3(\lambda,Z_3) \, A_2(\lambda,Z_2)  \, A_1(\lambda,Z_1)$, by using (\ref{Lax_R}) and (\ref{Lax_T}),
allow us to deduce that
\be
        \mathcal{T}_{\bsy{23}}(2,3) \circ \mathcal{T}_{\bsy{13}}(1,3) \circ \mathcal{R}_{\bsy{12}}(1,2)
    =  \mathcal{R}_{\bsy{12}}(1,2) \circ \mathcal{T}_{\bsy{13}}(1,3) \circ \mathcal{T}_{\bsy{23}}(2,3)    
     \, . \label{TTR=RTT}
\ee
\item Rewriting $A_1(\lambda,X_1) \, B_2(\lambda,X_2) \, B_3(\lambda,X_3)$ as 
$B_3(\lambda,Z_3) \, B_2(\lambda,Z_2)  \, A_1(\lambda,Z_1)$ in the two possible ways, with certain $Z_i$, 
using (\ref{Lax_R^-1}) and (\ref{Lax_T}), leads to
\be
        \mathcal{R}^{-1}_{\bsy{23}}(2,3) \circ \mathcal{T}_{\bsy{13}}(1,3) \circ \mathcal{T}_{\bsy{12}}(1,2)
    =  \mathcal{T}_{\bsy{12}}(1,2) \circ \mathcal{T}_{\bsy{13}}(1,3) \circ \mathcal{R}^{-1}_{\bsy{23}}(2,3)         
              \, .  \label{RinvTT=TTRinv}
\ee
\item Moreover, transforming $B_1(\lambda,X_1) \, B_2(\lambda,X_2) \, A_3(\lambda,X_3)$ to 
$A_3(\lambda,Z_3) \, B_2(\lambda,Z_2)  \, A_1(\lambda,Z_1)$, with certain $Z_i$, using (\ref{Lax_R^-1}) 
and (\ref{Lax_T^-1}), we obtain
\be
        \mathcal{T}^{-1}_{\bsy{23}}(2,3) \circ \mathcal{T}^{-1}_{\bsy{13}}(1,3) \circ \mathcal{R}^{-1}_{\bsy{12}}(1,2)
    =  \mathcal{R}^{-1}_{\bsy{12}}(1,2) \circ \mathcal{T}^{-1}_{\bsy{13}}(1,3) \circ \mathcal{T}^{-1}_{\bsy{23}}(2,3)  
       \, .   \label{TinvTinvRinv=RinvTinvTinv}
\ee
But this equivalent to (\ref{TTR=RTT}). 
\item Finally, rewriting $B_1(\lambda,X_1) \, A_2(\lambda,X_2) \, A_3(\lambda,X_3)$ in the form
$A_3(\lambda,Z_3) \, A_2(\lambda,Z_2)  \, B_1(\lambda,Z_1)$, using (\ref{Lax_R}) and (\ref{Lax_T^-1}),
implies
\be
        \mathcal{R}_{\bsy{23}}(2,3) \circ \mathcal{T}^{-1}_{\bsy{13}}(1,3) \circ \mathcal{T}^{-1}_{\bsy{12}}(1,2)
    =  \mathcal{T}^{-1}_{\bsy{12}}(1,2) \circ \mathcal{T}(1,3)^{-1}_{\bsy{13}} \circ \mathcal{R}_{\bsy{23}}(2,3) \, ,
               \label{RTinvTinv=TinvTinvR}
\ee
which, however, is equivalent to (\ref{RinvTT=TTRinv}).
\end{enumerate}

\begin{remark}
A system of equations like that given by the Yang-Baxter equation (\ref{YB_eq}), supplemented by 
(\ref{TTR=RTT}) or  (\ref{RTinvTinv=TinvTinvR}), appeared 
in \cite{Hlav94} under the name ``braided Yang-Baxter equations''. 
The same holds for the Yang-Baxter equation for $\mathcal{R}^{-1}$, supplemented by
 (\ref{RinvTT=TTRinv}) or (\ref{TinvTinvRinv=RinvTinvTinv}). Also see \cite{Vlad93}.
Via (\ref{TTR=RTT}) and (\ref{RinvTT=TTRinv}), as well as via (\ref{TinvTinvRinv=RinvTinvTinv}) 
and (\ref{RTinvTinv=TinvTinvR}), we have examples of what has been called ``WXZ system'' in \cite{Hlav+Snob99}, 
later also named ``Yang-Baxter system'' \cite{Brze+Nich05}. This system apparently first appeared in \cite{Vlad93}. 
Here we obtained solutions of these systems. Equation (\ref{TTR=RTT}) also appeared in \cite{KNW09}, where a solution 
emerged in the context of the \emph{scalar} discrete KP hierarchy. Since the solutions considered in the present work become 
trivial in the scalar case, the latter solution is of a different nature. 
\end{remark}

\section{The p2DTL$_K$ equation}
\label{sec:Toda}
In this section, we address the following matrix version of the potential two-dimensional Toda lattice equation, 
\be
  \varphi_{xy} - \varphi^+ + 2 \varphi - \varphi^- = (\varphi^+ - \varphi) K \varphi_y - \varphi_y K (\varphi - \varphi^-)
       \, ,   \label{matrixToda_K}
\ee
where $\varphi$ is an $m \times n$ matrix of (real or complex) functions and $K$ a constant $n \times m$ matrix of maximal rank. 
A subscript indicates a partial derivative with respect to the respective variable, here $x$ or $y$. 
A superscript $+$ or $-$ means a shift, respectively inverse shift, in a discrete variable, which we will denote by $k$. 
We refer to this equation as p2DTL$_K$. 

In the vector case $n=1$, writing $K=(k_1,\ldots,k_m)$, (\ref{matrixToda_K}) reads
\bez
   \varphi_{i,xy} - \varphi_i^+ + 2 \varphi_i - \varphi_i^- = (\varphi_i^+ - \varphi_i) \sum_j k_j \varphi_{j,y} 
      - \varphi_{i,y} \sum_j k_j (\varphi_j - \varphi_j^-) \qquad i=1,\ldots,m \, .
\eez
By a transformation and redefinition of $\varphi$, we can then achieve that $K = (1,0,\ldots,0)$, so that
\bez
  && \varphi_{1,xy} - \varphi_1^+ + 2 \varphi_1 - \varphi_1^- = (\varphi_1^+ - \varphi_1) \varphi_{1,y} 
      - \varphi_{1,y} (\varphi_1 - \varphi_1^-) \, , \\
  && \varphi_{j,xy} - \varphi_j^+ + 2 \varphi_j - \varphi_j^- = (\varphi_j^+ - \varphi_j) \varphi_{1,y} 
      - \varphi_{j,y} (\varphi_1 - \varphi_1^-)  \qquad j=2,\ldots,m \, ,    
\eez
which is the scalar potential 2DTL equation, extended by $m-1$ linear equations. 

In terms of new independent variables  
\bez
    t = x+y \, , \qquad  z = x-y \, ,
\eez    
equation (\ref{matrixToda_K}) reads
\be
    \varphi_{tt} - \varphi_{zz}  - \varphi^+ + 2 \varphi - \varphi^- 
  = (\varphi^+ - \varphi) K (\varphi_t - \varphi_z) - (\varphi_t - \varphi_z) K (\varphi - \varphi^-) \, .
      \label{matrixToda_K_2}
\ee
If $\varphi$ is independent of $z$, the last equation reduces to 
\be
 \varphi_{tt}  - \varphi^+ + 2 \varphi - \varphi^- 
  = (\varphi^+ - \varphi) K \varphi_t - \varphi_t K (\varphi - \varphi^-) \, .      \label{p1DTL_K}
\ee
We will refer to this equation as p1DTL$_K$. 

\begin{remark}
\label{rem:scalar2DTL}
In terms of
\bez
      u := \varphi_y \, ,
\eez     
in the scalar case ($n=m=1$), and after differentiation with respect to $y$, (\ref{matrixToda_K}) with $K=1$ leads to
the \emph{two-dimensional Toda lattice (2DTL) equation} \cite{Mikh79} 
(also see \cite{HIK88,Ueno+Taka84,Hiro04,Bion+Wang10,Taka18}, for example)
\be
     (\ln(1+u))_{xy} = u^+ - 2u + u^-  \, .  \label{2DTL}
\ee     
A continuum limit of the 2DTL equation is the KP-II equation \cite{Bion+Wang10}. 
 If $u$ is independent of $z$, the 2DTL equation (\ref{2DTL}) reduces to the one-dimensional 
Toda lattice equation \cite{Toda89}
\bez
      (\ln(1+u))_{tt} = u^+ - 2u + u^- \, . 
\eez
Correspondingly, we may regard (\ref{p1DTL_K}) as a 
matrix version of the potential one-dimensional Toda lattice equation.
\end{remark}

\begin{remark}
\label{rem:from_scalar_to_matrix_solutions}
Multiplying any solution of the \emph{scalar} version of (\ref{matrixToda_K}) by an arbitrary constant  
$K$-projection matrix, yields a solution of the matrix equation (\ref{matrixToda_K}). 
In this way, a single scalar soliton solution determines single matrix soliton solutions of any rank up to
the maximal. 
\end{remark}

\subsection{A binary Darboux transformation for the p2DTL$_K$ equation}
\label{subsec:Toda_bDT}
The following binary Darboux transformation is a special case of a general result in 
bidifferential calculus, see Appendix~\ref{app:BDT}. 
Let $N \in \bbN$. 
The integrability condition of the linear system
\be
   \theta_x = \theta^+ - \theta + (\varphi_0^+ - \varphi_0) K \theta \, , \qquad
   \theta_y = \theta - \theta^- - \varphi_{0,y} K \theta^-  \, ,   \label{Toda_lin_sys}
\ee
where $\theta$ is an $m \times N$ matrix, is the p2DTL$_K$ equation for $\varphi_0$. The same holds for 
the adjoint linear system
\be
   \chi_x = \chi - \chi^- - \chi K (\varphi_0^+ - \varphi_0)  \, , \qquad 
   \chi_y = \chi^+ - \chi + \chi^+ K \varphi_{0,y}^+ \, ,  \label{Toda_adj_lin_sys}
\ee
where $\chi$ is an $N \times n$ matrix. So let $\varphi_0$ be a given solution of (\ref{matrixToda_K}).
Let the Darboux potential $\Omega$ satisfy the consistent  system of $N \times N$ matrix equations
\be
   \Omega - \Omega^- = - \chi K \theta \, , \qquad
   \Omega_x = - \chi K \theta^+ \, , \qquad
   \Omega_y = - \chi^+ K \theta - \chi^+ K \varphi^+_{0,y} \theta \, . \label{Toda_Omega}
\ee
Where $\Omega$ is invertible,  
\be
    \varphi = \varphi_0 - \theta (\Omega^-)^{-1} \chi^-    \label{Toda_new_solution}
\ee
is then a new solution of the p2DTL$_K$ equation (\ref{matrixToda_K}).

\begin{remark}
\label{rem:BDT_transf}
The equations (\ref{Toda_lin_sys}) - (\ref{Toda_new_solution}) are invariant under the transformation
\bez
    \theta \mapsto \theta \, C_1 \, , \qquad 
    \chi \mapsto C_2 \, \chi \, , \qquad
    \Omega \mapsto C_2 \, \Omega \, C_1 \, ,
\eez
with any invertible constant $N \times N$ matrices $C_a$, $a=1,2$. This observation is helpful in order 
to reduce the set of parameters, on which a generated solution depends. 
\end{remark}

Using (\ref{Toda_new_solution}) and the second of (\ref{Toda_Omega}), we find
\be
    \mathrm{tr}( K \varphi) &=& \mathrm{tr}(K \varphi_0) - \mathrm{tr}(K \theta \, (\Omega^-)^{-1} \chi^-) 
                 = \mathrm{tr}(K \varphi_0) - \mathrm{tr}( (\Omega^-)^{-1} \chi^- \, K \, \theta) \nonumber \\
                 &=& \mathrm{tr}(K \varphi_0) + \mathrm{tr}(\Omega^{-1} \Omega_x)^- 
                  = \mathrm{tr}(K \varphi_0) + (\log \det \Omega)_x^- \, ,
                     \label{trKvarphi}
\ee
so that $\det \Omega$ plays a role similar to the (Hirota) $\tau$-function of the (scalar) 2DTL equation.

\subsubsection{Solutions for vanishing seed}
\label{subsec:zero_seed}
The linear system (\ref{Toda_lin_sys}) with $\varphi_0=0$ reads
\bez
   \theta_x = \theta^+ - \theta \, , \qquad 
   \theta_y = \theta - \theta^- \, .
\eez
It possesses solutions of the form
\bez
    \theta = \sum_{a=1}^A \theta_a \, e^{\tilde{\vartheta}(P_a)} \, P_a^k \, .
\eez
Here  $\theta_a$, $a=1,\dots,A$, are constant $m \times N$ matrices,
$k$ denotes the discrete variable, $P_a$, $a=1,\ldots,A$, are constant $N \times N$ matrices, and
\be
     \tilde{\vartheta}(P) = (P-I) \, x + (I-P^{-1}) \, y \, .   \label{Toda_phase}
\ee
Correspondingly, the adjoint linear system (\ref{Toda_adj_lin_sys}) takes the form
\bez
   \chi_x = \chi - \chi^-  \, , \qquad 
   \chi_y = \chi^+ - \chi  \, ,  
\eez
which is solved by
\bez
    \chi = \sum_{b=1}^B e^{-\tilde{\vartheta}(Q_b)} \, Q_b^{-k} \, \chi_b \, ,
\eez
where $\chi_b$, $b=1,\ldots,B$, are constant $N \times n$ matrices and 
$Q_b$, $b=1,\ldots,B$, are constant $N \times N$ matrices. 

The equations for the Darboux potential $\Omega$ are reduced to
\bez
   \Omega - \Omega^- = - \chi K \theta \, , \qquad 
   \Omega_x = - \chi K \theta^+ \, , \qquad 
   \Omega_y = - \chi^+ K \theta  \, . 
\eez
Writing
\be
   \Omega = \Omega _0 + \sum_{a,b} e^{-\tilde{\vartheta}(Q_b)} Q_b^{-k} W_{ba} \, P_a^{k+1} \, e^{\tilde{\vartheta}(P_a)} \, ,  
                 \label{Omega_ansatz}
\ee
with a constant $N \times N$ matrix $\Omega_0$, 
it follows that $W_{ba}$ has to satisfy the Sylvester equation
\be
    Q_b W_{ba} - W_{ba} P_a = \chi_b K \theta_a \, .   \label{Sylvester}
\ee
If 
\be 
      P_a = \mathrm{diag}(p_{1,a},\ldots,p_{N,a}) \, , \qquad
      Q_b = \mathrm{diag}(q_{1,b},\ldots,q_{N,b}) \, ,   \label{P,Q_diag}
\ee      
and if $p_{i,a} \neq q_{j,b}$ for all $i,j=1,\ldots,N$ and $a=1,\ldots,A$, $b=1,\ldots,B$, then the 
unique solution is known to be given by the Cauchy-like $N \times N$ matrices 
\bez
     W_{ba} = \Big( \frac{\chi_{ib} K \theta_{ja} }{q_{ib}-p_{ja} } \Big) \, .
\eez
Assuming that $\Omega_0$ is invertible, Remark~\ref{rem:BDT_transf} shows that we can set $\Omega_0 = 1_N$ 
without loss of generality. The remaining transformations, according to Remark~\ref{rem:BDT_transf}, 
can be used to reduce the parameters in $\theta$ or $\chi$.

\subsection{Pure solitons}
\label{subsec:pure}
We further restrict the class of p2DTL$_K$ solutions specified in Section~\ref{subsec:zero_seed} by setting $A=B=1$ 
and assume that the matrices $P := P_1$ and $Q := Q_1$ are diagonal (so that (\ref{P,Q_diag}) holds). Solutions from 
this class which are regular and satisfy the spectrum condition
\bez
      \mathrm{spec}(P) \cap \mathrm{spec}(Q) = \emptyset
\eez
will be called ``pure solitons''.

Let us write
\bez
  && P = \mathrm{diag}(p_{1,1},\ldots,p_{N,1}) =: \mathrm{diag}(p_1,\ldots,p_N)  \, , \\
  && Q =: \mathrm{diag}(p_{1,2},\ldots,p_{N,2}) =: \mathrm{diag}(q_1,\ldots,q_N) \, ,  \\
  && \theta_1 = (\xi_1, \ldots, \xi_N)(Q-P) \, , \qquad
     \chi_1 = \left( \begin{array}{c} \eta_1 \\ \vdots \\ \eta_N \end{array} \right) \, ,
\eez
where $\xi_i$, $i=1,\ldots,N$, are constant $m$-component column vectors and $\eta_i$, $i=1,\ldots,N$, are 
constant $n$-component row vectors.
We shall assume that $p_i>0$ and $q_i>0$, $i=1,\ldots,N$, since otherwise the generated solution of (\ref{matrixToda_K}) will 
be singular. The above spectrum condition means $p_i \neq q_j$ for $i,j=1,\ldots,N$, and we have
\bez
    W := W_{1,1} = \Big( \frac{\kappa_{ij} \, (q_j-p_j)}{q_i-p_j} \Big) \, , \qquad \kappa_{ij} = \eta_i K \xi_j \, .
\eez
Introducing
\bez
    \vartheta(p) := p \, x - p^{-1} y + k \, \log p
                       = \frac{1}{2} (p-p^{-1}) \, t + \frac{1}{2} (p+p^{-1}) \, z  + k \, \log p    \, ,
\eez
provisionally\footnote{Finally we only have to make sure that the expressions for $\vartheta_I$ (see below), appearing in 
a generated solution of the p2DTL$_K$ equation, are real.} 
assuming $p>0$, we obtain
\bez
    \Omega_{ij} = \delta_{ij} + \frac{\kappa_{ij} (q_j-p_j)}{q_i-p_j} \, e^{\vartheta(p_j)^+ - \vartheta(q_i)} \, .
\eez
Let us introduce 
\bez
   && \vartheta_{i,1} := \vartheta(p_i)^+ \, , \qquad \vartheta_{i,2} := \vartheta(q_i) \, , \\
   && \vartheta_I := \sum_{i=1}^N \vartheta_{i,a_i}  \qquad \mbox{if} \quad
    I = (a_1,\ldots,a_N) \in \{1,2\}^N  \, .
\eez
Instead of using $(a_1,\ldots,a_N)$ as a subscript, we simply write $a_1 \ldots a_N$ 
in the following. For example, $\vartheta_{a_1 \ldots a_N} = \vartheta_{(a_1,\ldots,a_N)}$. 
 From (\ref{Toda_new_solution}) we find that a pure soliton solution of the p2DTL$_K$ equation 
can be expressed as
\be
    \varphi^+ = \frac{F}{\tau} \, ,   \label{varphi^+=F/tau}
\ee
with
\be
    \tau &:=& e^{\vartheta_{\boldsymbol{2}}} \, \det \Omega  \, ,   \label{tau_pure}  \\
    F &:=& - e^{\vartheta_{\boldsymbol{2}}} \, \theta_1 \, e^{\vartheta(P)^+} \, 
          \mathrm{adj}(\Omega) \, e^{-\vartheta(Q)} \, \chi_1 \, ,    \label{F_pure}
\ee          
where $\mathrm{adj}(\Omega)$ denotes the adjugate of the matrix $\Omega$ and $\boldsymbol{2} := 2\ldots 2 = (2,\ldots,2)$.
The following result is proved in the same way as Proposition~3.1 in \cite{DMH19LMP}.

\begin{proposition}
\label{prop:tau,F_expansion}
$\tau$ and $F$ have expansions
\be
      \tau = \sum_{I \in \{1,2\}^N} \mu_I \, e^{\vartheta_I} \, ,    \label{tau_pure_expansion}  \\
       F = \sum_{I \in \{1,2\}^N} M_I \, e^{\vartheta_I} \, , \label{F_pure_expansion}
\ee
with constants $\mu_I$ and constant $m \times n$ matrices $M_I$.
We have $\mu_{\boldsymbol{2}}=1$ and $M_{\boldsymbol{2}}=0$.  \hfill $\Box$
\end{proposition}

Besides conditions imposed on $p_i$ and $q_i$ such that all the $\vartheta_I$ appearing in (\ref{varphi^+=F/tau}) 
are real, the regularity of a pure $N$-soliton solution requires 
$\mu_I \geq 0$ for all $I \in \{1,2\}^N$, and $\mu_J >0$ for at 
least one $J \in \{1,2\}^N$. We will impose the slightly stronger condition $\mu_I > 0$ for all $I \in \{1,2\}^N$.
\vspace{.2cm}

It follows that
\be
    u^+ = \varphi_y^+ = \Big( \frac{1}{\tau} \sum_{I \in \{1,2\}^N} M_I \, e^{\vartheta_I} \Big)_y
        = \frac{1}{2 \tau^2} \sum_{I,J \in \{1,2\}^N} (\tilde{p}_J - \tilde{p}_I)(M_I \mu_J - \mu_I M_J) 
          \, e^{\vartheta_I} e^{\vartheta_J} \, ,     \label{u(mu,M)}
\ee
where
\bez
    \tilde{p}_I := \sum_{i=1}^N \frac{1}{p_{i,a_i}} \qquad \mbox{if} \quad I = (a_1,\ldots,a_N) \in \{1,2\}^N  \, .
\eez

\begin{example}
For $N=1$, writing $p_1=p$, $q_1=q$, $\xi_1=\xi$, $\eta_1=\eta$ and $\kappa = \eta K \xi$, we have the 
single soliton solution
\bez
   \varphi = \frac{\kappa \, (p-q) \, e^{\vartheta(p)}}{e^{\vartheta(q)^-} + \kappa \, e^{\vartheta(p)}} 
              \, \frac{\xi \otimes \eta}{\kappa} \, , 
\eez
which leads to
\bez
   u = \varphi_y 
     = \frac{(p-q)^2}{4pq} \, \mathrm{sech}^2\Big[\frac{1}{2} ( \vartheta(p) - \vartheta(q)^- + \log \kappa ) \Big] 
      \, \frac{\xi \otimes \eta}{\kappa}   \, .        
\eez
In terms of the variables $t = x+y$ and $z = x-y$, it reads
\bez
    u= \frac{(p-q)^2}{4 p q} \text{sech}^2\left[\frac{1}{2} \Big( \frac{1}{2}(p-q-p^{-1} + q^{-1}) t 
       - \frac{1}{2} (p-q+p^{-1}-q^{-1}) z + \log(p/q) \, k + \log(q \kappa) \Big) \right] 
          \, \frac{\xi \otimes \eta}{\kappa} \, .
\eez
We have to restrict the parameters such that $p/q > 0$ and $q \kappa >0$. 
The solution becomes independent of $z$ if we choose $q = p^{-1}$, in which case the above $\varphi$ reduces 
to a single soliton solution of the p1DTL$_K$ equation (\ref{p1DTL_K}), and we have
\bez
    u= \frac{(p-p^{-1})^2}{4} \text{sech}^2\left[ \frac{1}{2}(p-p^{-1}) \, t 
       + \log(p) \, k + \frac{1}{2} \log(\kappa/p) \right] 
          \, \frac{\xi \otimes \eta}{\kappa} \, . 
\eez
It is obvious from (\ref{Toda_new_solution}) and the sizes of its matrix constituents that, for $N=1$, 
the binary Darboux transformation with zero seed can only yield a \emph{rank one} solution.  
\end{example}

\subsubsection{Tropical limit of pure solitons}
\label{subsec:trop}
We define the tropical limit of a matrix soliton solution via the tropical 
limit of the scalar function $\tau$ (cf. \cite{DMH11KPT,DMH12KPBT,DMH14KdV}). Let
\be
    \varphi_I := \varphi\Big|_{\vartheta_J \to -\infty, J \neq I} = \frac{M_I}{\mu_I} \, . \label{varphi_I}
\ee
In the region of $\mathbb{R}^2 \times \mathbb{Z}$, where a phase $\vartheta_I$ dominates all others, in the sense that 
$\log(\mu_I \, e^{\vartheta_I}) > \log(\mu_J \, e^{\vartheta_J})$ for all participating $J \neq I$, 
the tropical limit of the potential $\varphi$ is given by (\ref{varphi_I}).\footnote{Such ``dominating phase regions'' 
have also been used, for example, in \cite{IWS03,Maru+Bion04,Bion+Chak06,Chak+Koda08JPA}, mostly for the 
asymptotic analysis of solitons. In our work we apply it to the whole soliton solution, not just in asymptotic regions. 
Also see \cite{DMH11KPT,DMH12KPBT,DMH14KdV,DMH19LMP,DMH18p,DMHC19}.}
These expressions do not depend on the variables $x,y,k$ (respectively, $z,t,k$). 

The boundary between the regions associated with the phases $\vartheta_I$ and $\vartheta_J$ 
is determined by the condition 
\be
     \mu_I \, e^{\vartheta_I} = \mu_J \, e^{\vartheta_J} \, .    \label{phase_boundary}
\ee
Not all parts of such a boundary are ``visible'', in general, since some of them may lie in a region where 
a third phase dominates the two phases. The tropical limit of a soliton solution, more precisely, of the variable $u$, 
has support on the visible parts of the boundaries between the regions associated with phases appearing in $\tau$. 

For $I = (a_1,\ldots,a_N)$ we set
\bez
     I_j(a) = (a_1,\ldots,a_{j-1},a,a_{j+1}, \ldots,a_N) \, .
\eez
The $j$-th soliton (having parameters $p_j$ and $q_j$) lives, in the tropical limit, on the set of 
two-dimensional plane segments determined, via (\ref{phase_boundary}), by 
\bez
    e^{\vartheta_{I_j(1)} - \vartheta_{I_j(2)}} = \frac{\mu_{I_j(2)}}{\mu_{I_j(1)}} \, ,
\eez
for all $I \in \{1,2\}^N$. More explicitly, the last equation reads
\bez
    (p_j - q_j) \, x + (q_j^{-1} - p_j^{-1}) \, y + \log (p_j/q_j) \, k  
        + \log\left(p_j \, \frac{\mu_{I_j(1)}}{\mu_{I_j(2)}}\right) = 0 \, ,
\eez
which requires
\be
    p_j/q_j >0 \, , \qquad  p_j \, \frac{\mu_{I_j(1)}}{\mu_{I_j(2)}} > 0 \quad \forall I \in \{1,2\}^N \, .  
    \label{reg_cond}
\ee   
All these plane segments are parallel. 
In general there are relative shifts between the segments, they do not constitute together a single plane. 
This gives rise to the familiar (asymptotic) ``phase shift'' of solitons caused by their interaction. 
Fig.~\ref{fig:vec2DToda_2s_21} below shows this for a 2-soliton example, considered at constant time, so that 
the configuration of planes is projected to a graph in two dimensions. 
For $j=1,\ldots,N$, the regularity conditions (\ref{reg_cond}) will be assumed in the following.

On a (visible) boundary segment, the value of $u$ is given by
\bez
     u_{IJ} = - \frac{1}{4} (\tilde{p}_I - \tilde{p}_J) \left( \varphi_I - \varphi_J \right)  \, .
\eez
This follows from (\ref{u(mu,M)}) by use of (\ref{varphi_I}) and (\ref{phase_boundary}).
Instead of the above expressions for the tropical values of $u$, we will rather consider
\be
    \hat{u}_{IJ} = \frac{\varphi_I - \varphi_J}{p_I - p_J}  \, ,  \label{hatu_IJ}
\ee
where
\bez
    p_I := \sum_{i=1}^N p_{i,a_i}  \qquad \mbox{if} \quad  I = (a_1,\ldots,a_N) \in \{1,2\}^N \, .
\eez
(\ref{hatu_IJ}) has the form of a discrete derivative.

Using (\ref{trKvarphi}), (\ref{varphi^+=F/tau}), (\ref{tau_pure_expansion}) and (\ref{F_pure_expansion}), we find 
\be
     \mathrm{tr}(K M_I) = (p_I - p_{\boldsymbol{2}}) \, \mu_I \, ,  \label{tr(K phi_I)}
\ee
and thus
\bez
    \mathrm{tr}(K \varphi_I) = p_I - p_{\boldsymbol{2}} \, .
\eez
As a consequence, we have the normalization
\be
     \mathrm{tr}(K \hat{u}_{IJ}) = 1 \, .    \label{normalization}
\ee

\begin{remark}
We note that 
\bez
     \mathrm{tr}(K u_{I_j(1),I_j(2)}) 
   = - \frac{1}{4} (\tilde{p}_{I_j(1)} - \tilde{p}_{I_j(2)})(p_{I_j(1)} - p_{I_j(2)}) 
   = - \frac{1}{4} \Big( \frac{1}{p_j} - \frac{1}{q_j} \Big) (p_j-q_j) \, ,
\eez
which shows that its value is the same everywhere (i.e., for all $I$) on the tropical support of the 
$j$-th soliton. 
\end{remark}

\subsection{Pure 2-soliton solution and the Yang-Baxter map}
\label{subsec:2sol}
For $N=2$ we find $\varphi^+ = F/\tau$ with
\bez
  \tau &=& \alpha_{12} \, \kappa_{11} \kappa_{22} \, e^{\vartheta_{11}}  + \kappa_{11} \, e^{\vartheta_{12}} 
           + \kappa_{22} \, e^{\vartheta_{21}} + e^{\vartheta_{22}}  \, ,  \\         
     F &=& (p_1-q_1)(p_2-q_2) \Big( \frac{\kappa_{22}}{p_2-q_2} \, \xi_1 \otimes \eta_1 
           + \frac{\kappa_{11}}{p_1-q_1} \, \xi_2 \otimes \eta_2
           - \frac{\kappa_{12}}{p_2-q_1} \, \xi_1 \otimes \eta_2 - \frac{\kappa_{21}}{p_1-q_2} \, \xi_2 \otimes \eta_1 \Big) 
           \, e^{\vartheta_{11}} \\
       &&  + (p_1-q_1) \, \xi_1 \otimes \eta_1 \, e^{\vartheta_{12}} + (p_2-q_2) \, \xi_2 \otimes \eta_2 \, e^{\vartheta_{21}} \, , 
\eez
where
\bez
     \alpha_{12} = 1 - \frac{(p_1-q_1)(p_2-q_2) \, \kappa_{12} \, \kappa_{21}}{(p_2-q_1)(p_1-q_2) \, \kappa_{11} \, \kappa_{22}} 
\eez
and 
\bez
      \vartheta_{11} = \vartheta(p_1)^+ + \vartheta(p_2)^+ \, , \quad
      \vartheta_{12} = \vartheta(p_1)^+ + \vartheta(q_2) \, , \quad
      \vartheta_{21} = \vartheta(p_2)^+ + \vartheta(q_1) \, , \quad
      \vartheta_{22} = \vartheta(q_1) + \vartheta(q_2) \, .
\eez
The above expressions for $\tau$ and $F$ coincide with those derived in the KP$_K$ case \cite{DMH19LMP}. The only difference 
is in the expressions for the phases, but the latter do not enter the expressions for the polarizations.  

\begin{example}
\label{ex:vec2DToda_2s_21}
Let 
\bez
  K = \left(\begin{array}{cc} 1 & 1 \end{array} \right) \, , \;
  \xi_1 = \left(\begin{array}{c} 1 \\ 0 \end{array} \right) \, , \;
    \xi_2 = \left(\begin{array}{c} 0 \\ 1 \end{array} \right) \, , \;
    \eta_1 = \eta_2 = 1 \, , \;
    q_1 = 1/4, \; q_2 = 3, \; p_1 = 3/2, \; p_2 = 2 \, .
\eez
Fig.~\ref{fig:vec2DToda_2s_21} shows the phase constellation and the tropical limit graph of the corresponding 
2-soliton solution of the vector p2DTL$_K$ equation at $t=0$.\footnote{Here, and in all other plots in this work, 
we have chosen the parameters in such a way that, as the vertical coordinate tends to $-\infty$, the solitons 
are naturally ordered in horizontal direction.}

\begin{figure}[h] 
\begin{center}
\begin{minipage}{0.04\linewidth}
\vspace*{.15cm}
\bez
 \begin{array}{cc} k & \\ \uparrow & \\ & \rightarrow z \end{array}                
\eez
\end{minipage}
\includegraphics[scale=.25]{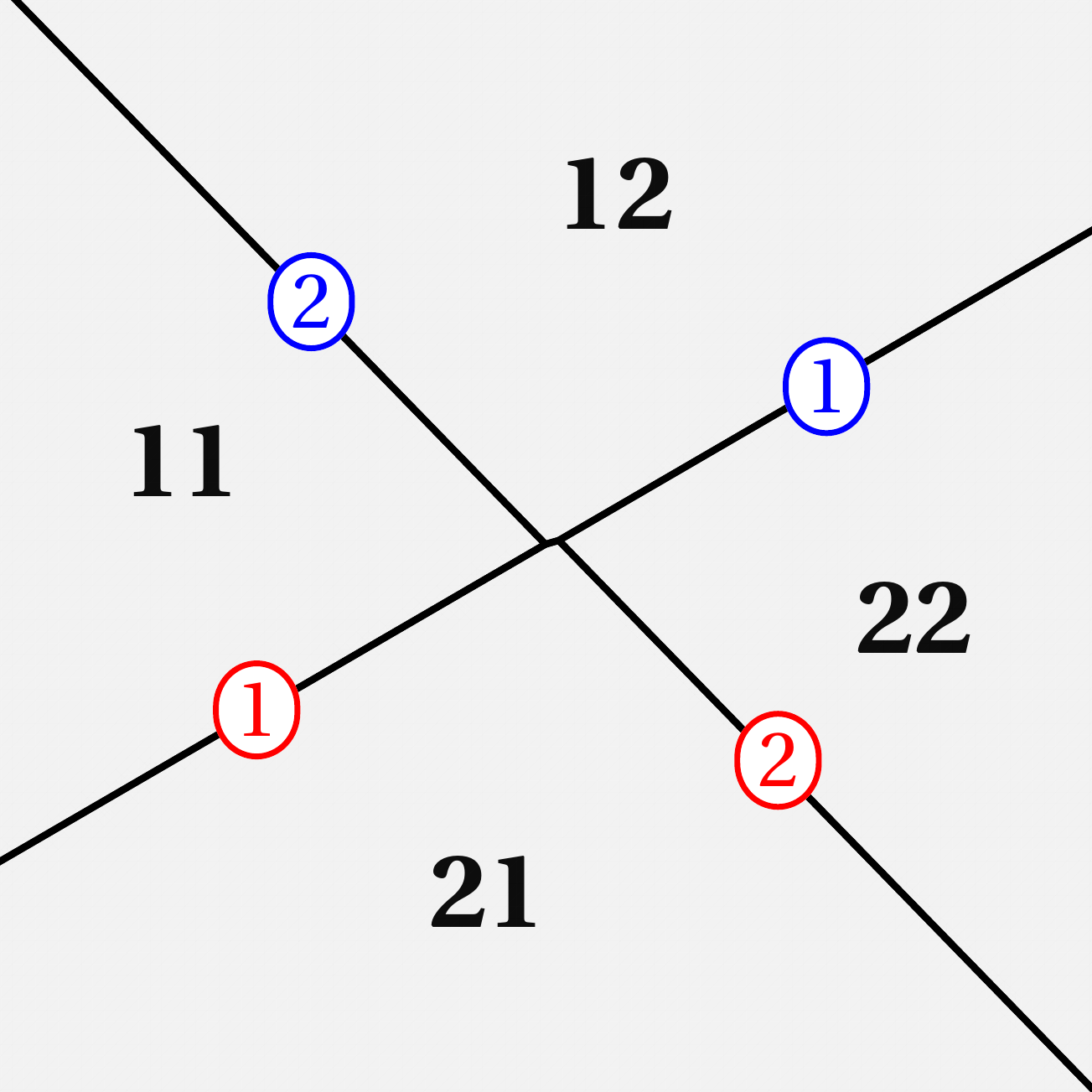} 
\parbox{15cm}{
\caption{Dominating phase regions and tropical limit graph of a pure 2-soliton solution of the vector p2DTL$_K$ equation, 
at $t=x + y=0$. The horizontal coordinate is $z=x-y$ and the discrete coordinate $k$ is continuously extended. 
The numbering of dominating phase regions corresponds to $I=(1,1),(1,2),(2,1),(2,2)$. 
We also marked the two parts of soliton 1, respectively 2.  
\label{fig:vec2DToda_2s_21} }
}
\end{center}
\end{figure}

\end{example}

Let us now consider the graph in Fig.~\ref{fig:vec2DToda_2s_21} as a scattering process evolving from bottom to top.
Defining
\be
    u_1 := \hat{u}_{11,21} \, , \quad
    u_2 := \hat{u}_{21,22} \, , \quad
    u_1' := \hat{u}_{12,22} \, , \quad
    u_2' := \hat{u}_{11,12} \, ,     \label{choice_of_initial/final_polarizations}
\ee
we find
\be
  u_1 &=& \frac{1}{\alpha_{12}  \kappa_{11}} \, 
     \Big( \xi_1 - \frac{ (p_2-q_2) \kappa_{21}}{(p_1-q_2) \kappa_{22}} \, \xi_2 \Big) \otimes
     \Big( \eta_1 - \frac{ (p_2-q_2) \kappa_{12}}{(p_2-q_1) \kappa_{22}} \, \eta_2 \Big) \nonumber \\
      &=& \frac{1}{\alpha_{12}} \, A_2(p_1,u_2) \, u'_1 \, \tilde{B}_2(q_1,u_2) \, , \nonumber \\ 
  u_2 &=& \frac{\xi_2 \otimes \eta_2}{\kappa_{22}} \, ,  \nonumber\\
  u_1' &=& \frac{\xi_1 \otimes \eta_1}{\kappa_{11}}  \, ,  \nonumber\\
  u_2' &=& \frac{1}{\alpha_{12}  \kappa_{22}} \, 
       \Big( \xi_2 - \frac{ (p_1-q_1) \kappa_{12}}{(p_2-q_1) \kappa_{11} } \, \xi_1 \Big) \otimes
       \Big( \eta_2 - \frac{ (p_1-q_1) \kappa_{21}}{(p_1-q_2) \kappa_{11} } \, \eta_1 \Big)   \nonumber \\
      &=& \frac{1}{\alpha_{12}} A_1(p_2,u'_1) \, u_2 \, \tilde{B}_1(q_2,u'_1) \, .    
          \label{2soliton_polarization_relations}
\ee
We observe that $u_i$ and $u_i'$ all have rank one. They are $K$-projections, i.e.,
\be
     u_i K u_i = u_i \, , \qquad u_i' K u_i' = u_i' \, .   \label{u_i_K-proj}
\ee
Furthermore, they satisfy
\bez
     (p_1-q_1)(u_1'-u_1) +(p_2-q_2)(u_2'-u_2) = 0 \, .   
\eez
The equations (\ref{2soliton_polarization_relations}) imply
\bez
    u'_1 = \alpha_{12} \, B_2(p_1,u_2) \, u_1 \, \tilde{A}_2(q_1,u_2) \, , \qquad
    u'_2 = \alpha_{12} \, A_1(q_2,u_1) \, u_2 \, \tilde{B}_1(p_2,u_1) \, ,
\eez
provided that $\{p_1,p_2\} \cap \{q_1,q_2\} = \emptyset$.
Comparison with (\ref{X_1,2,out}) shows that $(u_1,u_2) \mapsto (u_1',u_2')$ provides us with 
a realization of the Yang-Baxter map $\mathcal{R}$.\footnote{The 
definition of $\alpha_{12}$ in Section~\ref{subsec:2sol} is in accordance with the expression 
in (\ref{alpha_12}). }
 
\begin{remark}
\label{rem:vector2DTL_K}
If $n=1$, we are dealing with an $m$-component vector 2DTL equation. Then $\eta_i$ and $K \xi_i$, $i=1,\ldots,N$, 
are scalars. In this case the Yang-Baxter map is linear,
\bez
      (u_i', u_j') = (u_i , u_j) R(i,j) \, , 
\eez
with $R(i,j)$ defined in (\ref{R-matrix}). It solves the Yang-Baxter equation 
$R_{\bsy{12}}(1,2) \, R_{\bsy{13}}(1,3) \, R_{\bsy{23}}(2,3) 
   = R_{\bsy{23}}(2,3) \,  R_{\bsy{13}}(1,3) \, R_{\bsy{12}}(1,2)$ 
on a threefold direct sum. Also see \cite{DMH19LMP} for the case of the vector KP$_K$ equation. Introducing
\bez
   \tilde{R}(i,j) := \left(\begin{array}{cc} 1 & 0 \\
                                     0 & -\frac{p_j-q_i}{p_i-q_j} \end{array}\right) \, , \qquad
   S(i,j) := \left(\begin{array}{cc} \frac{p_i-q_i}{p_j-q_j} & -1 \\
                                    1 &  1 \end{array}\right) \, , 
\eez
we have $R(i,j) = S(i,j) \tilde{R}(i,j) S(i,j)^{-1}$, and $\tilde{R}(i,j)$ also satisfies the Yang-Baxter equation. We further note 
that $RA(i,j) := A(i,j)^{-1} R(i,j) A(i,j)$ with $A(i,j)=\mathrm{diag}(a_{ij},b_{ij})$ solves the Yang-Baxter equation
if the constants $a_{ij},b_{ij}$ satisfy the relations $ a_{ik} b_{ij} b_{jk} = a_{ij} a_{jk} b_{ik}$ for pairwise 
distinct $i,j,k$. If we drop the normalization condition (\ref{normalization}) in the computation of the Yang-Baxter map, 
the resulting $R$-matrix turns out to be of the latter form. The same holds if we consider $v = \varphi^+ - \varphi$ 
instead of $\hat{u}$. 
\end{remark}

\subsection{Yang-Baxter and non-Yang-Baxter maps at work}
\label{subsec:other_maps}
In Section~\ref{subsec:2sol} we looked at the relation between the polarizations associated with the boundary segments of 
dominant phase regions of a pure 2-soliton solution, selecting a ``propagation direction''. But there is actually no preferred direction. 
It is therefore more adequate to regard (\ref{2soliton_polarization_relations}) just as determining a relation 
between four polarizations, and there are several ways in which this determines a map from two ``incoming'' to two  
``outgoing'' polarizations. 

\begin{figure}[h] 
\begin{center}
\begin{minipage}{0.04\linewidth}
\vspace*{.15cm}
\bez
 \begin{array}{cc} k & \\ \uparrow & \\ & \rightarrow z \end{array}                
\eez
\end{minipage}
\includegraphics[scale=.25]{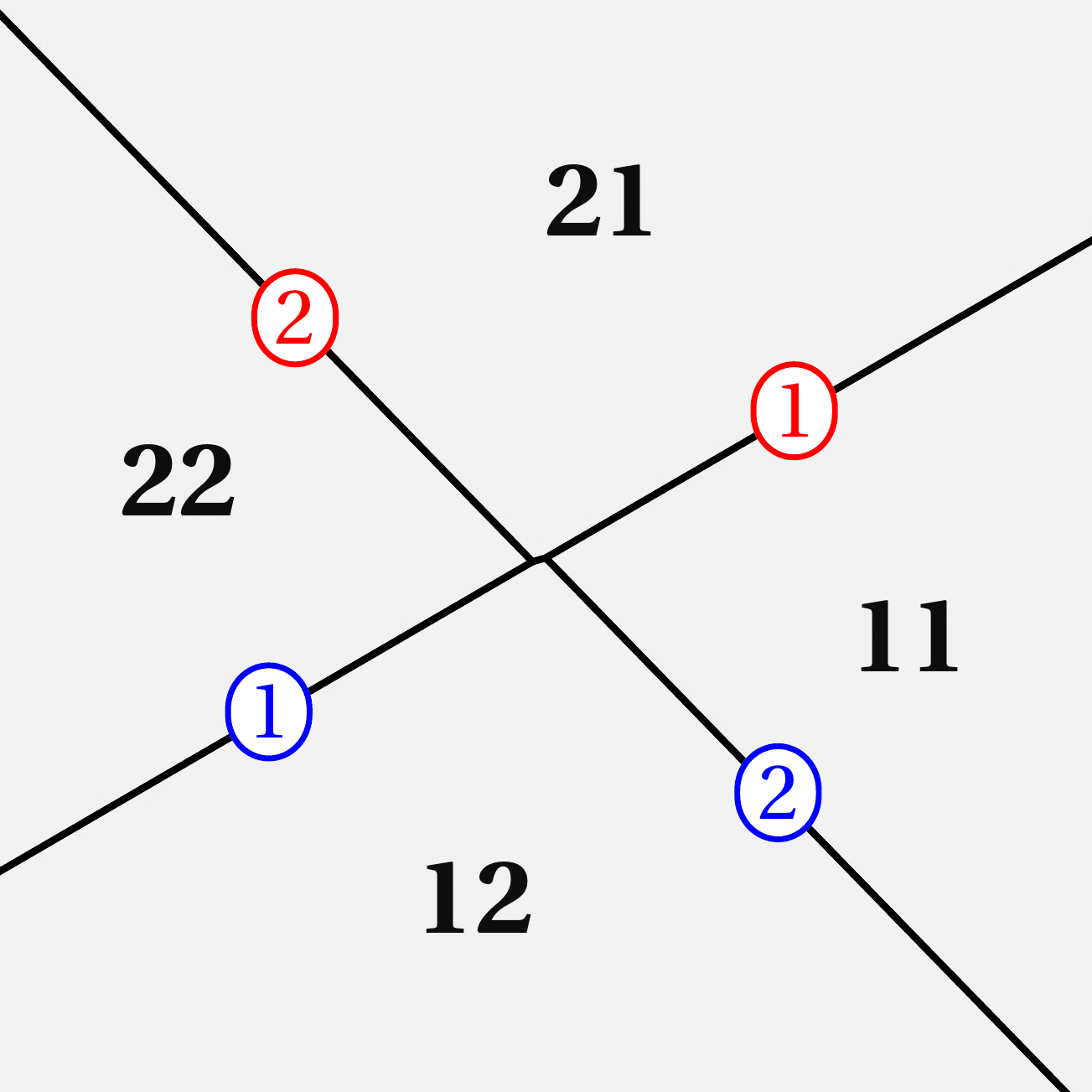} 
\hspace{.5cm}
\includegraphics[scale=.25]{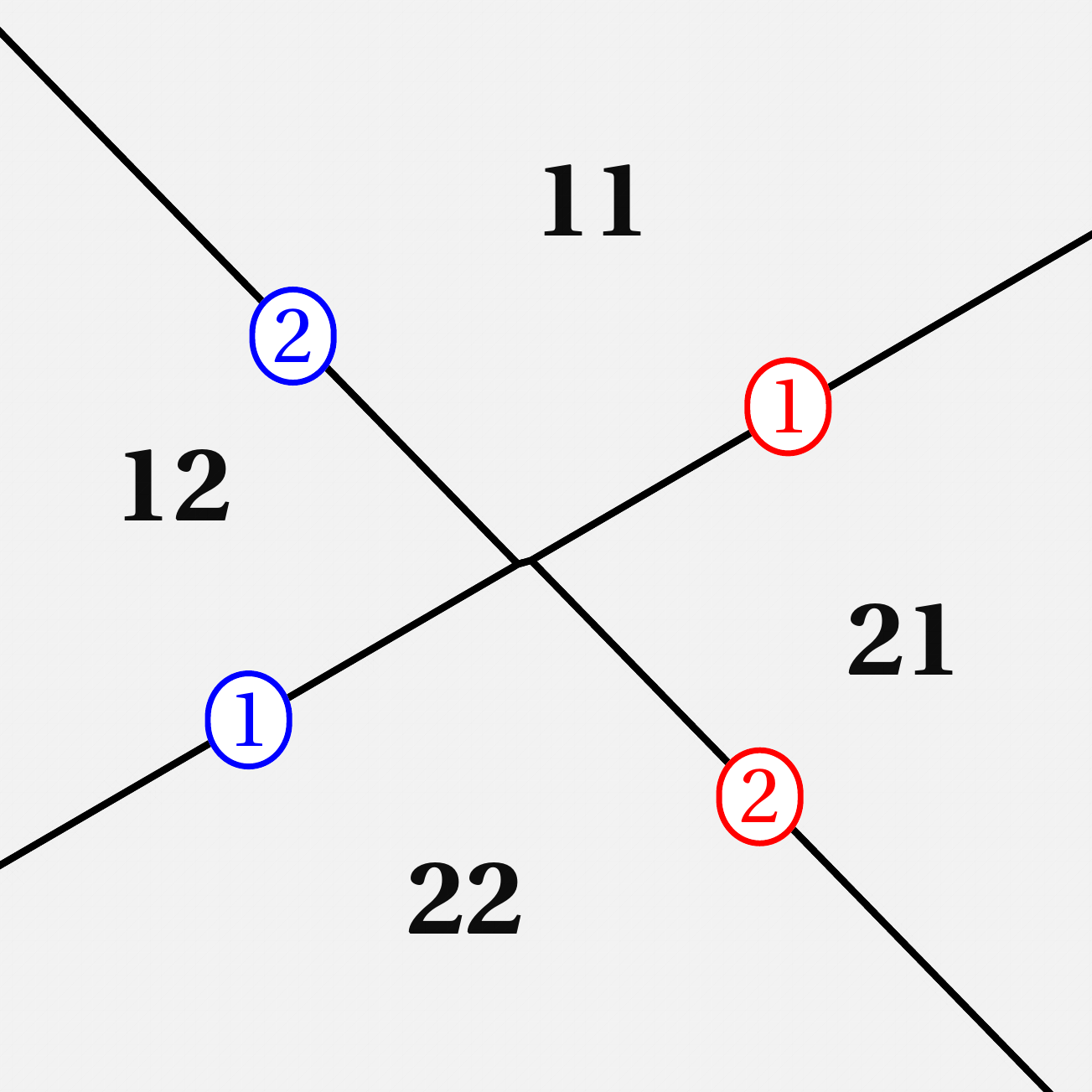}
\hspace{.5cm}
\includegraphics[scale=.25]{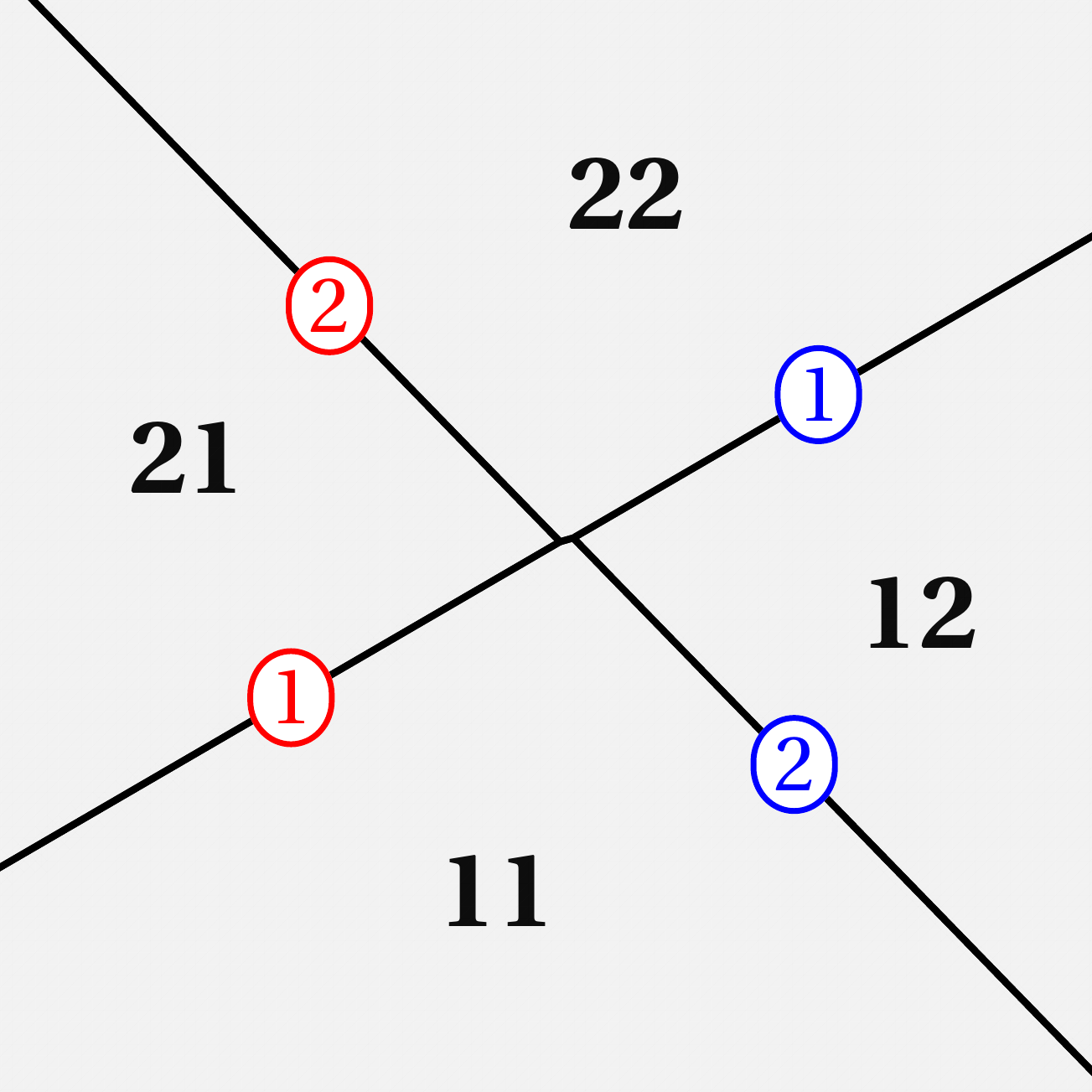} 
\parbox{15cm}{
\caption{Dominating phase regions and tropical limit graph of a pure 2-soliton solution of the vector p2DTL$_K$ equation, 
at $t=x + y=0$, for different values of the parameters $p_i$ and $q_i$, $i=1,2$. See Example~\ref{ex:vec2DToda_2s_other}. 
Viewing the first graph as a process from top to bottom, i.e., mapping the polarizations along the upper two legs to 
those along the lower two, we recover the Yang-Baxter map $\mathcal{R}$. For the second (third) graph,  $\mathcal{R}$
is obtained by viewing it as a process from right (left) to left (right). In the colored figure, this means mapping the 
polarizations along the  red-labeled to those along the blue-labeled soliton lines.  
\label{fig:vec2DToda_2s_other} }
}
\end{center}
\end{figure}

\begin{example}
\label{ex:vec2DToda_2s_other}
We keep the choices for $K$, $\xi_i$ and $\eta_i$, made in Example~\ref{ex:vec2DToda_2s_21}. Then
$\kappa_{ij} =1$ and $\alpha_{12} = (q_2-q_1)/(q_2-p_1)$, so that the regularity conditions (\ref{reg_cond}) 
read $p_2,q_2 > 0$, $p_1 (q_2-q_1)/(q_2-p_1) >0$ and $p_1/q_1 >0$. Furthermore, without loss of generality, 
we can choose the parameters such that, for large enough negative value of $k$, soliton 1 appears in 
$z$-direction to the left of soliton 2.  
\begin{enumerate}
\item $q_1 < p_1 < p_2 < q_2$ or $q_1 < p_2 < q_2 < p_1$. In this case the phase constellation is 
that shown in Fig.~\ref{fig:vec2DToda_2s_21}. Regarding it as a process \emph{in $k$-direction},
the map of polarizations is the Yang-Baxter map $\mathcal{R}$.
\item $p_1 < q_1 < q_2 < p_2$ or $p_1 < q_2 < p_2 < q_1$. The phase constellation is that shown in the first 
plot of Fig.~\ref{fig:vec2DToda_2s_other} (which is generated with $p_1 = 1/4$, $q_1 = 3/2$, $q_2 = 2$ 
and $p_2 = 3$). The map of polarizations, in $k$-direction, leads us to the inverse of the Yang-Baxter 
map $\mathcal{R}$, which is also a Yang-Baxter map. 
\item $q_1 < p_1 < q_2 < p_2$ or $q_1 < q_2 < p_2 < p_1$. The phase constellation is that shown in the 
second plot of Fig.~\ref{fig:vec2DToda_2s_other} (which is generated with $q_1 = 1/4$, $p_1 = 3/2$, 
$q_2 = 2$ and $p_2 = 3$).
Instead of (\ref{choice_of_initial/final_polarizations}), here we define initial and final polarizations as
\bez
    u_1 := \hat{u}_{12,22} \, , \quad
    u_2 := \hat{u}_{21,22} \, , \quad
    u_1' := \hat{u}_{11,21} \, , \quad
    u_2' := \hat{u}_{11,12} \, ,
\eez
and obtain the map $\mathcal{T}$, given by (\ref{X_in->X_out_T}), which is \emph{not} a Yang-Baxter map.
\item $p_1 < q_1 < p_2 < q_2$ or $p_1 < p_2 < q_2 < q_1$. The phase constellation is that shown 
in the third plot of Fig.~\ref{fig:vec2DToda_2s_other} 
(which is generated with $p_1 = 1/4$, $q_1 = 3/2$, $p_2 = 2$ and $q_2 = 3$).
In this case, the map of polarizations, from bottom to top, is the inverse of $\mathcal{T}$. 
\end{enumerate}
For any one of the plots in Figs.~\ref{fig:vec2DToda_2s_21} or \ref{fig:vec2DToda_2s_other}, we 
obtain realizations of \emph{all} the maps by choosing different directions. 
\end{example}

The naive expectation that 2-soliton scattering yields a Yang-Baxter map is therefore wrong. 
But we have to keep in mind that the Yang-Baxter property is a statement about \emph{three} solitons. 
A 3-soliton solution involves three 2-particle interactions. In the tropical limit, this means that 
a composition of three of the above maps carries the polarizations along the tropical limit support.
We will see in the next subsection that the Yang-Baxter equation is indeed only required to hold 
for a mixture of the maps, but not for each map separately. 
 
\begin{remark} 
As seen above, the constellation of dominant phase regions depends on the concrete values of the parameters.
If, for a certain constellation, we select a direction and obtain a Yang-Baxter map, then the latter has the 
Yang-Baxter property for all choices of parameters. From the above, we conclude that the map, relating the (relativ to our 
choice of direction) incoming and outgoing polarizations, is a Yang-Baxter map if and only if the phase 
region that lies between the two incoming solitons is that of $\vartheta_{12}$ or $\vartheta_{21}$. 
\end{remark}

\begin{figure}[h] 
\begin{center}
\begin{minipage}{0.04\linewidth}
\vspace*{.15cm}
\bez
 \begin{array}{cc} t & \\ \uparrow & \\ & \rightarrow k \end{array}                
\eez
\end{minipage}
\includegraphics[scale=.2]{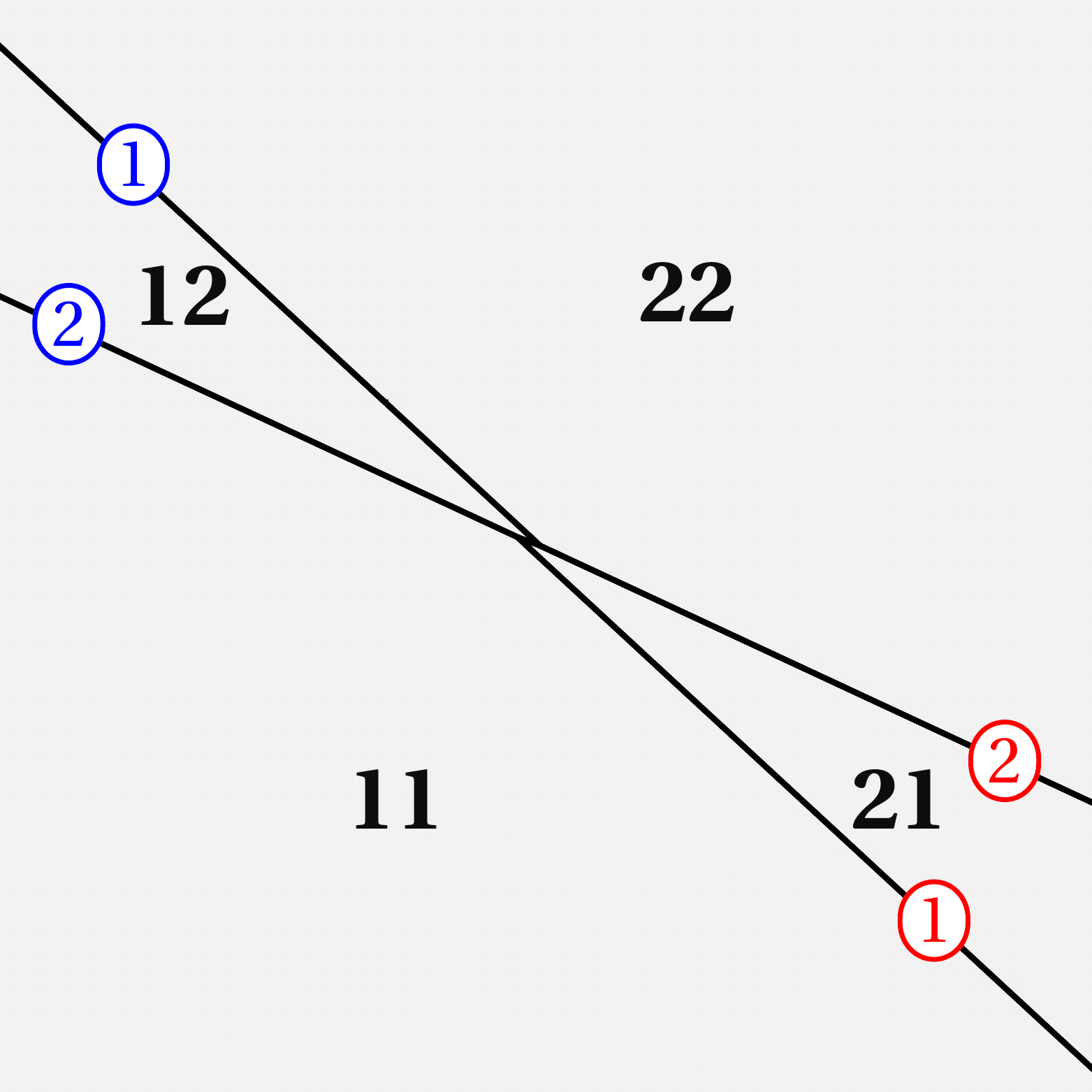} 
\hspace{.5cm}
\includegraphics[scale=.2]{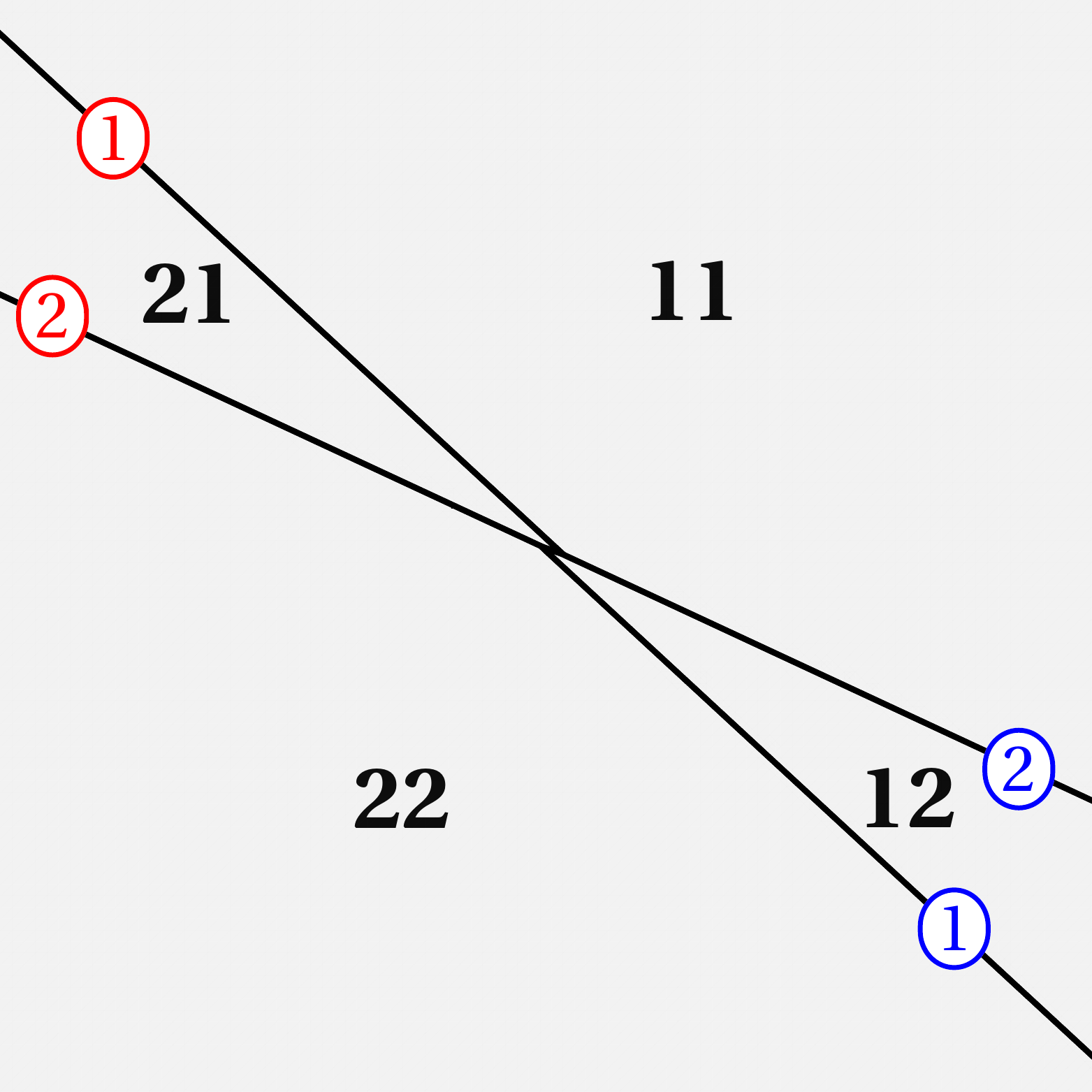} 
\hspace{.5cm}
\includegraphics[scale=.2]{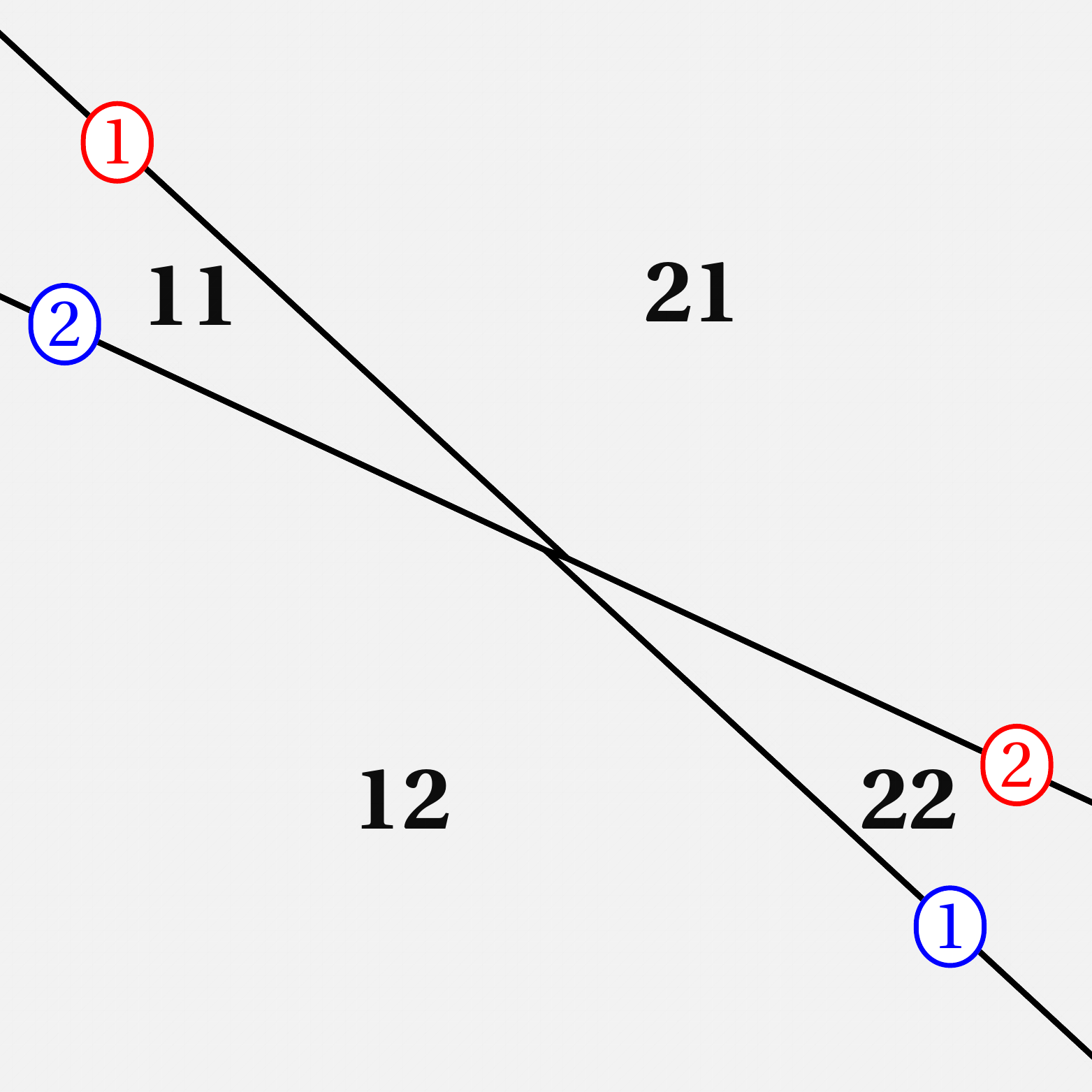} 
\hspace{.5cm}
\includegraphics[scale=.2]{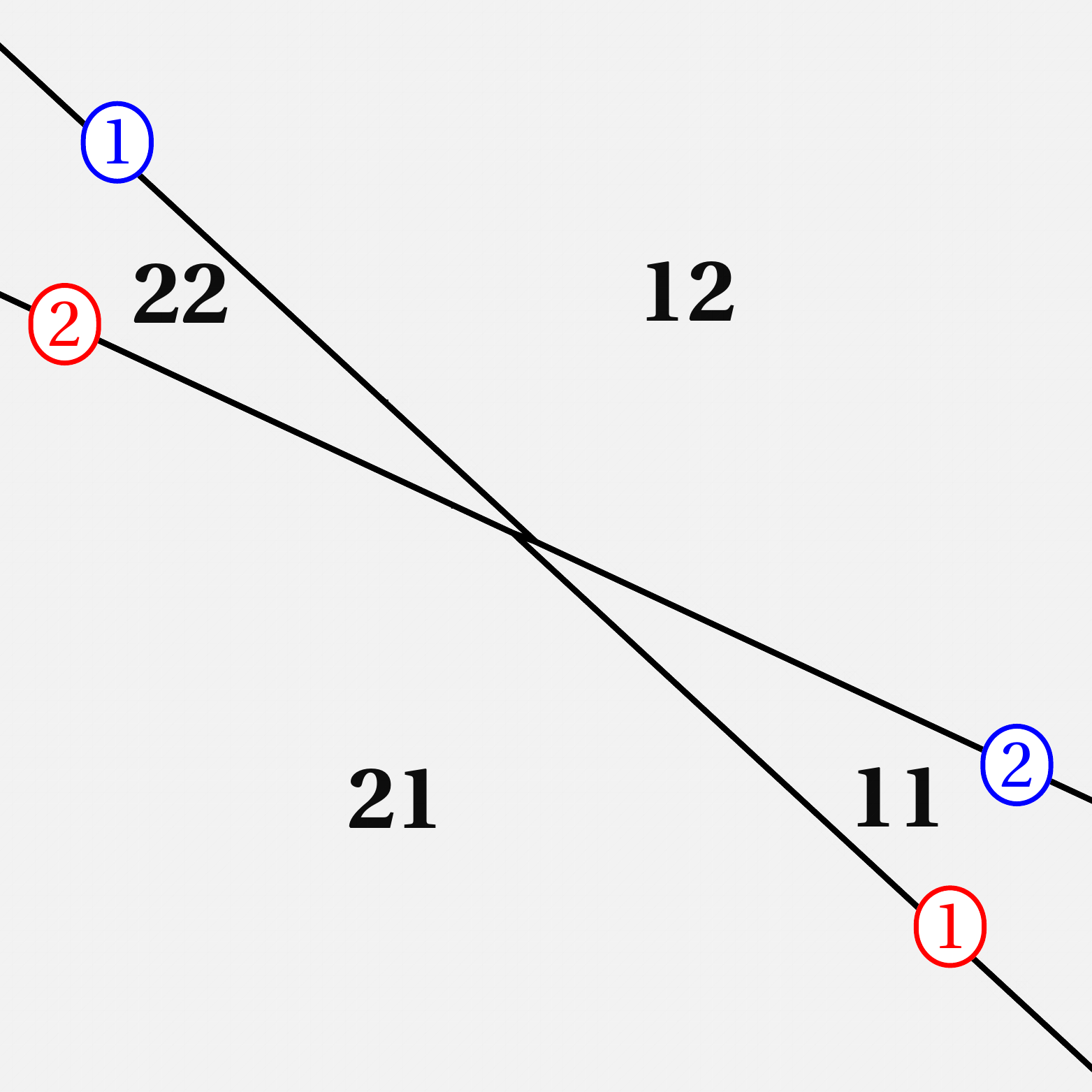}
\parbox{15cm}{
\caption{Dominating phase regions and tropical limit graph of pure 2-soliton solution of the vector p1DTL$_K$ equation, 
with the data given in Example~\ref{ex:1dToda_2s}. The discrete coordinate $k$ is smoothed out. 
\label{fig:1dToda_2s} }
}
\end{center}
\end{figure}

\begin{example}
\label{ex:1dToda_2s}
Imposing the reduction condition $q_i = p_i^{-1}$ on the solutions of the p2DTL$_K$ equation, determines solutions of 
the p1DTL$_K$ equation (\ref{p1DTL_K}). The corresponding Yang-Baxter map is obtained from (\ref{KP_YB_map}) 
by applying this reduction condition. Let us choose again
\bez
  K = \left(\begin{array}{cc} 1 & 1 \end{array} \right) \, , \;
    \xi_1 = \left(\begin{array}{c} 1 \\ 0 \end{array} \right) \, , \;
    \xi_2 = \left(\begin{array}{c} 0 \\ 1 \end{array} \right) \, , \;
    \eta_1 = \eta_2 = 1 \, .
\eez
The regularity condition (\ref{reg_cond}) is then $(p_2-p_1)/(p_2-1) > 0$.
\begin{enumerate}
\item $0 < p_2 <p_1 < 1$. The map of incoming to outgoing polarizations, in $t$-direction, is the reduced Yang-Baxter map.
The tropical limit graph, for parameters $p_1 = 1/2$ and $p_2 = 1/10$ is shown in the first plot
of Fig.~\ref{fig:1dToda_2s}. 
\item $1<p_1<p_2$. In this case, the polarization map is the inverse of the (reduced) Yang-Baxter map.
For $p_1 = 2$ and $p_2 = 10$, we obtain the second plot of Fig.~\ref{fig:1dToda_2s}.
\item $0 < p_1 < 1 < p_2$ and $p_1<p_2^{-1}$. In this case, the (reduced) map $\mathcal{T}$ is at work. For $p_1=1/2$ 
and $p_2=10$, the third plot of Fig.~\ref{fig:1dToda_2s} is obtained.
\item $0 < p_2 < 1 < p_1$ and $p_1^{-1}<p_2$. Here $\mathcal{T}^{-1}$ applies. For $p_1=2$ 
and $p_2=1/10$, we obtain the fourth plot of Fig.~\ref{fig:1dToda_2s}.
\end{enumerate}
\end{example}

\subsection{Pure 3-soliton solution}
\label{subsec:3sol}
For $N=3$ we find 
\bez
   \tau &=& \kappa_{11} \kappa_{22} \kappa_{33} \, \beta \, e^{\vartheta_{111}} 
          + \kappa_{11} \kappa_{22} \, \alpha_{12} \, e^{\vartheta_{112}}
          + \kappa_{11} \kappa_{33} \, \alpha_{13} \, e^{\vartheta_{121}}   \nonumber \\
        &&  + \kappa_{22} \kappa_{33} \, \alpha_{23} \, e^{\vartheta_{211}}
          + \kappa_{11} \, e^{\vartheta_{122}}
          + \kappa_{22} \, e^{\vartheta_{212}}
          + \kappa_{33} \, e^{\vartheta_{221}}
          + e^{\vartheta_{222}} \, , 
\eez
where 
\bez
    \alpha_{ij} &=& 1- \frac{(p_i-q_i)(p_j-q_j)}{(p_i-q_j)(p_j-q_i)} \, 
                   \frac{\kappa_{ij} \kappa_{ji}}{\kappa_{ii} \kappa_{jj}} \, , \\
    \beta &=& -2 + \alpha_{12} + \alpha_{13} + \alpha_{23} 
       + \frac{(p_1-q_1)(p_2-q_2)(p_3-q_3)}{(p_1-q_3)(p_2-q_1)(p_3-q_2)} \, 
                   \frac{\kappa_{12} \kappa_{23} \kappa_{31}}{\kappa_{11} \kappa_{22} \kappa_{33}} \\
     && + \frac{(p_1-q_1)(p_2-q_2)(p_3-q_3)}{(p_1-q_2)(p_2-q_3)(p_3-q_1)} \, 
                   \frac{\kappa_{13} \kappa_{21} \kappa_{32}}{\kappa_{11} \kappa_{22} \kappa_{33}} \, ,
\eez
and 
\bez
    && \vartheta_{111} = \vartheta(p_1)^+ + \vartheta(p_2)^+ + \vartheta(p_3)^+ \, , \quad
      \vartheta_{112} = \vartheta(p_1)^+ + \vartheta(p_2)^+ + \vartheta(q_3) \, , \\
    && \vartheta_{121} = \vartheta(p_1)^+ + \vartheta(q_2) + \vartheta(p_3)^+ \, , \quad
       \vartheta_{211} = \vartheta(q_1) + \vartheta(p_2)^+ + \vartheta(p_3)^+ \, , \\ 
    && \vartheta_{122} = \vartheta(p_1)^+ + \vartheta(q_2) + \vartheta(q_3) \, , \quad
       \vartheta_{212} = \vartheta(q_1) + \vartheta(p_2)^+ + \vartheta(q_3) \, , \\
    && \vartheta_{221} = \vartheta(q_1) + \vartheta(q_2) + \vartheta(p_3)^+ \, , \quad
       \vartheta_{222} = \vartheta(q_1) + \vartheta(q_2) + \vartheta(q_3) \, .
\eez
Again, we set $\kappa_{ij} = \eta_i K \xi_j$. Furthermore, we have
\bez
   F &=& \left(p_1-q_1\right) \left(p_2-q_2\right) \left(p_3-q_3\right) \Big( 
     \frac{\alpha_{12} \kappa_{11} \kappa_{22}}{\left(p_1-q_1\right) \left(p_2-q_2\right)} \xi _3 \otimes \eta _3
      + \frac{\alpha_{13} \kappa_{11} \kappa_{33}}{\left(p_1-q_1\right) \left(p_3-q_3\right)} \xi _2\otimes \eta _2 \\
     && + \frac{\alpha_{23} \kappa_{22} \kappa_{33}}{\left(p_2-q_2\right) \left(p_3-q_3\right)} \xi _1\otimes \eta _1
      - \frac{\alpha_{123} \, \kappa_{11} \kappa_{23}}{\left(p_1-q_1\right) \left(p_3-q_2\right)} \xi _2\otimes \eta _3
      - \frac{\alpha_{132} \, \kappa_{11} \kappa_{32}}{\left(p_1-q_1\right) \left(p_2-q_3\right)} \xi _3\otimes \eta _2 \\
     && - \frac{\alpha_{213} \, \kappa_{13} \kappa_{22}}{\left(p_3-q_1\right) \left(p_2-q_2\right)} \xi _1\otimes \eta _3
      - \frac{\alpha_{231} \, \kappa_{22} \kappa_{31}}{\left(p_2-q_2\right) \left(p_1-q_3\right)} \xi _3\otimes \eta _1
      - \frac{\alpha_{312} \, \kappa_{12} \kappa_{33}}{\left(p_2-q_1\right) \left(p_3-q_3\right)} \xi _1\otimes \eta _2 \\
     && - \frac{\alpha_{321} \, \kappa_{21} \kappa_{33}}{\left(p_1-q_2\right) \left(p_3-q_3\right)} \xi _2\otimes \eta _1 \Big)
      \, e^{\vartheta_{111}}  \\
     && + \left(p_1-q_1\right) \left(p_2-q_2\right) \Big( \frac{\kappa_{11}}{p_1-q_1} \xi _2\otimes \eta _2
        - \frac{\kappa_{12}}{p_2-q_1} \xi _1\otimes \eta _2 - \frac{\kappa_{21}}{p_1-q_2} \xi _2\otimes \eta _1 
        + \frac{\kappa_{22}}{p_2-q_2} \xi _1\otimes \eta _1 \Big) e^{\vartheta_{112}} \\
     && + \left(p_1-q_1\right) \left(p_3-q_3\right) \Big( \frac{\kappa_{11}}{p_1-q_1} \xi _3\otimes \eta _3
         - \frac{\kappa_{13}}{p_3-q_1} \xi _1\otimes \eta _3 - \frac{\kappa_{31}}{p_1-q_3} \xi _3\otimes \eta _1
         + \frac{\kappa_{33}}{p_3-q_3} \xi _1\otimes \eta _1 \Big) e^{\vartheta_{121}} \\
     && + \left(p_2-q_2\right) \left(p_3-q_3\right) \Big( \frac{\kappa_{22}}{p_2-q_2} \xi _3\otimes \eta _3 
          - \frac{\kappa_{23}}{p_3-q_2} \xi _2\otimes \eta _3 -\frac{\kappa_{32}}{p_2-q_3} \xi _3\otimes \eta _2 
          + \frac{\kappa_{33}}{p_3-q_3} \xi _2\otimes \eta _2 \Big) e^{\vartheta_{211}} \\
     && + \left(p_1-q_1\right)  \xi _1\otimes \eta _1 \, e^{\vartheta_{122}}
        + \left(p_2-q_2\right) \xi _2\otimes \eta _2 \, e^{\vartheta_{212}} 
        + \left(p_3-q_3\right) \xi _3\otimes \eta _3 \, e^{\vartheta_{221}} \, ,  
\eez
where
\bez
    \alpha_{kij} = 1 
     - \frac{(p_j-q_i)(p_k-q_k) \kappa_{ik} \, \kappa_{kj}}{(p_k-q_i)(p_j-q_k) \kappa_{ij} \, \kappa_{kk}} \, .
\eez
Note that $\alpha_{ij} = \alpha_{ijj}$. The above expressions coincide with those obtained in the KP$_K$ case \cite{DMH19LMP}. 
The only difference is in the phases entering the exponentials. Here we wrote $F$ in a more compact form.

\begin{example}
\label{ex:3s_1}
We choose
\bez
  &&  K = \left(\begin{array}{cc} 1 & 1 \\  1 & 1 \\ 1 & 2 \end{array} \right) \, , \quad
    \xi_1 = \left(\begin{array}{c} 1 \\ 0 \end{array} \right) \, , \quad
    \xi_2 = \left(\begin{array}{c} 0 \\ 1 \end{array} \right) \, , \quad
    \xi_3 = \left(\begin{array}{c} 1 \\ 1 \end{array} \right) \, , \\
  &&  \eta_1 = \left(\begin{array}{ccc} 1 & 0 & 0 \end{array}\right) \, , \quad
    \eta_2 = \left(\begin{array}{ccc} 0 & 1 & 0 \end{array}\right) \, , \quad
    \eta_3 = \left(\begin{array}{ccc} 0 & 0 & 1 \end{array}\right) \, , \\
  &&  p_1=1, \, p_2=1/4, \, p_3=3/2 \, , q_1 = 1/2 \, , q_2=6/5 \, , q_3 =3 \, .
\eez
Fig.~\ref{fig:3soliton_1} displays the tropical limit graphs of the corresponding 3-soliton solution at a large 
negative and a large positive value of $t$. 
The respective sequences of interactions correspond to the two sides of the Yang-Baxter equation. Since 
the polarizations do not depend on the variables $z,t,k$, we conclude that, starting with the same initial polarizations, 
in both cases we end up with the same polarizations. This implies that the maps acting along the tropical limit 
graph satisfy a mixed version of the Yang-Baxter equation. 
\begin{figure}[h] 
\begin{center}
\begin{minipage}{0.04\linewidth}
\vspace{-3cm}
\bez
 \begin{array}{cc} k & \\ \uparrow & \\ & \rightarrow z \end{array}                
\eez
\end{minipage}
\hspace{2cm}
\includegraphics[scale=.2]{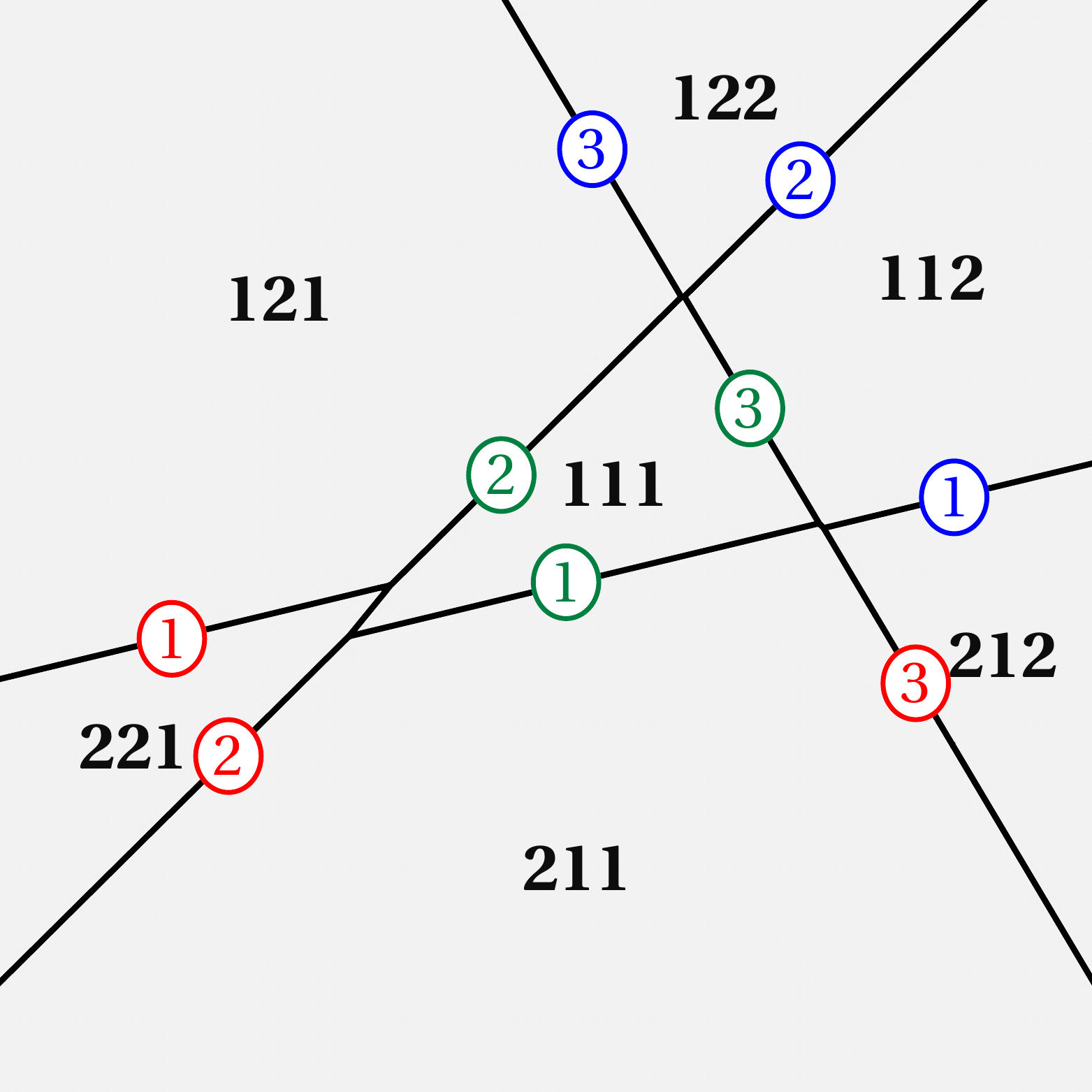}
\hspace{2.cm}
\includegraphics[scale=.2]{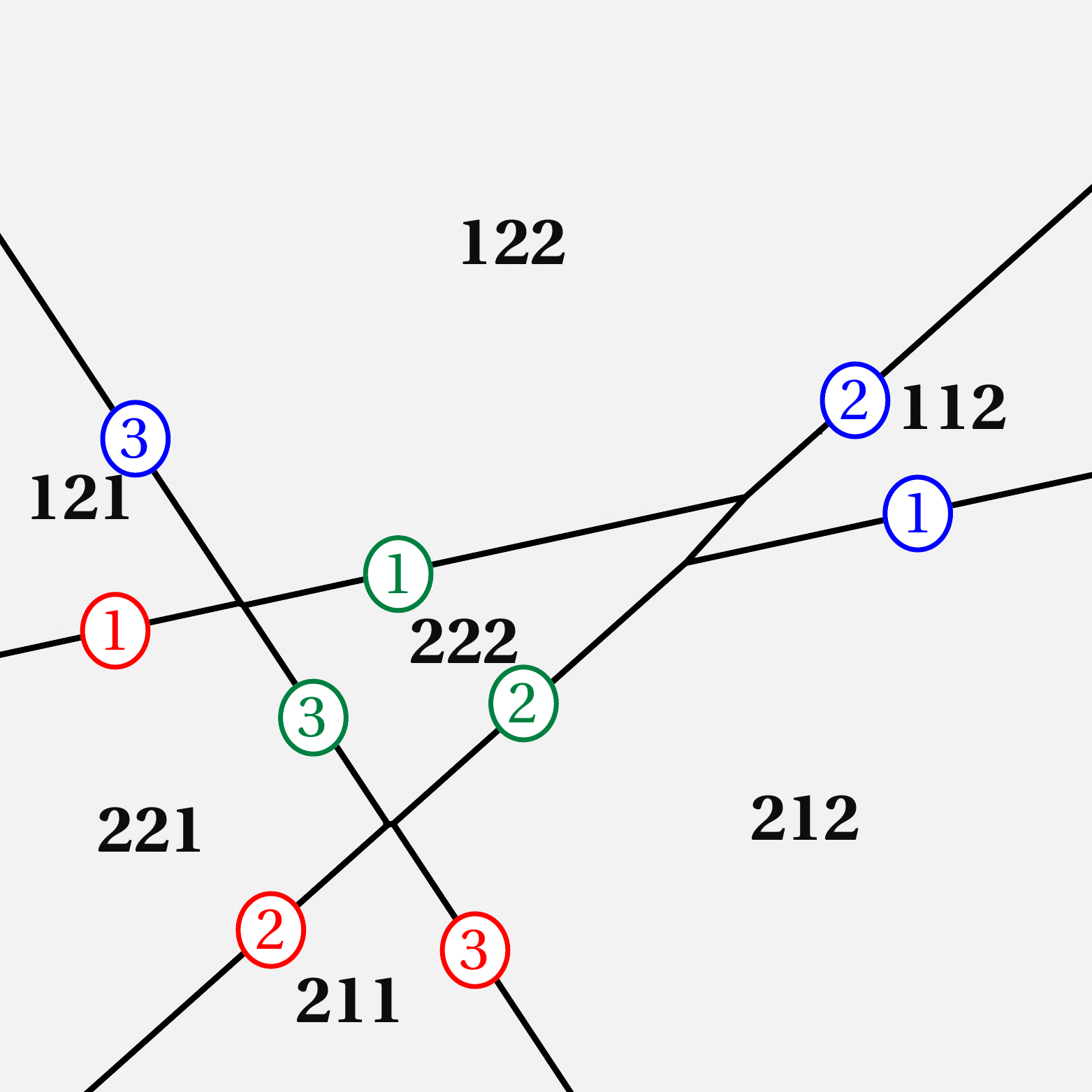}
\parbox{15cm}{
\caption{Dominating phase regions and tropical limit graphs of the pure 3-soliton solution of the p2DTL$_K$ equation 
with the data given in Example~\ref{ex:3s_1}. 
The left graph is obtained for $t=-50$. Regarding it as an evolution of three particles (solitons in the tropical limit) 
in $k$-direction, first particle 1 meets particle 2, then particle 1 and particle 3 meet, and finally particle 2 
interacts with particle 3. This corresponds to one side of the (mixed) Yang-Baxter equation. 
The right graph is obtained for $t=50$ and the sequence of interactions corresponds to the other side of the (mixed)
Yang-Baxter equation. 
\label{fig:3soliton_1} }
}
\end{center}
\end{figure}
\end{example}

The polarizations at a node of the tropical limit graph at $t = t_0$ are completely determined 
if we know the soliton numbers associated with two ``incoming lines'', e.g. soliton 1 with the left 
and soliton 2 with the right line, and the number of the enclosed phase region, for example $222$.
In this configuration, to the left of incoming 
soliton line 1 we have the phase with number $122$,  to the right of incoming 
soliton line 2 the phase with number $212$, and the remaining one has necessarily number $112$. Hence there are 
incoming polarizations $\hat{u}_{122,222},\hat{u}_{212,222}$ and outgoing polarizations $\hat{u}_{112,212},\hat{u}_{112,122}$ 
of solitons $1,2$. These can be computed from the above 3-soliton solution and it can be verified that they are related 
by one of the maps considered in Section~\ref{subsec:other_maps}. This also holds for all possible other nodes.  

In the following, we consider the situation shown in the left plot in Fig.~\ref{fig:3soliton_1}, regarding it as a 
process from bottom to top. Accordingly, we set
\bez
  &&  u_{1,\mathrm{in}} = \hat{u}_{121,221} \, , \quad
    u_{1,\mathrm{m1}} = \hat{u}_{111,211} \, , \quad
    u_{1,\mathrm{out}} = \hat{u}_{111,212} \, , \\
  && u_{2,\mathrm{in}} = \hat{u}_{221,211} \, , \quad
    u_{2,\mathrm{m1}} = \hat{u}_{121,111} \, , \quad
    u_{2,\mathrm{out}} = \hat{u}_{122,112} \, , \\  
  && u_{3,\mathrm{in}} = \hat{u}_{211,212} \, , \quad
    u_{3,\mathrm{m1}} = \hat{u}_{111,112} \, , \quad
    u_{3,\mathrm{out}} = \hat{u}_{121,122} \, ,  
\eez
where a subscript $\mathrm{m1}$ indicates that the line segment appears in the middle of the first plot in 
Fig.~\ref{fig:3soliton_1}. 
For the right plot in Fig.~\ref{fig:3soliton_1}, we only have to replace the polarizations of middle segments by
\bez
   u_{1,\mathrm{m2}} = \hat{u}_{122,222} \, , \quad 
   u_{2,\mathrm{m2}} = \hat{u}_{222,212} \, , \quad
   u_{3,\mathrm{m2}} = \hat{u}_{221,222} \, .
\eez

Now we can check that either the Yang-Baxter map $\mathcal{R}$, the map $\mathcal{T}$, or one of their inverses 
acts at each crossing of the tropical limit graph. 
Proceeding from bottom to top in the left plot of Fig.~\ref{fig:3soliton_1}, the first crossing involves only 
solitons 1 and 2, so we can ignore the last digit of the numbers of the involved phases. The situation is  
that of the 2-soliton interaction sketched in the first drawing of Fig.~\ref{fig:crossings}. Accordingly, 
 $\mathcal{T}$, given by (\ref{X_in->X_out_T}), should map 
$(u_{1,\mathrm{in}}, u_{2,\mathrm{in}})$ to $(u_{1,\mathrm{m1}},u_{2,\mathrm{m1}})$. 
Indeed, this can be varified using the above data. 

\begin{figure}[h]
\begin{center}
\includegraphics[scale=.2]{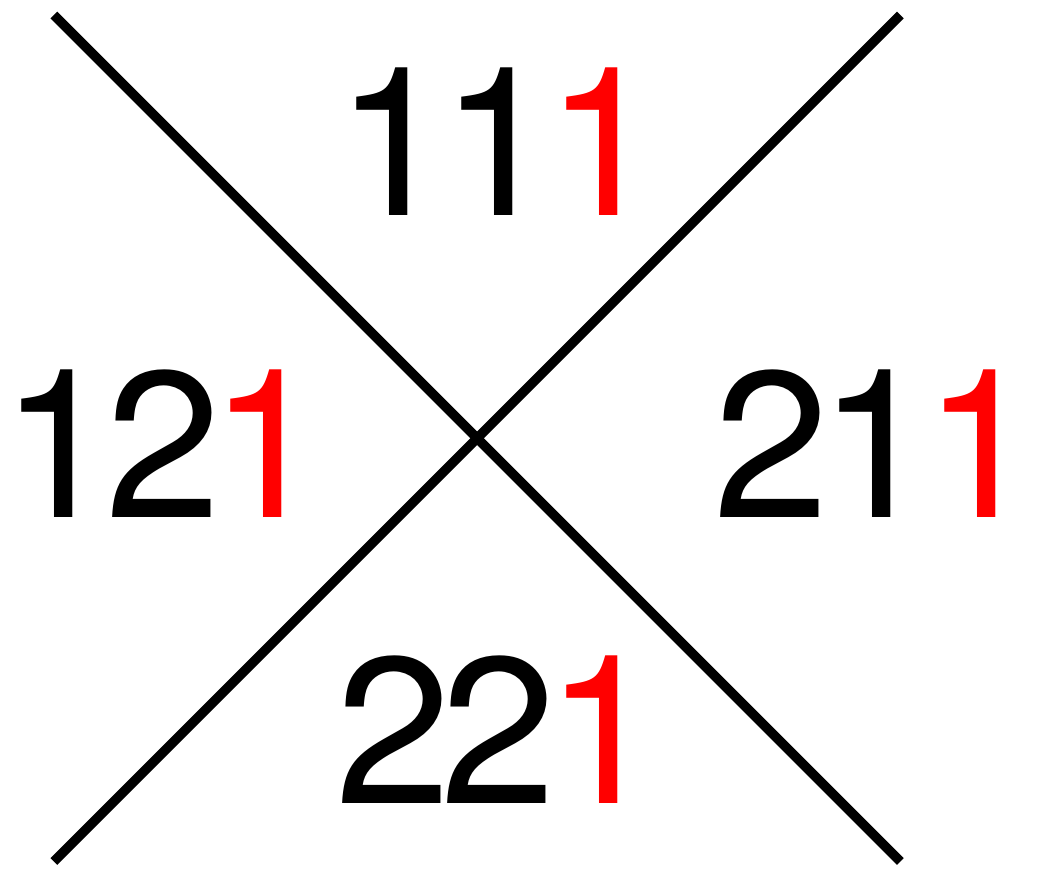}
\hspace{1cm}
\includegraphics[scale=.2]{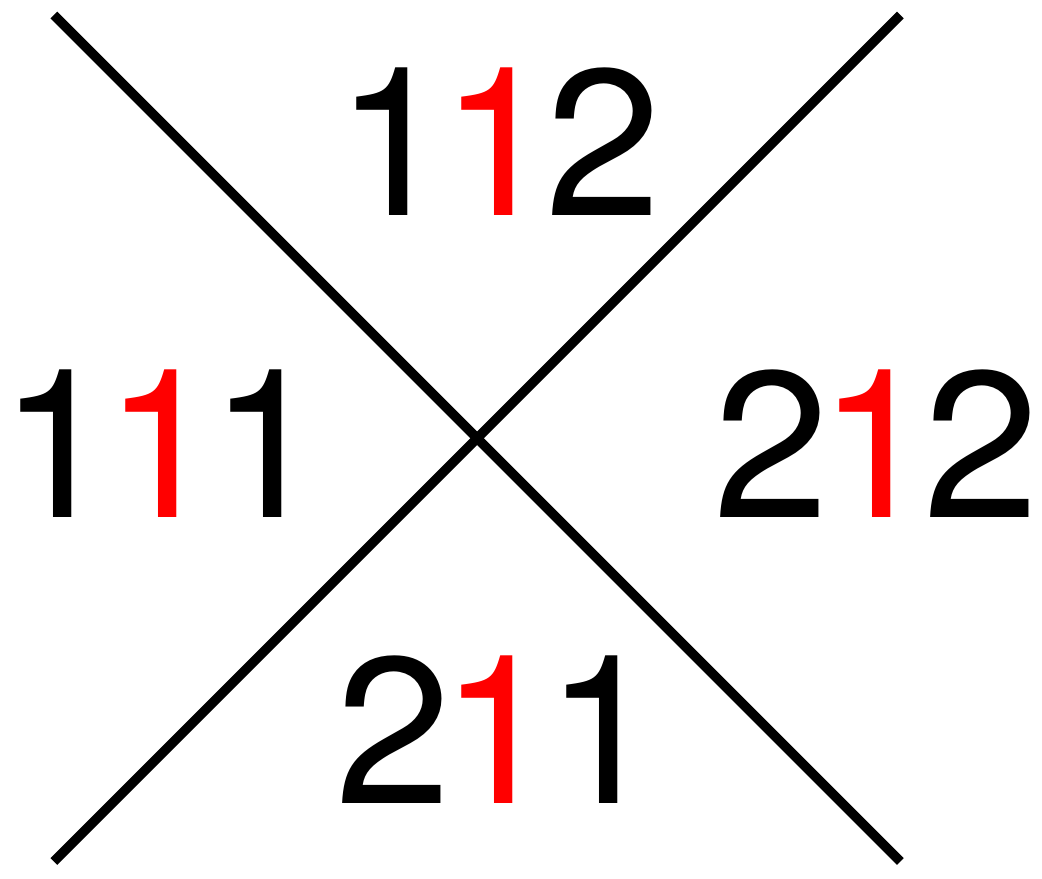}
\hspace{1cm}
\includegraphics[scale=.2]{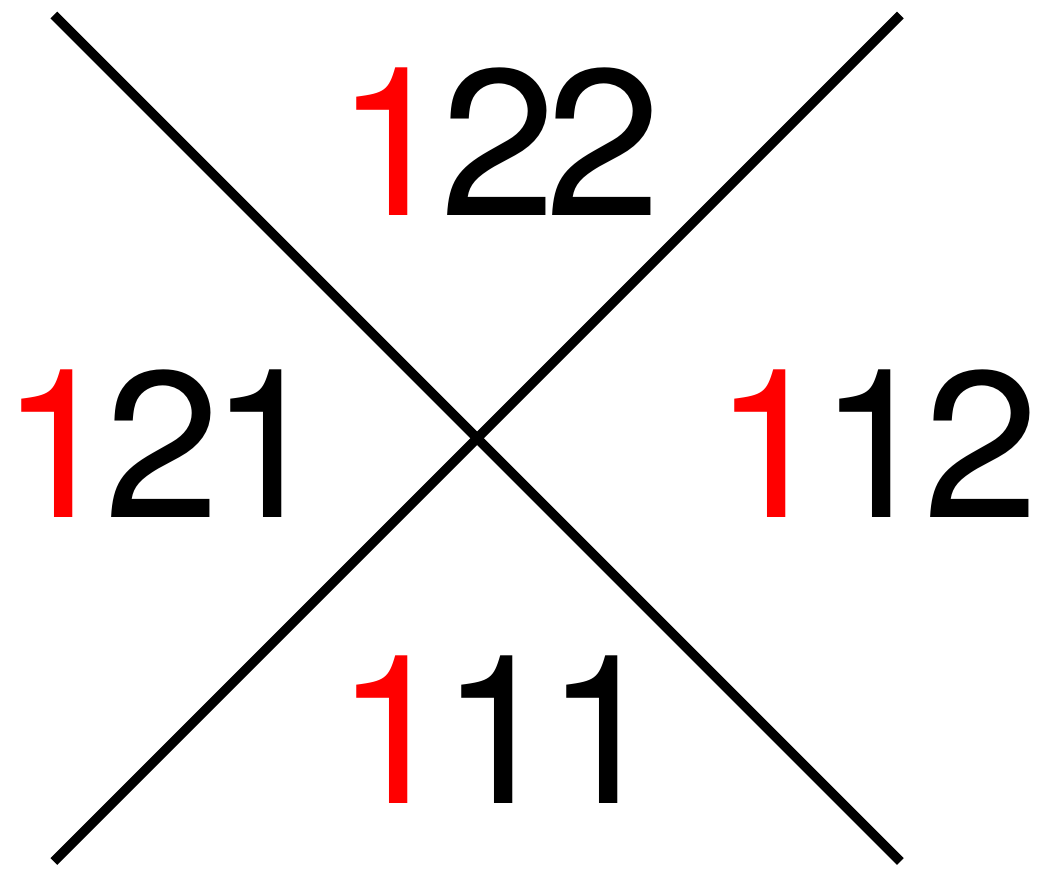}
\parbox{15cm}{
\caption{The dominating phase structures around the three crossings of solitons in the left plot of 
Fig.~\ref{fig:3soliton_1}, proceeding from bottom to top. In the phase numbers we marked (with red color) 
the digit corresponding to the soliton that does \emph{not} take part in the respective interaction. 
Disregarding this digit, the drawing describes the phase constellation of a 2-soliton interaction. 
In this way, for example, the first drawing corresponds to an application of $\mathcal{T}$.
\label{fig:crossings} }
}
\end{center}
\end{figure}

The next crossing, where solitons 1 and 3 interact, is sketched as a 2-soliton interaction in the second drawing 
of Fig.~\ref{fig:crossings}. In this situation, the Yang-Baxter map $\mathcal{R}$  
should yield $(u_{1,\mathrm{m1}},u_{3,\mathrm{in}}) \mapsto (u_{1,\mathrm{out}},u_{3,\mathrm{m1}})$. 
Indeed, we find
\bez
    u_{1,\mathrm{out}} = \frac{B_3(p_1,u_{3,\mathrm{in}}) \, u_{1,\mathrm{m1}} \, \tilde{A}_3(q_1,u_{3,\mathrm{in}})}
        {\mathrm{tr}[B_3(p_1,u_{3,\mathrm{in}}) \, u_{1,\mathrm{m1}} \, \tilde{A}_3(q_1,u_{3,\mathrm{in}}) \, K ]} \, .
\eez
(which can be deduced from (\ref{xi,eta_out_from_m1})).
Similarly,
\bez
    u_{3,\mathrm{m1}} = \frac{A_1(q_3,u_{1,\mathrm{m1}}) \, u_{3,\mathrm{in}} \, \tilde{B}_1(p_3,u_{1,\mathrm{m1}})}
        {\mathrm{tr}[A_1(q_3,u_{1,\mathrm{m1}}) \, u_{3,\mathrm{in}} \, \tilde{B}_1(p_3,u_{1,\mathrm{m1}}) \, K ]} \, .
\eez

Finally, at the last crossing solitons 2 and 3 interact. The situation is sketched in the last drawing in 
Fig.~\ref{fig:crossings}. Accordingly, we expect the map $\mathcal{T}^{-1}$ to be at work and this can 
indeed be verified.
\vspace{.2cm}

\begin{figure}[h]
\begin{center}
\includegraphics[scale=.2]{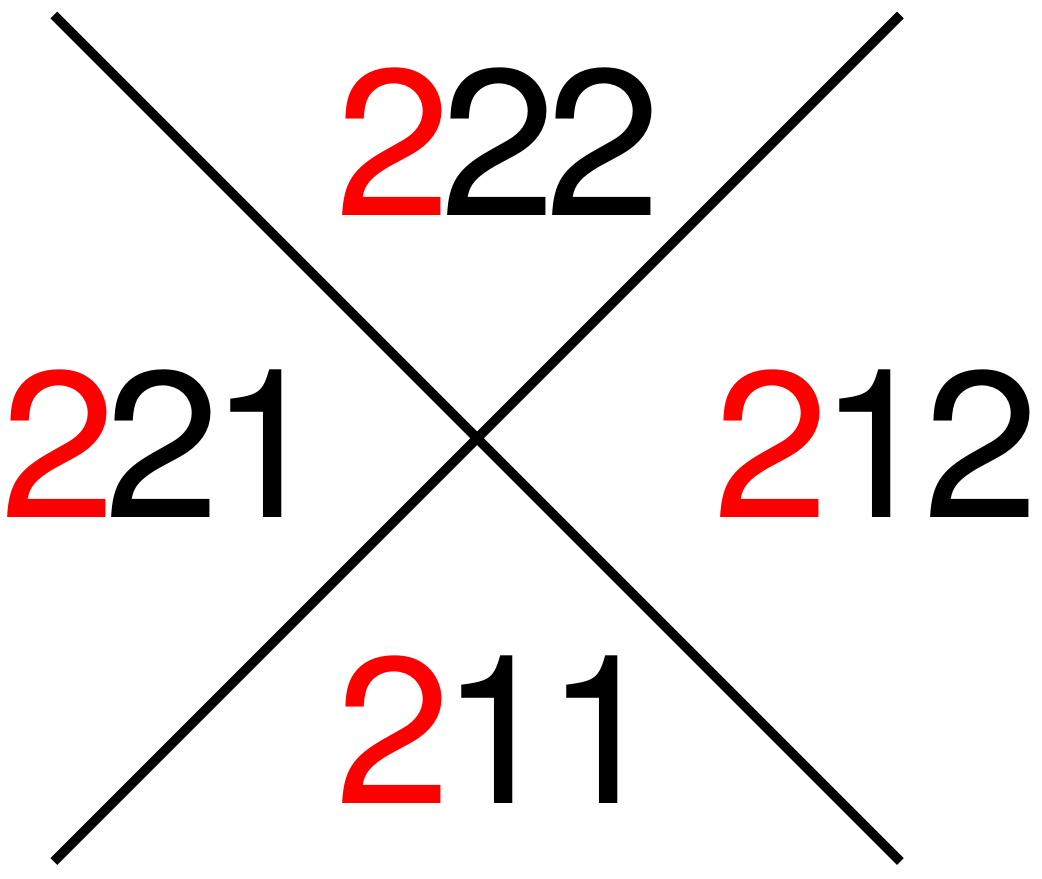}
\hspace{1cm}
\includegraphics[scale=.2]{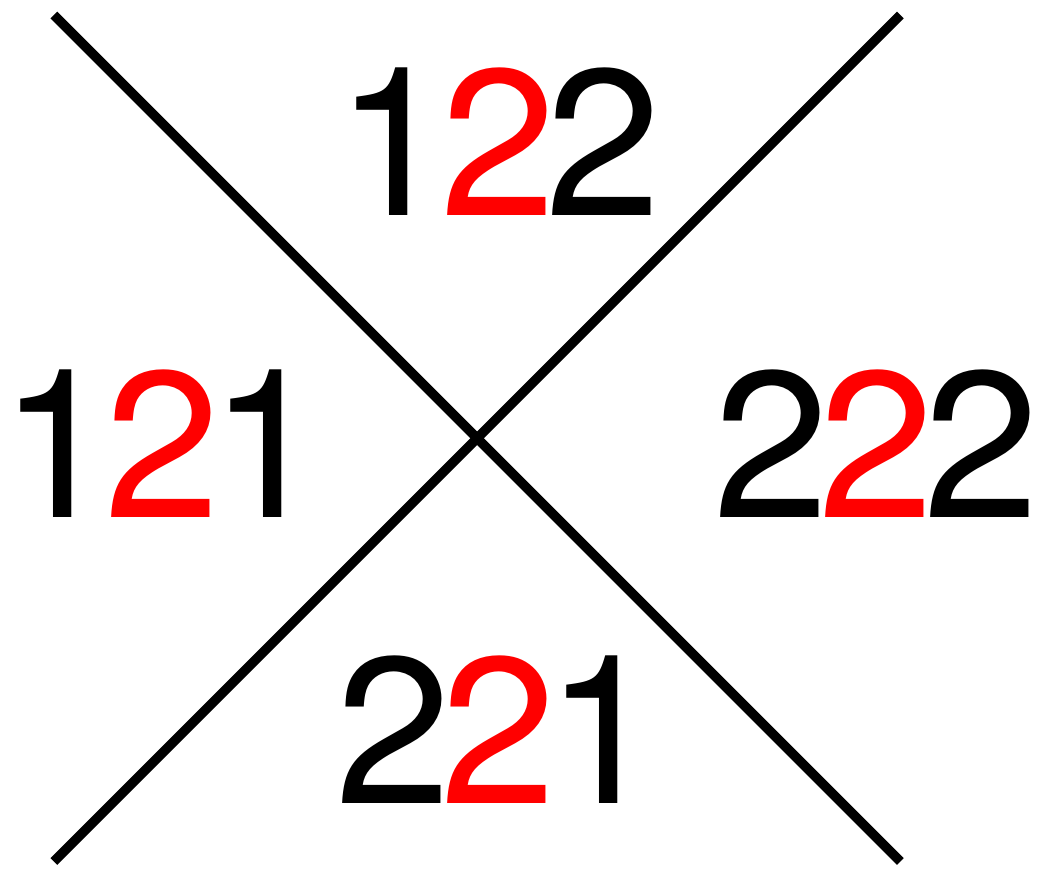}
\hspace{1cm}
\includegraphics[scale=.2]{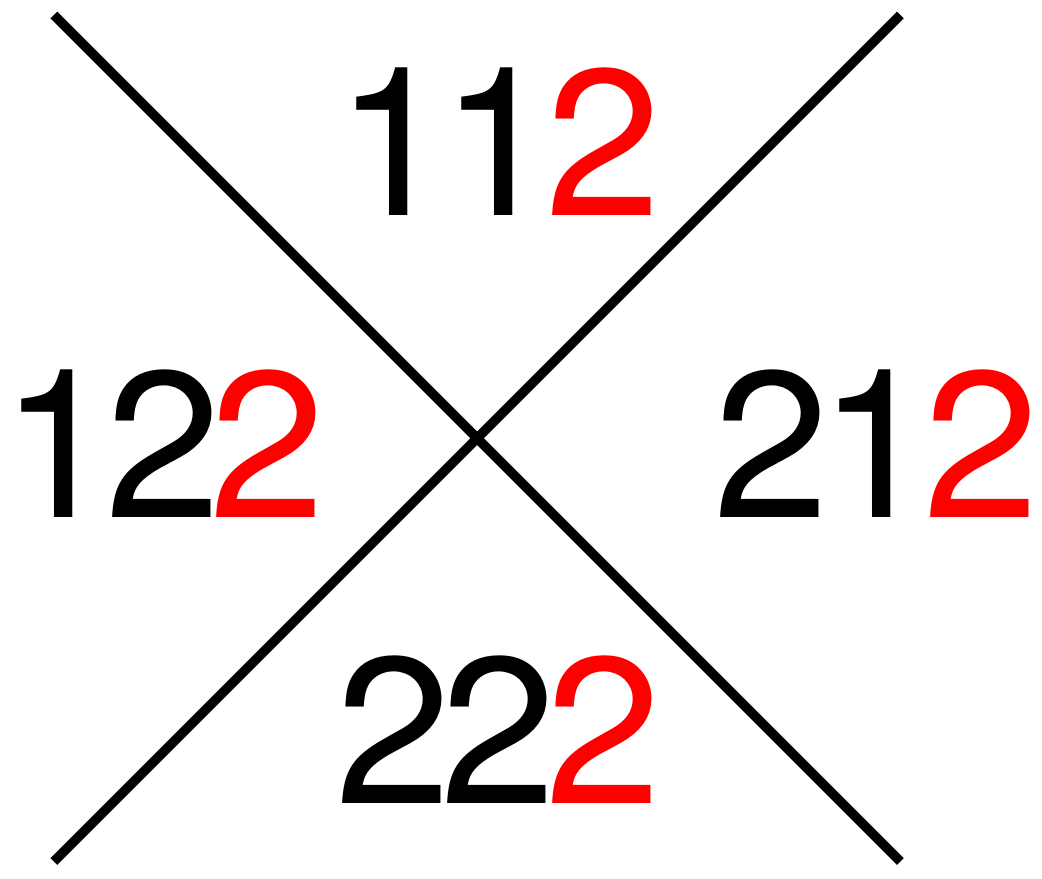}
\parbox{15cm}{
\caption{The dominating phase structures around the three crossings of solitons in the right plot of 
Fig.~\ref{fig:3soliton_1}, proceeding from bottom to top. 
\label{fig:crossings_r} }
}
\end{center}
\end{figure}

The 2-soliton subinteractions appearing in the right plot of Fig.~\ref{fig:3soliton_1} are sketched in 
Fig.~\ref{fig:crossings_r}. For the lowest, corresponding to the first drawing in Fig.~\ref{fig:crossings_r}, 
we can verify (see Appendix~\ref{app:comp_details}) that the polarizations are related by the map $\mathcal{T}^{-1}$, 
given by (\ref{X_in->X_out_alt_inverse}).
The second drawing in Fig.~\ref{fig:crossings_r} corresponds to an application of the Yang-Baxter map 
$\mathcal{R}$, the third to an application of $\mathcal{T}$.

At least for the parameter range, for which the tropical limit graphs at large negative, respectively positive 
time $t$ have the structure shown in Fig.~\ref{fig:3soliton_1}, we can now conclude that (\ref{YB_mixed}) holds.
This is so because the polarizations associated with line segments of the tropical limit graph do not depend on 
the variables $z,t,k$. Since the two situations shown in Fig.~\ref{fig:3soliton_1} belong to the same solution of 
the 2DTL$_K$ equation, the result of the application of the two sides of (\ref{YB_mixed}) are the same. We know  
that (\ref{YB_mixed}) indeed holds for all parameter values (provided that $p_i,q_i$ are all different). 

In the following examples, we present tropical limit graphs for pure 3-soliton solutions, which show a structure different 
from that in Fig.~\ref{fig:3soliton_1}.

\begin{example}
\label{ex:3s_2}
We choose $K$ and $\eta_i,\xi_i$, $i=1,2,3$, as in Example~\ref{ex:3s_1}, and set
\bez
     p_1=1/2, \, p_2=10, \, p_3=2 \, , q_1 = 1 \, , q_2=20 \, , q_3 =4 \, .
\eez
Corresponding tropical limit graph at constant values of $t$ are shown in Fig.~\ref{fig:3soliton_2}, from 
which we read off
\bez
      \mathcal{T}^{-1}_{\bsy{12}}(1,2) \circ \mathcal{T}^{-1}_{\bsy{13}}(1,3) \circ \mathcal{R}_{\bsy{23}}(2,3)
    =  \mathcal{R}_{\bsy{23}}(2,3) \circ \mathcal{T}^{-1}_{\bsy{13}}(1,3) \circ \mathcal{T}^{-1}_{\bsy{12}}(1,2) 
       \, , 
\eez
which is (\ref{RTinvTinv=TinvTinvR}). 
\begin{figure}[h] 
\begin{center}
\includegraphics[scale=.2]{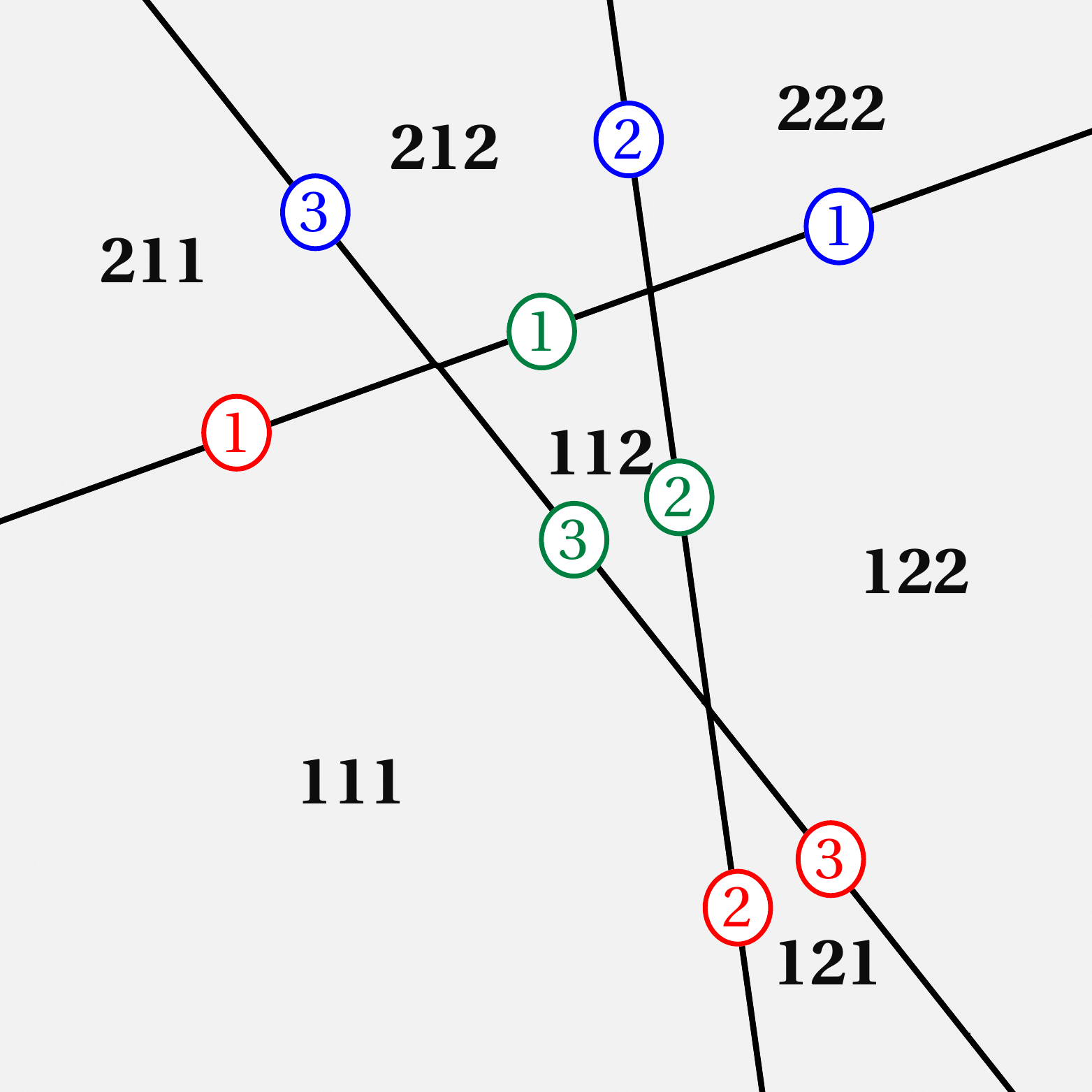}
\hspace{2.cm}
\includegraphics[scale=.2]{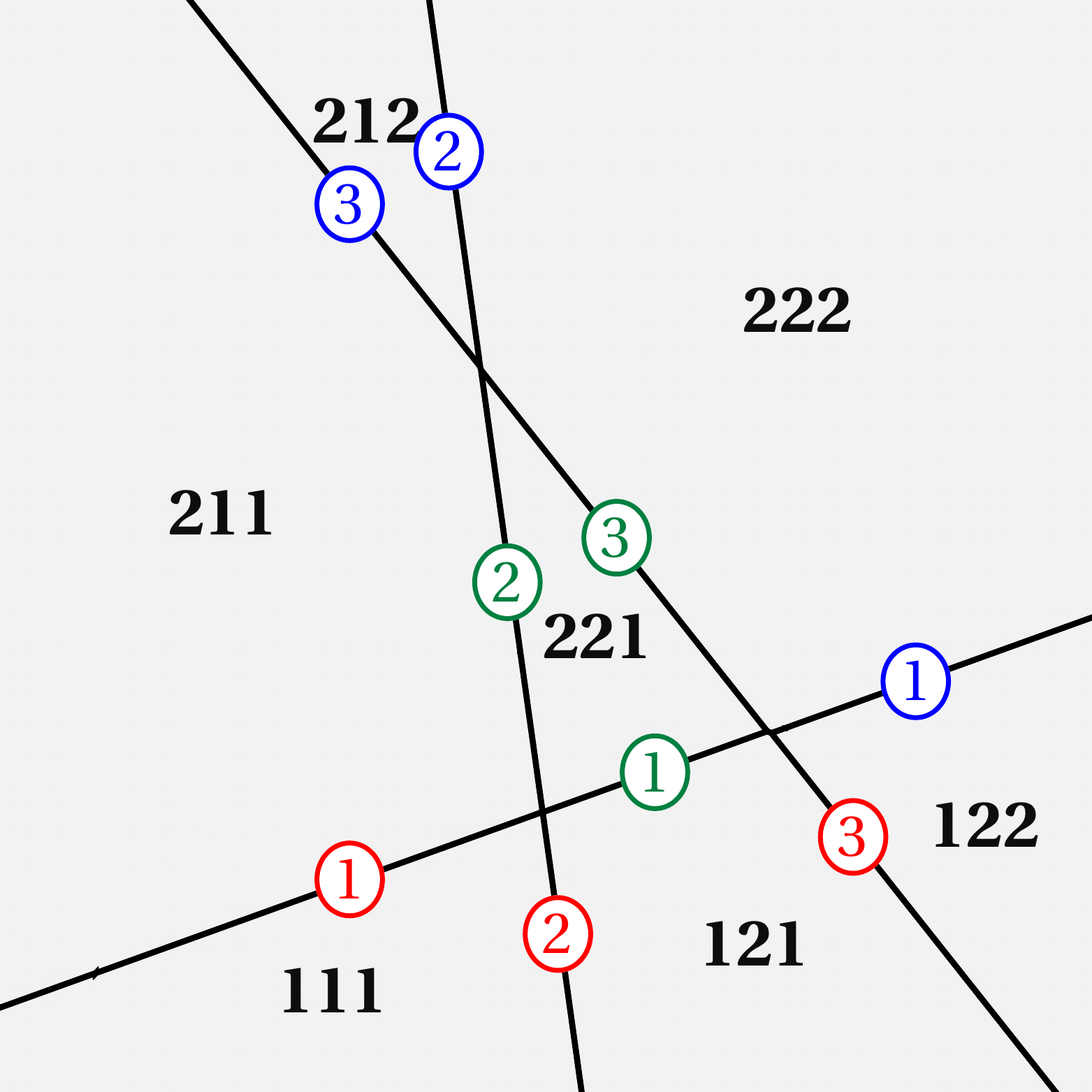} 
\parbox{15cm}{
\caption{Tropical limit graphs of a pure 3-soliton solution of the p2DTL$_K$ equation for a negative (left graph) 
and a positive (right graph) value of $t$, using the data of Example~\ref{ex:3s_2}.
\label{fig:3soliton_2} }
}
\end{center}
\end{figure}
\end{example}

\begin{example}
\label{ex:3s_3}
Again, we choose $K$ and $\eta_i,\xi_i$, $i=1,2,3$, as in Example~\ref{ex:3s_1}, but now
\bez
  p_1=1, \, p_2=10, \, p_3=1 \, , q_1 = 1/2 \, , q_2=10 \, , q_3 =2 \, .
\eez
Corresponding tropical limit graphs are shown in Fig.~\ref{fig:3soliton_3}, from which we read off
\bez
      \mathcal{R}_{\bsy{12}}(1,2) \circ \mathcal{R}_{\bsy{13}}(1,3) \circ \mathcal{R}_{\bsy{23}}(2,3)
    =  \mathcal{R}_{\bsy{23}}(2,3) \circ \mathcal{R}_{\bsy{13}}(1,3) \circ \mathcal{R}_{\bsy{12}}(1,2)  
       \, ,
\eez
which is the Yang-Baxter equation (\ref{YB_eq}).
\begin{figure}[h] 
\begin{center}
\includegraphics[scale=.2]{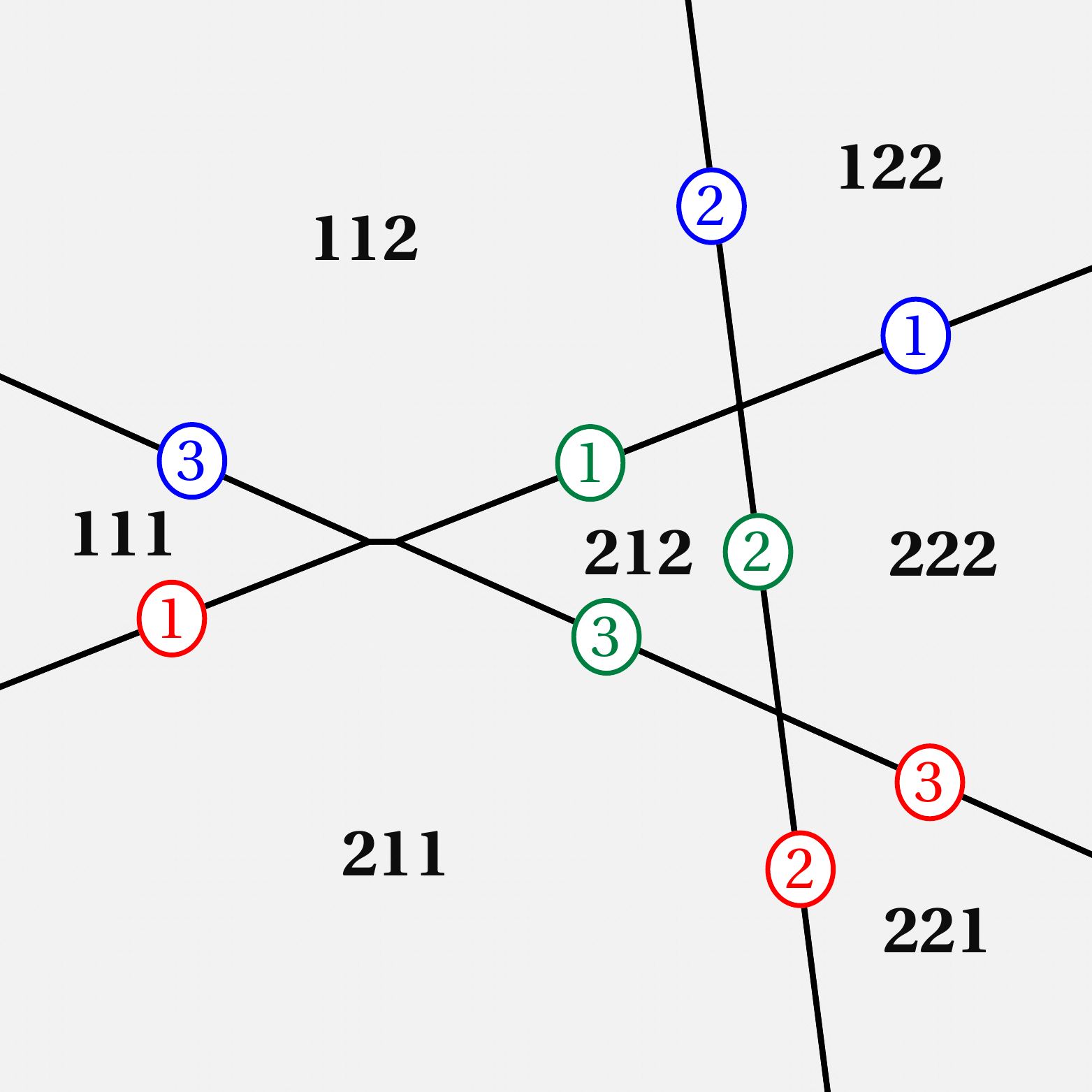}
\hspace{2.cm}
\includegraphics[scale=.2]{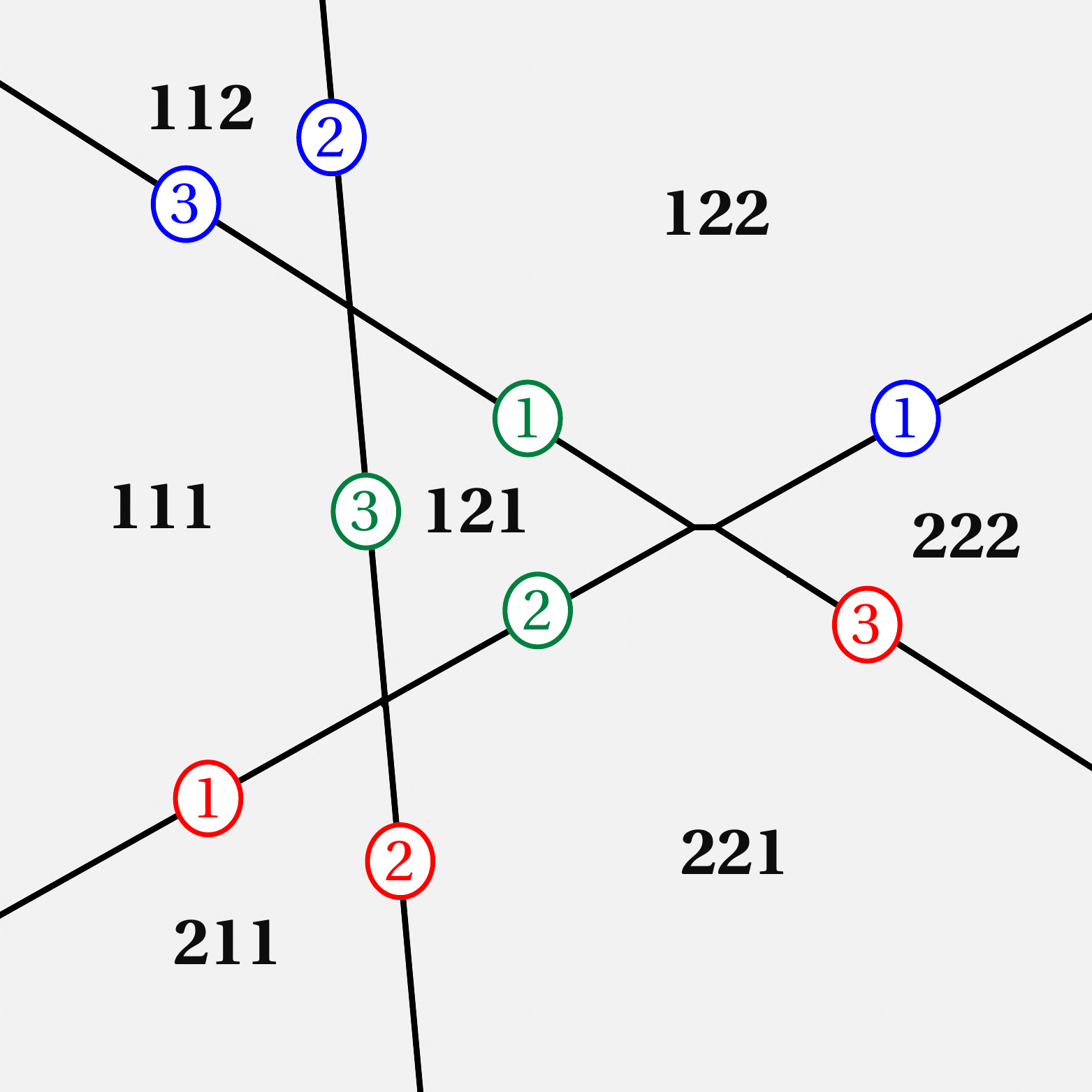}
\parbox{15cm}{
\caption{Tropical limit graphs of yet another pure 3-soliton solution of the p2DTL$_K$ equation for a negative (left graph) 
and a positive (right graph) value of $t$, using the data of Example~\ref{ex:3s_3}.
\label{fig:3soliton_3} }
}
\end{center}
\end{figure}
\end{example}

\section{Conclusions}
\label{sec:concl}
We presented a ``Lax representation'' for the Yang-Baxter map $\mathcal{R}$ obtained in \cite{DMH19LMP} from the pure 
2-soliton solution of the ($K$-modified) matrix KP equation. For the matrix KdV reduction, such a Lax representation 
had been found earlier in \cite{Gonc+Vese04}, also see \cite{Suri+Vese03}. Using also the inverse of the Lax matrix $A$, 
we were led to entwining Yang-Baxter maps, in the sense of \cite{Koul+Papa11BCP}. 
Besides the Yang-Baxter map, this involves another map, which is not Yang-Baxter, but both maps satisfy 
a ``mixed Yang-Baxter equation''. 

We demonstrated that this structure is indeed realized in 3-soliton interactions. Here we concentrated on an 
analysis of matrix generalizations of the two-dimensional Toda lattice equation, of which solitons can be 
generated via a binary Darboux transformation. For the subclass of ``pure solitons'', which only exhibit 
elastic scattering (i.e., no merging or splitting of solitons), we elaborated the tropical limit of  
the general 2- and 3-soliton solution. It turned out that exactly the same Yang-Baxter map as in the 
KP$_K$ case is a work. The crucial new insight is that the flow of polarizations for three solitons 
is not, in general, described by the Yang-Baxter map alone, but by entwining Yang-Baxter maps. 
This also holds for KP$_K$ and even for KdV$_K$ (with $K=1_n$, for example), also see Appendix~\ref{app:KdV_2s}. 
This result crucially relies on our tropical limit analysis of solitons.  

More precisely, in the preceding section we found that (\ref{YB_eq}), (\ref{YB_mixed}) and also (\ref{RTinvTinv=TinvTinvR})
are realized by pure 2DTL$_K$ solitons. We do not know whether this is also so for the remaining mixed 
Yang-Baxter equations in Section~\ref{sec:Lax_system}. 

Since the soliton with number $i$ is specified 
via the parameters $p_i,q_i$, and the polarization $u_i$, we can associate the matrices $A_i(\lambda,u_i)$
and $B_i(\lambda,u_i)$ with it (as well as $\tilde{A}_i(\lambda,u_i)$ and $\tilde{B}_i(\lambda,u_i)$).
Comparing the three-fold products of these Lax matrices, which imply a mixed Yang-Baxter equation, 
with the tropical limit plots for pure 3-solitons, one observes the following rule.  
Whether $A_i$ or $B_i$ is at work, depends on the constellation of phases to the left and to the right of 
the line. For example, if the first soliton has phase $1ab$ to the left and $2cd$ to the right 
($a,b,c,d \in \{1,2\}$), we have to choose $A_1$, whereas with $2ab$ to the left and $1cd$ to the 
right, it has to be $B_1$. The 3-soliton interaction shown in the first plot in Fig.~\ref{fig:3soliton_1} 
corresponds to
\bez
     A_1 B_2 A_3 \stackrel{\mathcal{T}}{\longmapsto} B_2 A_1 A_3 \stackrel{\mathcal{R}}{\longmapsto}
     B_2 A_3 A_1 \stackrel{\mathcal{T}^{-1}}{\longmapsto} A_3 B_2 A_1  
\eez
(hiding away the parameters). The Lax matrices, which are subject to the refactorization equations 
(\ref{Lax_eqs}), generate a Zamolodchikov-Faddeev algebra. In S-matrix 
theory of an integrable QFT, these matrices play a role as creation and annihilation operators
(see, e.g., Section 4.2 in \cite{Bomb16}, and references cited there). 

We also found that the two maps $\mathcal{R}$, $\mathcal{T}$, and their inverses, provide us with solutions 
of the ``WXZ system'' \cite{Hlav+Snob99}, called \emph{Yang-Baxter system} in \cite{Brze+Nich05}. 

We further note that the Yang-Baxter map obtained for the matrix Nonlinear Schr\"odinger (NLS) equation in 
\cite{APT04IP,Tsuc04,Tsuc15} has the same form as the Yang-Baxter map $\mathcal{R}$ for matrix KP and matrix 
two-dimensional Toda lattice. 

Like KP$_K$, also the 2DTL$_K$ equation possesses many soliton solutions beyond pure solitons, 
see Section~\ref{subsec:Toda_bDT}. 
A corresponding analysis, following \cite{DMH18p}, goes beyond Yang-Baxter and will be postponed to future work.

\renewcommand{\theequation} {\Alph{section}.\arabic{equation}}
\renewcommand{\thesection} {\Alph{section}}

\makeatletter
\newcommand\appendix@section[1]{
  \refstepcounter{section}
  \orig@section*{Appendix \@Alph\c@section: #1}
  \addcontentsline{toc}{section}{Appendix \@Alph\c@section: #1}
}
\let\orig@section\section
\g@addto@macro\appendix{\let\section\appendix@section}
\makeatother

\begin{appendix}

\section{2-soliton solutions of the KdV$_K$ equation} 
\label{app:KdV_2s}
The 2-soliton solution of (\ref{KdV_K}) is obtained from the formulas in Section~\ref{subsec:2sol} 
by replacing the expressions for $\vartheta_{ab}$ by
\bez
      \vartheta_{11} = \vartheta(p_1) + \vartheta(p_2) \, , \quad
      \vartheta_{12} = \vartheta(p_1) + \vartheta(-p_2) \, , \quad
      \vartheta_{21} = \vartheta(p_2) + \vartheta(-p_1) \, , \quad
      \vartheta_{22} = \vartheta(-p_1) + \vartheta(-p_2) \, ,
\eez
where now
\bez
      \vartheta(p) = p \, x + p^3 \, t \, .
\eez
To determine the asymptotics of soliton 1 in the 2-soliton solution, we set
\bez
     x + p_1^2 t = \Lambda_1^{(\pm)} \, ,
\eez
with constants $\Lambda_1^{(\pm)}$, and take the limit as $t \to \pm \infty$.
Then we choose $\Lambda_1^{(\pm)}$ in such a way that phase shifts are compensated.\footnote{In 
\cite{APT04IP,Tsuc04}, which deal with vector NLS solitons, the derived Yang-Baxter maps include 
factors due to phase shifts. } 
Correspondingly for soliton 2. In this way we obtain the following.
\begin{enumerate}
\item Let $p_2 < p_1 < 0$. 
Then $u_{11,21}$ ($u_{21,22}$) and $u_{12,22}$ ($u_{11,12}$) are the polarizations
of soliton 1 (soliton 2) as $t \to -\infty$, respectively, $t \to \infty$. 
This is the phase constellation shown in Fig.~\ref{fig:KdV_2s_21}.
We find the following relations with polarizations defined via the tropical limit,
\bez
  &&  \hat{u}_{11,21} = \lim_{t \to -\infty} u\Big|_{x\mapsto -p_1^2 t + \Lambda_1^{(-)}} \, , \qquad
      \hat{u}_{12,22} = \lim_{t \to +\infty} u\Big|_{x\mapsto -p_1^2 t + \Lambda_1^{(+)}} \, , \\
  &&  \hat{u}_{21,22} = \lim_{t \to -\infty} u\Big|_{x\mapsto -p_2^2 t + \Lambda_2^{(-)}} \, , \qquad
      \hat{u}_{11,12} = \lim_{t \to +\infty} u\Big|_{x\mapsto -p_2^2 t + \Lambda_2^{(+)}} \, ,
\eez
where $\Lambda_1^{(-)} = - \log(\alpha_{12} \kappa_{11})/(2p_1)$, $\Lambda_1^{(+)} = - \log(\kappa_{11})/(2p_1)$,
$\Lambda_2^{(-)} = - \log(\kappa_{22})/(2p_2)$, $\Lambda_2^{(+)} = - \log(\alpha_{12} \kappa_{22})/(2p_1)$.
\item Let $0 < p_1 < p_2$. 
Then $u_{12,22}$ ($u_{11,12}$) and $u_{11,21}$ ($u_{21,22}$) are the polarizations
of soliton 1 (soliton 2) as $t \to -\infty$, respectively, $t \to \infty$. 
This is the phase constellation shown in the first plot of Fig.~\ref{fig:KdV_2s_other}.
\item Let $p_2 < 0 < p_1$, $|p_1|<|p_2|$.
Then $u_{11,21}$ ($u_{11,12}$) and $u_{12,22}$ ($u_{21,22}$) are the polarizations
of soliton 1 (soliton 2) as $t \to -\infty$, respectively, $t \to \infty$. 
This is the phase constellation shown in the second plot of Fig.~\ref{fig:KdV_2s_other}.
\item Let $p_1 < 0 < p_2$, $|p_1|<|p_2|$.
Then $u_{12,22}$ ($u_{21,22}$) and $u_{11,21}$ ($u_{11,12}$) are the polarizations
of soliton 1 (soliton 2) as $t \to -\infty$, respectively, $t \to \infty$. 
This is the phase constellation displayed in the third plot of Fig.~\ref{fig:KdV_2s_other}.
\end{enumerate}
In all these cases we have $|p_1|<|p_2|$, so that, for large enough negative values of $t$, 
soliton 1 appears to the left of soliton 2 in $x$-direction.
The concrete tropical limit graphs in Figs.~\ref{fig:KdV_2s_21} and  \ref{fig:KdV_2s_other} 
are obtained with 
\bez
    K=(1,1) \, , \quad
    \xi_1 = \left(\begin{array}{c} 1 \\ 0 \end{array} \right) \, , \quad
    \xi_2 = \left(\begin{array}{c} 0 \\ 1 \end{array} \right)  \, , \quad
    \eta_1 = \eta_2 = 1 \, ,
\eez
and the following data, respectively:
\begin{enumerate}
\item $p_1 =-1/2$, $p_2 = -3/2$. 
\item $p_1 =1/2$, $p_2 = 3/2$. 
\item $p_1 =1/2$, $p_2 = -3/2$. 
\item $p_1 =-1/2$, $p_2 = 3/2$. 
\end{enumerate}

\section{Proof of Theorem~\ref{thm:Lax}}
\label{app:unique_map}
Let $K$ be an $n \times m$ matrix with maximal rank. 

\begin{lemma}
\label{lemma:invert_sum_of_rank1_proj}
Let $K$ have maximal rank and $X_i$, $i=1,2$, be rank one $K$-projections.
Let $\alpha,\beta$ be constants such that
\bez
      \gamma := (1+\alpha)(1+\beta) - \alpha \beta \, \mathrm{tr}(X_i K X_j K) \neq 0 \, .
\eez
Then
\bez
    \Big(1_m + \alpha \, X_1 K + \beta \, X_2 K \Big)^{-1}
 = 1_m -  \frac{\alpha (1+\beta)}{\gamma} \, X_1 K - \frac{(1+\alpha) \beta}{\gamma} \, X_2 K 
    + \frac{\alpha\beta}{\gamma} \, (X_1 K X_2 K + X_2K X_1 K) \, .
\eez
\end{lemma}
\begin{proof}
Since $X_i$ is a $K$-projection, $X_i K$ and $K X_i$ are ordinary projection (i.e., idempotent) matrices. 
If $K$ has maximal rank, they  
have rank one iff $X_i$ has rank one. Hence there is a column vector $\xi_i$ and a row vector $\eta_i$ such that
$X_i K = \xi_i \eta_i$ and $\eta_i \xi_i =1$, and correspondingly for $K X_i$. It follows that 
\bez
  && X_i K X_j K X_i K = \mathrm{tr}(X_i K X_j K) \, X_i K \, , \\
  && K X_i K X_j K X_i = \mathrm{tr}(K X_i K X_j) \, K X_i = \mathrm{tr}(X_i K X_j K) \, K X_i \, . 
\eez  
Since $K$ is assumed to have maximal rank, this implies
\be
    X_i K X_j K X_i = \mathrm{tr}(X_i K X_j K) \, X_i \, .  \label{3proj_trace_relation}
\ee
Using further the $K$-projection property of $X_i$, our assertion can be directly verified. 
\end{proof}

\begin{proposition}
\label{prop:unique_map}
Let $K$ have maximal rank, $p_1,p_2,q_1,q_2$ be pairwise distinct, and $X_i \in \boldsymbol{S}$, $i=1,2$.
The system  (\ref{Lax_eqs}) determines a map 
$\boldsymbol{S} \times \boldsymbol{S} \rightarrow \boldsymbol{S} \times \boldsymbol{S}$ 
via $(X_1,X_2) \mapsto (X_1',X_2')$, which is given by (\ref{KP_YB_map}).
\end{proposition}
\begin{proof}
We multiply (\ref{Lax_eqs}) by $(\lambda-q_1)(\lambda-q_2)$ and expand in powers of $\lambda$.
Since $K$ is assumed to have maximal rank, from the coefficient linear in $\lambda$ we obtain
\bez
      X_2' = X_2 - \frac{p_1-q_1}{p_2-q_2} \, (X_1' - X_1)  \, ,
\eez
which can be used to replace $X_2'$ in the $\lambda$-independent part of the expression we started with,
\bez
  0 &=& q_1 (p_2-q_2) (X_2'-X_2) + q_2 (p_1-q_1)(X_1'-X_1) + (p_1-q_1)(p_2-q_2)(X_2' K X_1'-X_1 K X_2)  \\
    &=& (p_1-q_1)\Big( (q_2-q_1)(X_1'-X_1) + (p_2-q_2)(X_2 K X_1' - X_1 K X_2)
            - (p_1-q_1)(X_1'-X_1) K X_1' \Big) \, .
\eez 
Since $X_1'$ is a $K$-projection, this becomes
\bez
   (q_2-q_1)(X_1'-X_1) + (p_2-q_2)(X_2 K X_1' - X_1 K X_2)
            - (p_1-q_1) X_1'  + (p_1-q_1) X_1 K X_1'  = 0 \, ,
\eez 
so that
\be
     \Big(1_m - \frac{p_1-q_1}{p_1-q_2} X_1 K  
    -\frac{p_2-q_2}{p_1-q_2} X_2 K \Big)  X_1'
  = \frac{q_1-q_2}{p_1-q_2} X_1 \Big( 1_n - \frac{ p_2-q_2}{q_1-q_2}  K X_2 \Big) \label{X'_1_X_1_relation}
\ee
If $X_i$, $i=1,2$, have rank one, the inverse of the matrix multiplying $X_1'$ is given in the 
preceding lemma. We have to apply it to the right hand side of the last expression. 
First we compute
\bez
  &&  \left (1_m -  \frac{\alpha (1+\beta)}{\gamma} \, X_1 K - \frac{(1+\alpha) \beta}{\gamma} \, X_2 K 
    + \frac{\alpha\beta}{\gamma} \, (X_1 K X_2 K + X_2K X_1 K) \right) X_1  \\
 && \hspace{2cm} = \frac{1}{\gamma} \left( (1+\beta) \, 1_m - \beta \, X_2 K \right) X_1 \, ,
\eez 
using the $K$-projection property of $X_i$, $i=1,2$, and (\ref{3proj_trace_relation}). Here we have
\bez
       \alpha = -\frac{p_1-q_1}{p_1-q_2} \, , \quad
       \beta = - \frac{p_2-q_2}{p_1-q_2} \, , \quad
       \gamma = \frac{(q_1-q_2)(p_1-p_2)}{(p_1-q_2)^2} \, \alpha_{12}^{-1} \, ,
\eez
with $\alpha_{12}$ defined in (\ref{alpha_12}). A straightforward computation now leads to
\bez
   X_1' = \alpha_{12} \, \Big( 1_m - \frac{p_2-q_2}{p_2-p_1} X_2 K \Big) \, X_1 \, 
                    \Big( 1_n - \frac{p_2-q_2}{q_1-q_2} K X_2 \Big) \, ,
\eez
which is the first of equations (\ref{KP_YB_map}). The second equation 
is obtained in the same way and we can verify that $X_i' \in \boldsymbol{S}$, $i=1,2$.
\end{proof}

\begin{remark}
In \cite{Gonc+Vese04} it has been noted that, in the case associated with the matrix KdV equation, (\ref{Lax_eqs}) 
determines, more generally, a Yang-Baxter map if the set $\boldsymbol{S}$ is extended to the set  
of \emph{all} $K$-projections, i.e., without restriction to rank one. In the more general situation considered in the 
present work, we observe that (\ref{X'_1_X_1_relation}), which has been derived without restriction of the rank,
can be rewritten as
\bez
     \Big(1_m - \frac{p_1-q_1}{p_1-q_2} X_1 K  
    -\frac{p_2-q_2}{p_1-q_2} X_2 K \Big)  X'_1
   = X_1 \Big(1_m - \frac{p_1-q_1}{p_1-q_2} K X_1 -\frac{p_2-q_2}{p_1-q_2} K X_2 \Big) 
\eez
(cf. \cite{Tsuc15} for the NLS case). 
If the matrix multiplying $X'_1$ is invertible, it follows that the latter is a $K$-projection. A corresponding 
argument shows that also $X'_2$ is a $K$-projection. A convenient formula for the inverse of the matrix 
multiplying $X'_1$ (and correspondingly for $X'_2$), like that given in Lemma~\ref{lemma:invert_sum_of_rank1_proj} 
for the rank one case, is not available, however. 
At least in the special case where $X_1$ and $X_2$ both have rank $r$ and satisfy
\be
     X_i K X_j K X_i = \mu \, X_i \qquad \quad i,j=1,2, \quad i \neq j \, ,     \label{X_1X_2X_1_propto_X_1}
\ee
which implies that the scalar $\mu$ is given by $r^{-1} \mathrm{tr}(X_i K X_j K)$, one can show that
\bez
   X_1' = r \, \frac{B_2(p_1,X_2) \, X_1 \, \tilde{A}_2(q_1,X_2)}
                      {\mathrm{tr}[B_2(p_1,X_2) \, X_1 \, \tilde{A}_2(q_1,X_2) \, K]} \, , \quad
   X_2' = r \, \frac{A_1(q_2,X_1) \, X_2 \, \tilde{B}_1(p_2,X_1)}
                      {\mathrm{tr}[A_1(q_2,X_1) \, X_2 \, \tilde{B}_1(p_2,X_1) \, K]} \, .
\eez
Since (\ref{X_1X_2X_1_propto_X_1}) holds for any two \emph{rank one} $K$-projections, we recover (\ref{X_1,2,out}). 
\end{remark}

\section{Derivation of the binary Darboux transformation for the p2DTL$_K$ equation}
\label{app:BDT}
We recall a binary Darboux transformation result of bidifferential calculus \cite{DMH13SIGMA,CDMH16}. 

\begin{theorem}
Let $(\Omega,\d, \bd)$ be a bidifferential calculus and $\Delta, \Gamma, \boldsymbol{\lambda}, \boldsymbol{\kappa}$ 
solutions of
\bez
  &&  \bd \Gamma = \Gamma \, \d \Gamma + [\boldsymbol{\kappa} , \Gamma] \, , \qquad
    \bd \boldsymbol{\kappa} = \Gamma \, \d \boldsymbol{\kappa} + \boldsymbol{\kappa}^2 \, ,  \\
  &&  \bd \Delta = (\d \Delta) \, \Delta - [\boldsymbol{\lambda} , \Delta] \, ,  \qquad
     \bd \boldsymbol{\lambda} = (\d \boldsymbol{\lambda}) \, \Delta - \boldsymbol{\lambda}^2 \, ,
\eez
and $\phi_0$ a solution of 
\be
      \d \bd \phi + \d \phi \, K \, \d \phi = 0 \, ,   \label{bDT_phi_eq}
\ee
where $\d K = 0 = \bd K$. 
Let $\theta$ and $\chi$ be solutions of the linear system
\be
     \bd \theta = (\d \phi_0) \, K \, \theta + (\d \theta) \, \Delta + \theta \, \boldsymbol{\lambda} \, ,
       \label{bDT_lin_sys}
\ee
respectively the adjoint linear system
\be
    \bd \chi = - \chi \, K \, \d \phi_0 + \Gamma \, \d \chi + \boldsymbol{\kappa} \, \chi \, .
       \label{bDT_adj_lin_sys}
\ee
Let $\Omega$ solve the compatible linear system
\be
  &&  \Gamma \, \Omega - \Omega \, \Delta = \chi \, K \, \theta \, , \nonumber \\
  &&  \bd \Omega = (\d \Omega) \, \Delta - (\d \Gamma) \, \Omega + (\d \chi) \, K \, \theta
                   + \boldsymbol{\kappa} \, \Omega + \Omega \, \boldsymbol{\lambda} \, .   \label{bDT_Omega_eqs}
\ee
Where $\Omega$ is invertible, 
\be
     \phi = \phi_0 - \theta \, \Omega^{-1} \chi     \label{bDT_new_solution}
\ee
is a new solution of (\ref{bDT_phi_eq}). \hfill $\Box$
\end{theorem}

In the above theorem, we have to assume that all objects are such that the corresponding products are defined 
and that $\d$ and $\bd$ can be applied. 
Next we define a bidifferential calculus via
\bez
   && \d f = [\mathbb{S} , f ] \, \zeta_1 + f_y \, \zeta_2  \, , \\
   && \bd f = f_x \, \zeta_1 - [\mathbb{S}^{-1} , f] \, \zeta_2 \, , 
\eez
on the algebra $\cA = \cA_0[\mathbb{S},\mathbb{S}^{-1}]$, where $\cA_0$ is the algebra of smooth functions 
of two variables, $x$ and $y$, and also dependent on a discrete variable on which the shift operator $\mathbb{S}$ acts.
$\zeta_1, \zeta_2$ constitute a basis of a two-dimensional vector space $V$, from which we form 
the Grassmann algebra $\Lambda(V)$. $\d$ and $\bd$ extend to $\Omega = \cA \otimes \Lambda(V)$ 
in a canonical way, and to matrices with entries in $\Omega$. Setting
\bez
    \phi = \varphi \, \mathbb{S}^{-1} \, ,
\eez
the equation (\ref{bDT_phi_eq}) is equivalent to the p2DTL$_K$ equation (\ref{matrixToda_K}). 
Choosing a solution $\phi_0 = \varphi_0 \, \mathbb{S}^{-1}$ and setting
\bez
     \Delta = \Gamma = \mathbb{S}^{-1} \, , \quad
     \boldsymbol{\kappa} = \boldsymbol{\lambda} = 0 \, ,
\eez
the linear system (\ref{bDT_lin_sys}) and the adjoint linear system (\ref{bDT_adj_lin_sys}) 
lead to (\ref{Toda_lin_sys}) and (\ref{Toda_adj_lin_sys}), respectively. 
Furthermore, via $\Omega \mapsto \Omega \mathbb{S}$, (\ref{bDT_Omega_eqs}) implies (\ref{Toda_Omega}). 
According to the theorem, (\ref{bDT_new_solution}) 
yields a new solution of the p2DTL$_K$ equation (\ref{matrixToda_K}).

\section{Computational details for Section~\ref{subsec:3sol}} 
\label{app:comp_details}
Using the 3-soliton solution in Section~\ref{subsec:3sol}, and the notation introduced there,
we find that
\bez
     u_{i,\bullet} = \xi_{i,\bullet} \otimes \eta_{i,\bullet} \, , 
\eez
where $\bullet$ stands for $\mathrm{in},\mathrm{out},\mathrm{m1}$ or $\mathrm{m2}$, 
\bez
  \xi_{1,\mathrm{in}} &=& \frac{1}{\alpha_{13}} A_3(p_1,\Xi_3) \, \frac{\xi_1}{\kappa_{11}} \, , \qquad
  \eta_{1,\mathrm{in}} = \eta_1 \, \tilde{B}_3(q_1,\Xi_3) \, , \\
  \xi_{1,\mathrm{m1}} &=& \frac{\alpha_{23}}{\beta} 
     \Big( 1_m-\frac{(p_2-q_2) \alpha_{321}}{(p_1-q_2) \alpha_{23}} \Xi_2 K 
       - \frac{(p_3-q_3) \alpha_{231}}{(p_1-q_3) \alpha_{23}} \Xi_3 K \Big) \, \frac{\xi_1}{\kappa_{11}} \, , \\
  \eta_{1,\mathrm{m1}} &=& \eta_1 \, \Big( 1_n-\frac{(p_2-q_2) \alpha_{312}}{(p_2-q_1) \alpha_{23}} K \Xi_2 
      - \frac{(p_3-q_3) \alpha_{213}}{(p_3-q_1) \alpha_{23}} K \Xi_3 \Big)  \, , \\
  \xi_{1,\mathrm{out}} &=& \frac{1}{\alpha_{12}} A_2(p_1,\Xi_2) \, \frac{\xi_1}{\kappa_{11}} \, , \qquad  
  \eta_{1,\mathrm{out}} = \eta_1 \, \tilde{B}_2(q_1,\Xi_2) \, , \\
  \xi_{2,\mathrm{in}} &=& \frac{1}{\alpha_{23}} A_3(p_2,\Xi_3) \, \frac{\xi_2}{\kappa_{22}} \, , \qquad
  \eta_{2,\mathrm{in}} = \eta_2 \, \tilde{B}_3(q_2,\Xi_3) \, , \\
  \xi_{2,\mathrm{m1}} &=& \frac{\alpha_{13}}{\beta} \Big( 1_m-\frac{(p_1-q_1) \alpha_{312}}{(p_2-q_1) \alpha_{13}} 
     \Xi_1 K - \frac{(p_3-q_3) \alpha_{132}}{(p_2-q_3) \alpha_{13}} \Xi_3 K \Big) \, \frac{\xi_2}{\kappa_{22}} \, , \\
  \eta_{2,\mathrm{m1}} &=& \eta_2 \, \Big( 1_n-\frac{(p_1-q_1) \alpha_{321}}{(p_1-q_2) \alpha_{13}} K \Xi_1 
      - \frac{(p_3-q_3) \alpha_{123}}{(p_3-q_2) \alpha_{13}} K \Xi_3 \Big)  \, , \\  
  \xi_{2,\mathrm{out}} &=& \frac{1}{\alpha_{12}} A_1(p_2,\Xi_1) \, \frac{\xi_2}{\kappa_{22}} \, , \qquad
  \eta_{2,\mathrm{out}} = \eta_2 \, \tilde{B}_1(q_2,\Xi_1) \, , \\
  \xi_{3,\mathrm{in}} &=& \frac{1}{\alpha_{23}} A_2(p_3,\Xi_2) \, \frac{\xi_3}{\kappa_{33}} \, , \qquad
  \eta_{3,\mathrm{in}} = \eta_3 \, \tilde{B}_2(q_3,\Xi_2) \, , \\
  \xi_{3,\mathrm{m1}} &=& \frac{\alpha_{12}}{\beta} \Big( 1_m-\frac{(p_1-q_1) \alpha_{213}}{(p_3-q_1) \alpha_{12}} 
     \Xi_1 K - \frac{(p_2-q_2) \alpha_{123}}{(p_3-q_2) \alpha_{12}} \Xi_2 K \Big) \, \frac{\xi_3}{\kappa_{33}} \, , \\
  \eta_{3,\mathrm{m1}} &=& \eta_3 \, \Big( 1_n-\frac{(p_1-q_1) \alpha_{231}}{(p_1-q_3) \alpha_{12}} K \Xi_1 
      - \frac{(p_2-q_2) \alpha_{132}}{(p_2-q_3) \alpha_{12}} K \Xi_2 \Big)  \, , \\
  \xi_{3,\mathrm{out}} &=& \frac{1}{\alpha_{13}} A_1(p_3,\Xi_1) \,  \frac{\xi_3}{\kappa_{33}} \, , \qquad
  \eta_{3,\mathrm{out}} = \eta_3 \, \tilde{B}_1(q_3,\Xi_1) \, , 
\eez
and
\bez
  \xi_{1,\mathrm{m2}} = \frac{\xi_1}{\kappa_{11}} \, , \quad 
  \eta_{1,\mathrm{m2}} = \eta_1 \, , \quad
  \xi_{2,\mathrm{m2}} = \frac{\xi_2}{\kappa_{22}} \, , \quad 
  \eta_{2,\mathrm{m2}} = \eta_2 \, , \quad
  \xi_{3,\mathrm{m2}} = \frac{\xi_3}{\kappa_{33}} \, , \quad 
  \eta_{3,\mathrm{m2}} = \eta_3 \, .
\eez
Here we used (\ref{A,tA,B,tB}) and introduced the rank one $K$-projections 
\bez
     \Xi_i := \frac{\xi_i \otimes \eta_i}{\kappa_{ii}} \qquad i=1,2,3 \, . 
\eez

For example, we find
\be
   \xi_{1,\mathrm{out}} &=& \frac{B_3(p_1,u_{3,\mathrm{in}}) \, \xi_{1,\mathrm{m1}}}
        {\mathrm{tr}[B_3(p_1,u_{3,\mathrm{in}}) \, u_{1,\mathrm{m1}} \, \tilde{A}_3(q_1,u_{3,\mathrm{in}}) \, K ]} \, , \nonumber \\
   \eta_{1,\mathrm{out}} &=& \eta_{1,\mathrm{m1}} \, \tilde{A}_3(q_1,u_{3,\mathrm{in}}) \, ,     
   \label{xi,eta_out_from_m1}
\ee
and 
\bez
  &&  \xi_{2,\mathrm{m2}} = B_3(q_2,u_{3,\mathrm{in}}) \, \xi_{2,\mathrm{in}} \, , \quad
    \eta_{2,\mathrm{m2}} = \frac{\eta_{2,\mathrm{in}} \, \tilde{A}_3(p_2,u_{3,\mathrm{in}})}{\mathrm{tr}\left(
        B_3(q_2,u_{3,\mathrm{in}}) \, u_{2,\mathrm{in}} \, \tilde{A}_3(p_2,u_{3,\mathrm{in}}) \, K  \right)} \, , \\
  &&  \xi_{3,\mathrm{m2}} = B_2(q_3,u_{2,\mathrm{in}}) \, \xi_{3,\mathrm{in}} \, , \quad
    \eta_{3,\mathrm{m2}} = \frac{\eta_{3,\mathrm{in}} \, \tilde{A}_3(p_3,u_{2,\mathrm{in}})}{\mathrm{tr}\left(
        B_2(q_3,u_{2,\mathrm{in}}) \, u_{3,\mathrm{in}} \, \tilde{A}_2(p_3,u_{2,\mathrm{in}}) \, K  \right)} \, ,
\eez
 from which some statements in Section~\ref{subsec:3sol} are quickly deduced. 
 
\end{appendix}


\small
\begin{thebibliography}{10}

\bibitem{Bomb16}
Bombardelli, D.: 
 $S$-matrices and integrability,
 {\em J. Phys. A: Math. Theor.}
 {\bf 49}, 323003 (2016)

\bibitem{Gonc+Vese04}
Goncharenko, V., Veselov, A.: 
 Yang-Baxter maps and matrix solitons,
 in New Trends in Integrability and Partial Solvability,  (Shabat, A.,
  et~al., eds.), Vol. 132 of {\em NATO Science Series II: Math. Phys. Chem.},
 Kluwer, Dordrecht,
 pp.  191--197 (2004)

\bibitem{APT04IP}
Ablowitz, M., Prinari, B., Trubatch, A.: 
 Soliton interactions in the vector NLS equation,
 {\em Inv. Problems}
 {\bf 20}, 1217--1237 (2004)

\bibitem{Tsuc04}
Tsuchida, T.: 
 $N$-soliton collision in the Manakov model,
 {\em Progr. Theor. Phys.}
 {\bf  111}, 151--182 (2004)
 
\bibitem{Tsuc15}
Tsuchida, T.: 
 Integrable discretization of the vector/matrix nonlinear Schr\"odinger equation and the 
 associated Yang-Baxter map,
 {\em arXiv:1505.07924}
 (2004) 

\bibitem{DMH11KPT}
Dimakis, A., M\"uller-Hoissen, F.: 
 KP line solitons and Tamari lattices,
 {\em J. Phys. A: Math. Theor.}
 {\bf  44}, 025203 (2011)

\bibitem{DMH12KPBT}
Dimakis, A., M\"uller-Hoissen, F.: 
 KP solitons, higher Bruhat and Tamari orders,
 in Associahedra, Tamari Lattices and Related Structures,
  (M\"uller-Hoissen, F., Pallo, J., Stasheff, J., eds.), Vol. 299 of {\em
  Progress in Mathematics},
 Birkh\"auserBasel,
 pp.  391--423 (2012)

\bibitem{DMH14KdV}
Dimakis, A., M\"uller-Hoissen, F.: 
 KdV soliton interactions: a tropical view,
 {\em J. Phys. Conf. Ser.}
 {\bf  482}, 012010 (2014)

\bibitem{DMH19LMP}
Dimakis, A., M\"uller-Hoissen, F.: 
 Matrix KP: tropical limit and Yang-Baxter maps,
 {\em Lett. Math. Phys.}
 {\bf  109}, 799--827 (2019)

\bibitem{DMH18p}
Dimakis, A., M\"uller-Hoissen, F.: 
 Matrix KP: tropical limit, Yang-Baxter and pentagon maps,
 {\em Theor. Math. Phys.}
 {\bf  196}, 1164--1173 (2018)

\bibitem{DMHC19}
Dimakis, A., M\"uller-Hoissen, F., Chen, X.: 
 Matrix Boussinesq solitons and their tropical limit,
 {\em Physica Scripta}
 {\bf  94}, 035206 (2019)

\bibitem{Gonc01}
Goncharenko, V.: 
 Multisoliton solutions of the matrix KdV equation,
 {\em Theor. Math. Phys.}
 {\bf  126},  81--91 (2001)

\bibitem{Mikh79}
Mikhailov, A.: 
 Integrability of a two-dimensional generalization of the Toda
  chain,
 {\em JETP Lett.}
 {\bf  30}, 414--418 (1979)

\bibitem{Taka18}
Takasaki, K.: 
 Toda hierarchies and their applications,
 {\em J. Phys. A: Math. Theor.}
 {\bf  51}, 203001 (2018)

\bibitem{Bion+Wang10}
Biondini, G., Wang, D.: 
 On the soliton solutions of the two-dimensional Toda lattice,
 {\em J. Phys. A: Math. Theor.}
 {\bf  43}, 434007 (2010)

\bibitem{Koul+Papa11BCP}
Kouloukas, T., Papageorgiou, V.: 
 Entwining Yang-Baxter maps and integrable lattices,
 {\em Banach Center Publ.}
 {\bf  93}, 163--175 (2011)

\bibitem{DMH15}
Dimakis, A., M\"uller-Hoissen, F.: 
 Simplex and polygon equations,
 {\em SIGMA}
 {\bf  11}, 042 (2015)

\bibitem{Suri+Vese03}
Suris, Y., Veselov, A.: 
 Lax matrices for Yang-Baxter maps,
 {\em J. Nonl. Math. Phys.}
 {\bf  10}, 223--230 (2003)
 
\bibitem{Resh+Vese05}
Reshetikhin, N., Veselov, A.:
Poisson Lie groups and Hamiltonian theory of the Yang-Baxter maps,
{\em arXiv:math/0512328} (2005)

\bibitem{Hlav+Snob99}
Hlavat\'{y}, L., Snobl, L.: 
 Solution of the Yang-Baxter system for quantum doubles,
 {\em Int. J. Mod. Phys. A}
 {\bf  14}, 3029--3058 (1999)

\bibitem{Brze+Nich05}
Brzezi\'{n}ski, T., Nichita, F.: 
 Yang-Baxter systems and entwining structures,
 {\em Commun. Alg.}
 {\bf  33}, 1083--1093 (2005)

\bibitem{Vlad93}
Vladimirov, A.: 
 A method for obtaining quantum doubles from the Yang-Baxter
  $R$-matrices,
 {\em Mod. Phys. Lett. A}
 {\bf  8}, 1315--1322 (1993)

\bibitem{KNW09}
Kakei, S., Nimmo, J.J.C., Willox, R.:
Yang-Baxter maps and the discrete KP hierarchy,
{\em Glasgow Math. J.} 
{\bf 51A}, 107--119 (2009)

\bibitem{Hlav94}
Hlavat\'{y}, L.: 
 Quantized braid groups,
 {\em J. Math. Phys.}
 {\bf  35}, 2560--2569 (1994)

\bibitem{HIK88}
Hirota, R., Ito, M., Kako, F.: 
 Two-dimensional Toda lattice equations,
 {\em Prog. Theor. Phys. Suppl.}
 {\bf  94}, 42--58 (1988)

\bibitem{Ueno+Taka84}
Ueno, K., Takasaki, K.: 
 Toda lattice hierarchy,
 in Group Representations and Systems of Differential
  Equations,  (Okamoto, K., ed.), Vol.~4 of {\em Advanced Studies in Pure
  Mathematics},
 North-Holland, Amsterdam,
 pp.  1--95 (1984)

\bibitem{Hiro04}
Hirota, R.: 
 {\em The Direct Method in Soliton Theory},
 Vol. 155 of {\em Cambridge Tracts in Mathematics},
 Cambridge University Press,
 Cambridge (2004)

\bibitem{Toda89}
Toda, M.: 
 {\em Theory of Nonlinear Lattices},
 Springer,
 Berlin (1989)

\bibitem{IWS03}
Isojima, S., Willox, R., Satsuma, J.: 
 Spider-web solutions of the coupled KP equation,
 {\em J. Phys. A: Math. Gen.}
 {\bf  36}, 9533--9552 (2003)

\bibitem{Maru+Bion04}
Maruno, K., Biondini, G.: 
 Resonance and web structure in discrete soliton systems: the
  two-dimensional Toda lattice and its fully discrete and ultra-discrete
  analogues,
 {\em J. Phys. A: Math. Gen.}
 {\bf  37}, 11819--11839 (2004)

\bibitem{Bion+Chak06}
Biondini, G., Chakravarty, S.: 
 Soliton solutions of the Kadomtsev-Petviashvili II equation,
 {\em J. Math. Phys.}
 {\bf  47}, 033514 (2006)

\bibitem{Chak+Koda08JPA}
Chakravarty, S., Kodama, Y.: 
 Classification of the line-soliton solutions of KPII,
 {\em J. Phys. A: Math. Theor.}
 {\bf  41}, 275209 (2008)

\bibitem{DMH13SIGMA}
Dimakis, A., M\"uller-Hoissen, F.: 
 Binary Darboux transformations in bidifferential calculus and
  integrable reductions of vacuum Einstein equations,
 {\em SIGMA}
 {\bf  9}, 009 (2013)

\bibitem{CDMH16}
Chvartatskyi, O., Dimakis, A., M\"uller-Hoissen, F.: 
 Self-consistent sources for integrable equations via deformations of
  binary Darboux transformations,
 {\em Lett. Math. Phys.}
 {\bf  106}, 1139--1179 (2016)

\end{thebibliography}
\end{document}